\documentclass[journal]{IEEEtran}
\usepackage{orcidlink}
\usepackage{cite}
\usepackage{epsfig}
\usepackage[caption=false,font=footnotesize]{subfig}
\usepackage{graphicx}
\usepackage{amssymb}
\usepackage{array}
\newcolumntype{C}[1]{>{\centering\let\newline\\\arraybackslash\hspace{0pt}}m{#1}}
\usepackage{pifont}
\usepackage[cmex10]{amsmath, nccmath}
\usepackage{mathtools}
\usepackage{algorithm}
\usepackage{algorithmic}
\usepackage{etoolbox}
\AtBeginEnvironment{algorithmic}{\renewcommand{\baselinestretch}{1}\selectfont}
\usepackage{upgreek}
\usepackage{graphicx}
\usepackage{enumitem}
\setlist{leftmargin=*}
\usepackage{bm}
\newtheorem{theorem}{Theorem}
\newtheorem{remark}{Remark}

\newtheorem{lemma}{Lemma}
\newtheorem{definition}{Definition}
\newtheorem{corollary}{Corollary}

\newcommand{\E}[1]{\mathbb E\left[#1\right]}
\newcommand{\Prob}[1]{\mathbb P\left[#1\right]}

\newcommand{\mc}[1]{\mathcal{#1}}
\newcommand{\mb}[1]{\mathbf{#1}}
\newcommand{\mr}[1]{\mathrm{#1}}
\newcommand{\ms}[1]{\mathsf{#1}}

\newcommand{\thmref}[1]{Theorem~\ref{#1}}

\newcommand{\defnref}[1]{Definition~\ref{#1}}
\newcommand{\secref}[1]{Section~\ref{#1}}
\newcommand{\lemref}[1]{Lemma~\ref{#1}}

\newcommand{\remref}[1]{Remark~\ref{#1}}
\newcommand{\corref}[1]{Corollary~\ref{#1}}
\newcommand{\appref}[1]{Appendix~\ref{#1}}
\newcommand{\algref}[1]{Algorithm~\ref{#1}}
\newcommand{\figref}[1]{Fig.~\ref{#1}}

\DeclareMathOperator*{\argmax}{arg\,max}

\usepackage{hyperref}
\usepackage{url}
\usepackage{cases}
\hypersetup{colorlinks=true, linkcolor=blue, citecolor=blue, urlcolor = blue}
\usepackage{ifthen}
\newif\ifshowtodo
\showtodofalse

\renewcommand{\IEEEproofindentspace}{1\parindent}

\newcommand{\VersionLength}{long}
\providecommand{\ver}{\ifthenelse{\equal{\VersionLength}{long}}}
\allowdisplaybreaks
\begin{document}
\title{Posterior Matching over Binary-Input Memoryless Symmetric Channels: Non-Asymptotic Bounds and Low-Complexity Encoding}
\author{Recep Can Yavas~\orcidlink{0000-0002-5640-515X},~\IEEEmembership{Member,~IEEE}
\thanks{R. C. Yavas is with the Department of Electrical and Electronics Engineering at Bilkent University, Ankara, Turkey (e-mail: ryavas@bilkent.edu.tr).}}
\IEEEoverridecommandlockouts
\maketitle 
\begin{abstract}
We study variable-length feedback (VLF) codes over binary-input memoryless symmetric (BMS) channels using posterior matching with small-enough-difference (SED) partitioning. Prior analyses of SED-based schemes rely on bounded log-likelihood ratio (LLR) increments, restricting their scope to discrete-output channels such as the binary symmetric channel (BSC). We remove this restriction and provide an analysis of posterior matching that covers a broad class of BMS channels, including continuous-output channels such as the binary-input AWGN channel. We derive a novel non-asymptotic achievability bound on the expected decoding time that decomposes into communication, confirmation, and recovery terms with explicit dependence on the channel capacity~$C$, the KL divergence~$C_1$, and the Bhattacharyya parameter of the channel. The proof develops new stopping-time and overshoot bounds for submartingales and random walks with unbounded increments, drawing on tools from renewal theory. On the algorithmic side, we propose a low-complexity encoder that enforces the exact SED partition at every step by grouping messages according to their log-likelihood ratios that are assumed to land on a lattice, and applying a batched correction step that restores the partition balance. The resulting encoder complexity is polynomial in the number of transmitted bits. For continuous-output channels, the lattice structure is enforced through output quantization satisfying an exact induced-lattice constraint; the associated capacity loss is $O(\log B / B^2)$ for a $B$-level quantizer. These results yield a VLF coding scheme for BMS channels that simultaneously achieves strong non-asymptotic performance and practical encoder complexity.
\end{abstract}
\section{Introduction}
Feedback does not increase the capacity of discrete memoryless channels (DMCs) \cite{shannon1956zero}.
It nevertheless simplifies capacity-achieving coding schemes in several settings including Horstein's scheme for the BSC and the Schalkwijk--Kailath scheme for the Gaussian channel \cite{horstein, schalkwijk}. It also improves finite-blocklength performance: Wagner \emph{et al.} \cite{wagner2020} show that even one bit of feedback improves the second-order achievable rate for DMCs that admit multiple capacity-achieving input distributions with distinct dispersions. The effect of feedback is even stronger in variable-length feedback (VLF) coding, where the transmission duration is a stopping time that depends on the channel outputs.
 
Burnashev~\cite{burnashev1976data} characterizes the optimal reliability function of VLF codes over a DMC. For every rate $R$ less than the capacity $C$, the optimal error exponent is given by
\begin{align}
    E(R) = \lim_{\epsilon \to 0} -\frac{1}{\E\tau } \log \epsilon  = C_1 \left(1 - \frac{R}{C} \right), \label{eq:burnashev}
\end{align}
where
\begin{align}
    C_1 = \max_{x, x' \in \mc{X}} D(P_{Y|X = x} \| P_{Y|X = x'}) \label{eq:C1intro}
\end{align}
is the Kullback--Leibler (KL) divergence between the output distributions induced by the two most distinguishable channel inputs,
$\tau$ is the random decoding time, and $\epsilon$ is the average error probability. For any $R < C$, Burnashev's error exponent \eqref{eq:burnashev} is larger than the corresponding exponent for fixed-length codes without feedback \cite{Gallager1965}. Significant gains are already possible with the most economical form of feedback, known as stop-feedback (or decision feedback): the encoder transmits a pre-determined codeword and the receiver sends back only a one-bit stop/continue signal. Forney~\cite{Forney1968} shows that stop-feedback achieves an error exponent strictly between those of Gallager and Burnashev. For a non-vanishing error probability $\epsilon \in (0, 1)$, Polyanskiy \emph{et al.} \cite{polyanskiy2011feedback} derive the following achievability and converse bounds on $\log M^*(N, \epsilon)$, the logarithm of the maximum achievable codebook size compatible with an average stopping time $\E{\tau} \leq N$ and an average error probability not exceeding $\epsilon$
\begin{IEEEeqnarray}{rCl}
    \frac{N C}{1-\epsilon} - \log N + O(1) &\leq& \log M^*(N, \epsilon) \IEEEnonumber \\
    &\leq& \frac{N C}{1-\epsilon} + \frac{h_{\mr{b}}(\epsilon)}{1-\epsilon},\IEEEeqnarraynumspace \label{eq:pol}
\end{IEEEeqnarray}
where $h_{\mr{b}}(\epsilon) \triangleq - \epsilon \log \epsilon - (1-\epsilon) \log(1-\epsilon)$ is the binary entropy function. Hence, the $\epsilon$-capacity is $\frac{C}{1-\epsilon}$, and the second-order term in the achievable $\log M$ is $O\left(\log N\right)$. The lower bound in \eqref{eq:pol} employs stop-feedback together with an information-density threshold rule.
 
Burnashev's achievability proof uses a two-phase scheme with full noiseless feedback: a communication phase that aims to increase the posterior of the transmitted message, and a confirmation phase that verifies or rejects the tentative decision. Yamamoto and Itoh \cite{YamamotoItoh} give a simpler and more modular construction that also achieves the optimal exponent. Their scheme alternates between fixed-length communication blocks, which can use any capacity-achieving code, and fixed-length confirmation blocks, in which the transmitter sends one of two control sequences formed from the most distinguishable input pair achieving \eqref{eq:C1intro}. The receiver performs a binary hypothesis test on the noisy control sequence and feeds its decision back to the transmitter; if the estimate is rejected, a new epoch begins. Berlin \emph{et al.} \cite{berlin} give a converse showing that this communication/confirmation structure is implicit in any scheme that achieves the optimal exponent. Chen \emph{et al.} \cite{chen2023} derive a non-asymptotic achievability bound for VLF codes with finitely many feedback instances using a Yamamoto--Itoh-type construction with optimized phase lengths. Yavas and Tan \cite{yavasTan2025} further improve the achievability bound in \eqref{eq:pol}: their refinement of the Yamamoto--Itoh scheme, in which both the communication and confirmation blocks have variable length with optimized average durations, achieves the second-order term $-\frac{C}{C_1} \log N$, where $C_1 > C$ for all non-trivial channels. Examples that extend the non-asymptotic VLF framework to finitely many feedback instances and to multi-terminal settings include \cite{kim2015VLF, williamson2015VLConvolution, vakilinia2016, yavas2021VLF, yavas2023VLF, Yang2022, truong2018Journal}. Extensions to the universal setting where the channel parameters are unknown appear in \cite{yavasTan2025, tchamkerten2006variable}.
 
Although the Yamamoto--Itoh scheme and its variants receive the full channel output at every step, they do not exploit this information within the communication phase: the encoder transmits a pre-determined codeword regardless of the outputs observed so far. Shayevitz and Feder \cite{shayevitz} formalize the idea of exploiting the full output at every time step through \emph{posterior matching}, placing the works of Horstein and Schalkwijk--Kailath under a unified framework. Consider a communication scheme with feedback and suppose that the receiver has observed the sequence $Y^t$ induced by the message $W$ up to time $t$. Shayevitz and Feder identify the ``missing information'' about $W$ as a random variable $U$ that is statistically independent of $Y^t$ and, together with $Y^t$, determines $W$ almost surely. They propose that the next transmitted symbol $X_{t+1}$ should be a fixed function of $U$ that follows a capacity-achieving input distribution $P_X^*$, so as to convey the most information on average. In \cite[Th.~III.1]{shayevitz}, they show that when $W \in [0, 1]$ is a continuous random variable, passing the CDF of the posterior through the inverse CDF of $P_X^*$ achieves these properties, i.e.,
\begin{align}
    X_{t+1} = F_X^{-1}(F_{W|Y^t}(W|Y^t)). \label{eq:shayevitz}
\end{align}
In the special case where the input is binary and $P_X^*$ is the uniform distribution, this reduces to
\begin{align}
    X_{t+1} = 1\{W \geq \mr{median}(F_{W|Y^t}(\cdot | Y^t))\}, \label{eq:Xtshay}
\end{align}
which partitions the message set $[0, 1]$ into two sets whose total posteriors are exactly $\frac{1}{2}$.
 
In the standard communication scenario with $M$ messages, partitioning the set $\{1, \dots, M\}$ into two sets with exactly equal total posterior is in general not possible.
For symmetric binary-input DMCs, Naghshvar \emph{et al.} \cite{naghshvar} introduce the \emph{small-enough-difference} (SED) partitioning rule for $\{1, \dots, M\}$ such that the total posteriors of the two bins are sufficiently close. Their construction attains Burnashev's optimal error exponent for symmetric binary-input channels and avoids explicit communication/confirmation switching. Naghshvar \emph{et al.} \cite{naghshvar2015EJS} also develop an extrinsic Jensen--Shannon divergence and submartingale framework for non-asymptotic analysis, and they extend the scope of the method beyond the original symmetric setting.
 
Yang \emph{et al.} \cite{yang2022SED} extend the SED framework to binary asymmetric channels (BACs)\footnote{BACs refer to arbitrary binary-input binary-output channels.}, including the BSC as a special case. Their generalized SED scheme uses a deterministic two-phase operation in the BAC setting: a communication phase based on an SED-type partition condition and a confirmation phase driven by the most likely message. They derive non-asymptotic achievability bounds for BACs, and for the BSC they give a refined analysis that combines submartingale arguments, surrogate processes, and a generalized Markov-chain first-passage analysis. This refined analysis yields a tighter average decoding time bound than that in \cite{naghshvar2015EJS}. Compared to the Yamamoto--Itoh-based schemes, SED-based schemes employ smoother back-and-forth transitions between communication and confirmation modes.
 
Antonini \emph{et al.} \cite{antonini2024} develop a low-complexity sequential posterior-matching framework for the BSC that substantially improves the practicality of the SED line. They introduce the \emph{small-enough-absolute-difference} (SEAD) partition criterion and derive a new analysis that yields stronger achievable-rate guarantees than \cite[Th.~7]{yang2022SED}. Their main algorithmic contribution is the thresholding-of-ordered-posteriors (TOP) algorithm, together with an initial systematic transmission stage. The resulting encoder organizes messages by type (i.e., by Hamming distance from the received systematic word), orders messages according to their posterior values, and applies a threshold partition without swap operations. The posteriors of the messages are updated in groups rather than individually. This structure reduces the encoder complexity to the order $O(\log^2 M)$ whereas earlier SED implementations from \cite{naghshvar2015EJS, yang2022SED} have complexity on the order $O(M \log M)$, which scales exponentially in the number of transmitted bits. The complexity $O(\log^2 M)$ makes their scheme practical at message sizes of interest, while preserving near-capacity finite-blocklength performance.

Despite the advances above, the existing analyses of SED-based posterior matching schemes, including \cite{naghshvar2015EJS, yang2022SED, antonini2024}, all require the log-likelihood ratio (LLR) increments to be almost surely bounded. This restricts their applicability to channels with discrete output alphabets. For the practically important class of continuous-output binary-input memoryless symmetric (BMS) channels, such as the binary-input additive white Gaussian noise (BI-AWGN) channel, the LLR increments are unbounded, and the standard bounded-increment arguments break down.
 
The only non-asymptotic VLF achievability bound applicable to continuous-output channels is the Yamamoto--Itoh-based result of Yavas and Tan \cite{yavasTan2025}. That scheme, however, has encoder complexity that scales exponentially in the number of transmitted bits. On the other hand, Antonini \emph{et al.}'s low-complexity framework \cite{antonini2024} achieves polynomial encoder complexity, but only for the BSC. No existing scheme provides both strong non-asymptotic guarantees and practical encoder complexity for a broad class of BMS channels, including BI-AWGN. The present paper fills this gap.

\subsection{Main Contributions}
Our contributions are as follows.
 
\begin{enumerate}[leftmargin=*, label=\arabic*)]
\item \emph{Non-asymptotic achievability bound for a broad class of BMS channels.}
We present the first analysis of SED-based posterior matching that applies to a broad class of BMS channels, including channels with continuous output alphabets. We derive a non-asymptotic upper bound on the expected decoding time (\thmref{thm:main}) that decomposes into three interpretable terms: a communication term with drift governed by the channel capacity~$C$, a confirmation term with drift governed by the KL divergence~$C_1$, and a recovery term controlled by the Bhattacharyya parameter~$Z(P_{Y|X})$. The analysis does not require bounded LLR increments. Instead, we develop new stopping-time and overshoot bounds for submartingales and random walks with potentially unbounded increments, drawing on tools from renewal theory. These tools, stated as standalone results in \secref{sec:supporting}, may be of independent interest.
 
\item \emph{Low-complexity encoder via lattice-TOP with SED repair.}
We propose a new encoder architecture that satisfies the \emph{exact} SED partition condition at every time step while retaining the efficiency of grouped-posterior operations. The encoder works whenever the channel LLR takes values on a finite lattice: messages sharing the same cumulative LLR label are grouped together, and the partition operates on groups rather than individual messages. A TOP initialization produces an SEAD-valid partition at low cost, followed by a batched repair procedure that restores the full SED condition by transferring messages from minimum-posterior groups. The per-round encoder cost at round $t$ is $O(B^2 t^2)$, or $O(B t \log(Bt))$ with a heap-based implementation, where $B$ is the number of lattice points on the LLR scale. Moreover, the stopping time $\tau$ is $O(K)$ with high probability, so the cumulative encoder cost up to decoding is $O(B^2 K^3)$, or $O(B K^2 \log(BK))$ with a heap, with high probability, where $K = \log_2 M$ is the number of transmitted bits. This is polynomial in $K$, in contrast to the $O(M \log M) = O(2^K K)$ complexity of prior exact-SED implementations \cite{naghshvar2015EJS, yang2022SED}. Because the encoder satisfies the exact SED condition, the non-asymptotic bound derived in contribution~1) applies directly.
 
\item \emph{Output quantization with exact induced-lattice constraint.}
For continuous-output BMS channels, we design a symmetric output quantizer that maps the channel output to a finite alphabet while enforcing the property that the induced LLR on each quantized symbol equals the corresponding lattice value $k\delta$ exactly.  This output quantization serves a dual purpose.  First, it addresses a practical constraint: a real-world feedback link has finite capacity, so the decoder can only convey finitely many bits per use; quantizing the output to $B$ levels requires only $\lceil \log_2 B \rceil$ feedback bits per channel use, in contrast to the infinite-precision feedback assumed by the continuous-output model. Second, and equally important for our scheme, the quantization creates the lattice-valued LLR structure that the low-complexity encoder in contribution~2) exploits. Thus, a single design choice, the $B$-level quantizer, simultaneously makes the feedback link practical and reduces the encoder complexity from exponential to polynomial.
 
We derive a capacity-loss expansion showing that the induced capacity satisfies $C - C_{B,\delta} = O(\delta^2)$, where $\delta$ is the lattice spacing, and the coefficient of $\delta^2$ is a channel-dependent constant. For the BI-AWGN channel, we show that the optimal lattice spacing yields a capacity loss of $O(\log B / B^2)$ (\corref{cor:awgn_theta}). We also provide a numerical design algorithm that minimizes the capacity-loss given the noise power $\sigma^2$ and the quantization level $B$.
\end{enumerate}
 
Table~\ref{tab:comparison} below summarizes the scope, complexity, and achievable-rate characteristics of the present work alongside the most closely related prior results. 
 
\begin{table*}[t]
\renewcommand{\baselinestretch}{1}\selectfont
\centering
\caption{Comparison of VLF coding schemes for binary-input channels. Here $K = \log_2 M$ is the number of transmitted bits and $B$ is the number of quantization levels. The second- and third-order terms refer to the achievable expansion $\log M^*(N, \epsilon) \geq \frac{NC}{1-\epsilon} + \text{(second order)} + \text{(third order)} + o(1)$ as $N \to \infty$ with $\epsilon \in (0, 1/2)$ fixed.}
\label{tab:comparison}
\renewcommand{\arraystretch}{1.4}
\begin{tabular}{l c c c c c c c}
\hline\hline
 & \textbf{Channel} & \textbf{Bounded LLR} & \textbf{Coding} & \textbf{Encoder} & \textbf{Second-order} & \textbf{Third-order} & \textbf{Burnashev} \\[-2pt]
\textbf{Reference} & \textbf{class} & \textbf{required} & \textbf{scheme} & \textbf{complexity} & \textbf{term} & \textbf{term} & \textbf{exponent} \\
\hline
Polyanskiy \emph{et al.}~\cite{polyanskiy2011feedback} & DMC & Yes & stop-feedback & exponential in $K$ & $- \log N$ & --- & No \\
Naghshvar \emph{et al.}~\cite{naghshvar2015EJS} & Discrete BMS & Yes & SED & exponential in $K$ & $ -\left(1 + \frac{C}{C_1} \right)\log N $ & --- & Yes \\[2pt]
Yang \emph{et al.}~\cite{yang2022SED} & BAC & Yes & SED & exponential in $K$ & $-\frac{C}{C_1}\log N$ & $O(1)$ & Yes \\[2pt]
Antonini \emph{et al.}~\cite{antonini2024} & BSC & Yes & SEAD & $O(K^2)$ & $-\frac{C}{C_1}\log N$ & $O(1)$ & Yes \\[2pt]
Yavas--Tan~\cite{yavasTan2025} & DMC$^{\,a}$ & No & Yam.-Itoh & exponential in $K$ & $-\frac{C}{C_1}\log N$ & $-\log\log N$ & Yes \\[2pt]
\hline
\textbf{This paper} & \textbf{BMS}$^{\,a}$ & \textbf{No} & \textbf{SED} & $\bm{O(B K^2 \log(BK))}$ & $\bm{-\frac{C}{C_1}\log N}$ & $\bm{O(1)}$ & \textbf{Yes} \\
\hline\hline
\vspace{1pt}
\end{tabular}
\raggedright
\footnotesize
$^{a}$\,The theorem follows for binary-input continuous-output channels such as BI-AWGN as well. \quad
\end{table*}

\subsection{Paper Organization}
\secref{sec:notation} introduces the notation and defines BMS channels. \secref{sec:problemform} formulates the VLF coding problem and reviews the SED partitioning scheme and its variants. \secref{sec:mainresult} states our main non-asymptotic achievability bound. \secref{sec:analysis} proves this bound, beginning with supporting results on stopping times and ruin probabilities and then analyzing the communication, confirmation, and recovery phases separately. \secref{sec:algorithm} describes the low-complexity encoder based on lattice-LLR grouping and SED repair. \secref{sec:quantization} develops the output quantization framework and its specialization to the BI-AWGN channel. Supporting proofs are collected in the appendices. \secref{sec:extension} discusses extensions of the paper and summarizes our results. 

\section{Notation and Definitions} \label{sec:notation}
For $n \in \mathbb{N}$, we denote $[n] \triangleq \{1, \dots, n\}$ and the length-$n$ vector $x^n \triangleq (x_1, \dots, x_n)$. We write $\mathbb{N}_0 \triangleq \{0, 1, 2, \dots\}$ for the set of non-negative integers. The minimum of numbers $a$ and $b$ is denoted by $a \wedge b$. For a random variable $X$ on alphabet $\mc{X}$, its distribution is denoted by $P_X$. We use $P_X$ to denote the probability mass function (PMF) when $X$ is discrete and the probability density function (PDF) when $X$ is continuous. For identically distributed random variables $X$ and $Y$, we write $X \overset{d}{=} Y$. The sign of $X$ is denoted by $\mr{sgn}(X)$. The KL divergence between $P$ and $Q$ is denoted by $D(P \| Q)$.

For any random variable $X$, define $X^+ \triangleq \max\{0, X\}$ and $X^- \triangleq -\min\{0, X\}$. The \emph{positive} and \emph{negative mean-excess overshoot parameters} are
\begin{align}
    \eta_0^{+}(P_X) &\triangleq \sup_{x \geq 0} \E{X - x \mid X \geq x}, \label{eq:eta0pos} \\
    \eta_0^{-}(P_X) &\triangleq \sup_{x \geq 0} \E{-(X + x) \mid X < -x}, \label{eq:eta0negdef}
\end{align}
defined when $\Prob{X \geq 0} > 0$ and $\Prob{X < 0} > 0$, respectively.
Given a random variable $X$ with mean $\E{X} > 0$, we define the \emph{overshoot parameter}
\begin{align}
    \eta(P_X) 
    &\triangleq \min\left \{ \frac{\E{(X^+)^2}}{\E{X}}, \frac{3 \E{|X|^3}}{\E{X^2}}, \eta_0^+(P_X)\right\}. \label{eq:etadef}
\end{align}
The first term in the minimum in \eqref{eq:etadef} is Lorden's non-asymptotic bound \cite{lorden1970}, the second is Mogul'skii's bound \cite{mogul1974}, and the third arises from \thmref{thm:hitting_time_drift_overshoot} below. The term $\eta(P_X)$ is the best of three available universal upper bounds on the expected overshoot $\E{S_\tau - a}$ of the random walk $S_n = \sum_{i=1}^n X_i$ past level $a > 0$.

All logarithms have base $e$. The standard normal PDF and cumulative distribution function (CDF) are denoted by $\phi(\cdot)$ and $\Phi(\cdot)$, respectively. 

Given a random variable $X$ with distribution $P_X$ and the parameter $\lambda > 0$ given by the solution of $\E{e^{-\lambda X}} = 1$, we define the \emph{exponentially tilted} random variable $\bar{X}$ via its distribution
\begin{align}
    \E{g(\bar{X})} = \E{g(X) e^{-\lambda X}} \label{eq:bartilt}
\end{align}
for any measurable function $g$ for which the right-hand side is well-defined. Equivalently, $P_{\bar{X}}$ has Radon--Nikodym derivative $\frac{dP_{\bar{X}}}{dP_X}(x) = e^{-\lambda x}$. If $\E{X} > 0$ and $\Prob{X < 0} > 0$, then $\E{\bar{X}} < 0$, so $-\bar{X}$ has positive mean. Hence, the overshoot parameter $\eta(P_{-\bar{X}})$ is well-defined via \eqref{eq:etadef} applied to $-\bar{X}$.

Since our goal is to generalize posterior matching schemes to arbitrary BMS channels, we fix the input alphabet as 
\begin{align}
    \mathcal{X} = \{-, +\}. \label{eq:inputalph}
\end{align}
For additive channels such as the Gaussian channel, we identify $-$ and $+$ as $-1$ and $+1$, respectively. The output alphabet $\mathcal{Y} \subseteq \mathbb{R}$ may be uncountable as in the BI-AWGN channel.

To simplify notation, the conditional distributions of a binary-input channel $P_{Y|X}$ are denoted by
\begin{align}
    P_{-} \triangleq P_{Y|X=-}, \quad \quad P_{+} \triangleq P_{Y|X = +}. \label{eq:Pplusminus}
\end{align}
The Bhattacharyya parameter, the \emph{ruin parameter}, and the LLR random variable are
\begin{align}
    Z(P_{Y|X}) &\triangleq \int \sqrt{P_+(y)\, P_-(y)}\, dy, \label{eq:Zbhatt} \\
    \chi(P_{Y|X}) &\triangleq 1 - \sqrt{1 - Z(P_{Y|X})}, \label{eq:chidef} \\
    \Lambda &\triangleq \log \frac{P_{+}(Y)}{P_-(Y)}, \quad Y \sim P_+. \label{eq:Lambdadef}
\end{align}

\begin{definition}[BMS channel] \label{def:BMS}
    A channel $(\mathcal{X}, \mathcal{Y}, P_{Y|X})$ is called BMS if 
    \begin{enumerate}
        \item $\mathcal{X} = \{-, +\}$,
        \item there exists an involutive permutation $\sigma$ of the output alphabet $\mathcal{Y} \subseteq \mathbb{R}$ (i.e., $\sigma$ is equal to its inverse) such that $P_{+}(y) = P_{-}(\sigma(y))$,
        \item the $n$-th extension of the channel is memoryless: $P_{Y^n|X^n}(y^n|x^n) = \prod_{i =1}^n P_{x_i}(y_i)$.
    \end{enumerate}
    To simplify notation, we relabel the output symbols so that
    \begin{align}
        P_{+}(y) &= P_{-}(-y), \quad \forall\, y \in \mathcal{Y}, \label{eq:PYXcond1} \\
        P_{+}(y) &\geq P_{-}(y), \quad\quad \text{if } y \geq 0. \label{eq:PYXcond2}
    \end{align}
    In other words, under the uniform input distribution, the output $Y$ is more likely to be generated by $X = \mr{sgn}(Y)$.
\end{definition}

Examples of BMS channels include the BSC, the binary erasure channel (BEC)\footnote{For the BSC and BEC, we relabel the input symbols 0 and 1 as $-$ and $+$, respectively; the output symbol 0 is relabeled as $-1$. For the BEC, the erasure symbol $\mathsf{e}$ is relabeled as 0.}, and additive channels of the form $Y = X + Z$, where $Z$ is independent of $X$ and $P_Z(z)$ is a decreasing function of $|z|$. In particular, the BI-AWGN channel satisfies this with $Z \sim \mathcal{N}(0, \sigma^2)$, i.e.,
\begin{align}
    P_Z(z) = \frac{1}{\sqrt{2 \pi \sigma^2}} e^{-\frac{z^2}{2 \sigma^2}}, \quad z \in \mathbb{R}.
\end{align}

By symmetry, the unique capacity-achieving input distribution for BMS channels is $P_X = \textrm{Uniform}(\{-, +\})$, and the corresponding output distribution is the mixture $P_{Y} = \frac{1}{2} P_{-} + \frac{1}{2} P_{+}$, which is symmetric around $0$.

The capacity and the KL divergence between the conditional output distributions are
\begin{align}
    C &= D(P_{+} \| P_{Y}) = D(P_{-} \| P_{Y}), \label{eq:capBMS} \\
    C_1 &= D(P_{+} \| P_{-}) = D(P_{-} \| P_{+}) = \E{\Lambda}, \label{eq:C1BMS}
\end{align}
with $C_1 > C$ for all non-trivial BMS channels.

\begin{remark}[BMS symmetry of the tilted LLR] \label{rem:tiltedLLR}
For a BMS channel, the LLR $\Lambda$ satisfies $\E{e^{-\Lambda}} = \int P_+(y) \frac{P_-(y)}{P_+(y)} \, dy = 1$, so $\lambda = 1$. The tilted random variable $\bar{\Lambda}$ satisfies, for any measurable $g$,
\begin{align}
    \E{g(\bar{\Lambda})} &= \E{g(\Lambda) e^{-\Lambda}}  \\
    &= \int g\!\left(\log \frac{P_+(y)}{P_-(y)}\right) P_-(y) \, dy. \label{eq:tiltedLLR}
\end{align}
By the BMS symmetry $P_+(y) = P_-(-y)$, the substitution $y \to -y$ in \eqref{eq:tiltedLLR} gives $\E{g(\bar{\Lambda})} = \int g(-\log \frac{P_+(u)}{P_-(u)}) P_+(u)\,du = \E{g(-\Lambda)}$, so $\bar{\Lambda} \overset{d}{=} -\Lambda$. Therefore,
\begin{align}
    \eta(P_{-\bar{\Lambda}}) = \eta(P_\Lambda). \label{eq:etatilteq}
\end{align}
\end{remark}

We employ the standard $o(\cdot)$, $O(\cdot)$, $\Omega(\cdot)$, and $\Theta(\cdot)$ notations for asymptotic relationships of functions.

\section{Problem Formulation and Preliminaries} \label{sec:problemform}
Following \cite{polyanskiy2011feedback, yavas2021VLF}, we define VLF codes as follows. 
\begin{definition}[VLF code {\cite[Def.~1]{polyanskiy2011feedback}}]
    \label{def:VLF}
    Fix $\epsilon \in (0, 1)$, $N > 0$, and a positive integer $M$. An $(N, M, \epsilon)$-VLF code consists of 
    \begin{enumerate}
        \item a common randomness random variable $V$ with $|\mc{V}| \leq 2$ and distribution $P_V$,
        \item encoding functions $\mathsf{f}_t \colon \mc{V} \times [M] \times \mc{Y}^{t-1} \to \mc{X}$ such that
        \begin{align}
            X_t(W) = \ms{f}_t(V, W, Y^{t-1}) \quad \forall t \in \mathbb{N}_0, \label{eq:encoder}
        \end{align}
        where $W$ is the equiprobable message on $[M]$,
        \item a stopping time $\tau \in \mathbb{N}_0$ with respect to the filtration generated by $\{V, Y^{t}\}_{t \geq 0}$, meeting the average decoding time constraint
        \begin{align}
            \E{\tau} \leq N, \label{eq:avtime}
        \end{align}
        \item a decoding function $\ms{g}_{\tau} \colon \mc{V} \times \mc{Y}^\tau \to [M]$ giving the estimate
        \begin{align}
            \hat{W} = \ms{g}_{\tau}(V, Y^\tau) \label{eq:decoder}
        \end{align}
        that satisfies the average error probability constraint
        \begin{align}
        \Prob{\hat{W} \neq W} \leq \epsilon. \label{eq:errorconst}
        \end{align}
    \end{enumerate}
\end{definition}

We define the minimum achievable expected stopping time and the maximum achievable codebook size as
\begin{align}
    N^*(M, \epsilon) &\triangleq \min\{ N \geq 0 \colon \exists \, (N, M, \epsilon)\text{-VLF code}\}, \label{eq:Nstar} \\
    M^*(N, \epsilon) &\triangleq \max\{M \in \mathbb{N} \colon \exists \, (N, M, \epsilon)\text{-VLF code}\}. \label{eq:Mstar}
\end{align}

\subsection{SED Scheme and Its Variants} \label{sec:YangSED}
In this section, we review the SED coding scheme introduced by Naghshvar \emph{et al.} \cite{naghshvar2015EJS}, together with the variants by Yang \emph{et al.} \cite{yang2022SED} and Antonini \emph{et al.} \cite{antonini2024}. These are deterministic VLF schemes.\footnote{Naghshvar et al.'s scheme is designed for binary-input DMCs; Yang et al.'s scheme is designed for binary-input binary-output channels; and Antonini et al.'s scheme is designed for the BSC.}

Let $\rho_m(t)$ be the posterior probability of the message $m \in [M]$ at time $t \in \mathbb{N}_0$ upon observing the output sequence $Y^t$
\begin{align}
    \rho_m(t) = \Prob{W = m | Y^t}. \label{eq:rhom}
\end{align}
Since $W$ is uniform on $[M]$, $\rho_m(0) = \frac{1}{M}$ for all $m$. The posterior vector $\boldsymbol{\rho}(t) = (\rho_1(t), \dots, \rho_M(t))$ is a sufficient statistic for estimating $W$ at time $t$ \cite{naghshvar2015EJS, burnashev1976data}.
By Bayes' rule, after observing the output symbol $Y_t = y_t$, the posteriors are updated as 
\begin{align}
    \rho_m(t) = \frac{\rho_m(t-1) P_{x_t(m)}(y_t)}{\sum \limits_{j \in [M]} \rho_j(t-1) P_{x_t(j)}(y_t)} \quad m \in [M], \label{eq:bayesupdate}
\end{align}
where $x_{t}(j) = \mathsf{f}_t(v, j, y^{t-1})$ is the input symbol for message $j \in [M]$ at time $t \in \mathbb{N}_0$. Through noiseless feedback, the transmitter knows $y^{t-1}$ and therefore knows $\boldsymbol{\rho}(t-1)$.

Given a target error probability $\epsilon \in (0, \frac{1}{2})$, the SED scheme operates as follows.
\begin{enumerate}
    \item \textit{Partitioning of messages:} At time $t \geq 1$, we partition $[M]$ into two bins $S_{-}(t-1)$ and $S_+(t-1)$. We use Yang \emph{et al.}'s variant \cite[eq.~(43)]{yang2022SED}, which requires
    \begin{align}
        - \min\limits_{j \in S_{+}(t-1)} \rho_j(t-1) &\leq \pi_{-}(t-1) - \pi_{+}(t-1) \\
        &\leq \min\limits_{j \in S_{-}(t-1)} \rho_j(t-1), \label{eq:SED_Yang}
    \end{align}
where 
\begin{align}
    \pi_{-}(t-1) &\triangleq \sum_{j \in S_{-}(t-1)} \rho_j(t-1), \label{eq:piminus} \\
    \pi_{+}(t-1) &\triangleq \sum_{j \in S_{+}(t-1)} \rho_j(t-1) \label{eq:piplus}
\end{align}
are the total posteriors of the bins. This condition \eqref{eq:SED_Yang} ensures that if $\rho_m(t-1) > \frac{1}{2}$ for some message $m$, then either $S_{-}(t-1) = \{m\}$ or $S_+(t-1) = \{m\}$. 

Naghshvar \emph{et al.}'s SED variant in \cite{naghshvar2015EJS} requires
\begin{align}
    0 \leq \pi_{-}(t-1) - \pi_{+}(t-1) \leq \min\limits_{j \in S_{-}(t-1)} \rho_j(t-1), \label{eq:nagh}
\end{align}
which is stricter than \eqref{eq:SED_Yang}. Antonini \emph{et al.}'s SEAD variant in \cite{antonini2024} requires
\begin{align}
     &\left \lvert \pi_{-}(t-1) - \pi_{+}(t-1) \right \rvert \leq \min\limits_{j \in S_{-}(t-1)} \rho_j(t-1), \label{eq:SED_antonini1} \\
     &\rho_j(t-1) \geq \frac{1}{2} \Longrightarrow S_{-}(t-1) = \{j\} \text{ or } S_+(t-1) = \{j\}, \label{eq:SED_antonini2}
\end{align}
which is more relaxed than \eqref{eq:SED_Yang}. The SED partitioning in \eqref{eq:SED_Yang} leads to a uniform upper bound on the conditional expected stopping time $\E{\tau | W = m}$ while the SEAD partitioning in \eqref{eq:SED_antonini1}--\eqref{eq:SED_antonini2} does not directly provide such a bound.
\item \textit{Encoding function:} At time $t$, the transmitter sends the sign of the bin containing $W$:
\begin{align}
    X_t(W) = \begin{cases}
        + &\text{if } W \in S_{+}(t-1) \\
        - &\text{if } W \in S_{-}(t-1).
    \end{cases} \label{eq:Xt}
\end{align} 
\item \textit{Stopping time:} The decoder stops at
\begin{align}
    \tau \triangleq \inf\left\{t \geq 0 \colon \max_{m \in [M]} \rho_m(t) \geq 1 - \epsilon\right\}, \label{eq:tau}
\end{align}
which is a stopping time with respect to the filtration $\mc{F}_t = \sigma(V, Y^t)$.
\item \textit{Decoding function:} At time $\tau$, the decoder outputs
\begin{align}
    \hat{W} = \argmax_{m \in [M]} \rho_m(\tau). \label{eq:hatW}
\end{align}
\end{enumerate}
As long as the encoder and decoder use the same deterministic algorithm satisfying the SED partitioning rule, the scheme is deterministic and the common randomness $V$ can be dropped. Also, because $\epsilon < \frac{1}{2}$, the estimate $\hat{W}$ in \eqref{eq:hatW} is unique.

\section{Main Result} \label{sec:mainresult}
The following theorem is the main result of the paper, giving a non-asymptotic achievability bound for VLF codes over BMS channels with $C_1 < \infty$ and
$\eta(P_{\Lambda}) < \infty$. For discrete-output BMS channels, these conditions reduce to mutual absolute continuity of $P_+$ and $P_-$ (equivalently, every output symbol is accessible from both inputs). The BI-AWGN channel satisfies the conditions for any $\sigma^2 < \infty$.
\begin{theorem}[Achievability bound for BMS channels] \label{thm:main}
    Fix a positive integer $M$ and error probability $\epsilon \in (0, \frac{1}{2})$. Fix a BMS channel $P_{Y|X}$ with capacity $C > 0$, KL divergence $C_1 = D(P_+ \| P_-) < \infty$\footnote{An important example where $C_1 = \infty$ is the BEC. In that case, a simple coding scheme that retransmits each information bit until they are not erased achieves $\log M^*(N, \epsilon) \geq \left\lfloor \frac{NC}{1-\epsilon} \right\rfloor$. See \cite[Th.~7]{polyanskiy2011feedback}.}, ruin parameter $\chi(P_{Y|X}) \in (0, 1)$, and LLR random variable $\Lambda$ defined in \eqref{eq:Lambdadef} satisfying $\eta(P_{\Lambda}) < \infty$.
    Define 
    \begin{align}
        &N_0(M, \epsilon') \triangleq \frac{\log (M-1) + \log 2}{C} + \frac{\log \frac{1-\epsilon'}{\epsilon'} + \eta(P_{\Lambda})}{C_1} \notag \\
        & + \frac{\min\!\big\{\alpha(\epsilon') \cdot \eta(P_{\Lambda}),\, \chi(P_{Y|X}) \cdot \eta_0^{-}(P_{\Lambda})\big\} + \chi(P_{Y|X}) \log 2}{(1-\chi(P_{Y|X})) \, C}, \label{eq:N0def}
    \end{align}
    where $\eta_0^-(P_{\Lambda})$ and $\eta(P_{\Lambda})$ are given in \eqref{eq:eta0negdef}--\eqref{eq:etadef} evaluated for the distribution of $\Lambda$, and
    \begin{align}
         &\alpha(\epsilon') \triangleq \max\!\left\{\frac{\chi(P_{Y|X})}{1 - e^{- \frac{1}{2}\log \frac{1-\epsilon'}{\epsilon'}}},\; e^{- \frac{1}{2}\log \frac{1-\epsilon'}{\epsilon'}}\right\}. \label{eq:Reps_halfsplit}
    \end{align}
    Then,
    \begin{align}
        N^*(M, \epsilon) &\leq \min_{\epsilon_0 \in [0, \epsilon)} \, (1-\epsilon_0) \, N_0\!\left(M, \frac{\epsilon - \epsilon_0}{1-\epsilon_0}\right). \label{eq:Nach}
    \end{align}
    \end{theorem}
    The bound $N_0(M, \epsilon')$ is attained by Yang \emph{et al.}'s SED rule described in \secref{sec:YangSED}, and the improvement from the minimization over $\epsilon_0$ is achieved via time-sharing with a stop-at-time-zero strategy \cite{polyanskiy2011feedback, yavas2023VLF, yavasTan2025}.

The three terms in $N_0(M, \epsilon')$ have natural interpretations corresponding to three phases of the SED coding scheme.
\begin{enumerate}
\item \emph{Communication phase.} The first term, $\frac{\log(M-1) + \log 2}{C}$, captures the initial phase during which the posterior of the true message rises from its initial value $\frac{1}{M}$ toward $\frac{1}{2}$, equivalently toward the non-negative region of the log-posterior-odds scale. The drift of the underlying submartingale in this phase is at least the channel capacity~$C$. 
\item \emph{Confirmation phase.} The second term, $\frac{\log \frac{1-\epsilon'}{\epsilon'} + \eta(P_\Lambda)}{C_1}$, corresponds to the phase during which the log-posterior odds of the true message climbs from zero to the decoding threshold $\log \frac{1-\epsilon'}{\epsilon'}$. In this phase, the true message occupies a singleton bin, so the increments are i.i.d.\ log-likelihood ratios with drift $C_1$. 
\item \emph{Recovery phase.} The third term accounts for fallbacks: each time the confirmation-phase random walk drops below zero, the process must recover to the non-negative region before confirmation can resume. The factor $\frac{1}{1-\chi(P_{Y|X})}$ arises from the geometric sum over the number of such recovery events, whose probability is bounded by the ruin parameter $\chi(P_{Y|X})$. 
\end{enumerate}

The minimization over $\epsilon_0$ in \eqref{eq:Nach} is a time-sharing improvement between the SED scheme and a dummy scheme that decodes at time~0; the construction is detailed in the proof of \thmref{thm:main} in \secref{sec:proofmain}.

\begin{corollary}[Asymptotic expansion (finite-$\epsilon$ regime)] \label{cor:asymptotic}
    Under the same setup as \thmref{thm:main}, for any $\epsilon \in (0, \frac{1}{2})$,
    \begin{align}
        \log M^*(N, \epsilon) &\geq \frac{NC}{1-\epsilon} - \frac{C}{C_1} \log N + K_\epsilon + o(1) \label{eq:asymp}
    \end{align}
    as $N \to \infty$, where
    \begin{align}
        K_\epsilon &\triangleq -\frac{C}{C_1}\!\left(1 + \log \frac{C_1}{1-\epsilon}\right) - \frac{C}{C_1}\,\eta(P_\Lambda) \notag \\
        &\quad - \frac{\chi(P_{Y|X}) \cdot \min\!\big\{\eta(P_\Lambda),\, \eta_0^{-}(P_\Lambda)\big\}}{1-\chi(P_{Y|X})} - \frac{\log 2}{1-\chi(P_{Y|X})}. \label{eq:Keps}
    \end{align}
    The optimal stop-at-time-zero parameter satisfies
    \begin{align}
        \epsilon_0^* = \epsilon - \frac{(1-\epsilon)^2}{NC_1} + O(N^{-2}). \label{eq:eps0star}
    \end{align}
\end{corollary}
\begin{IEEEproof}
    See \appref{app:cor}.
\end{IEEEproof}

From \thmref{thm:main}, substituting $M = \lfloor e^{NR}\rfloor$ and letting $\epsilon \to 0$ in \eqref{eq:N0def} (with $\alpha(\epsilon') \leq 1$ from \lemref{lem:exit_ratio}(i)) also recovers Burnashev's optimal error exponent $E(R) = C_1(1 - R/C)$, as in \cite{naghshvar2015EJS, yang2022SED, antonini2024, yavasTan2025}.

\begin{remark} \label{rem:log2}
    For the BSC with the crossover probability $p \in (0, \frac{1}{2})$, the term $\log 2$ is improved to $\frac{\log (2(1-p))}{1-p}$. For the BI-AWGN channel with noise power $\sigma^2$, it is improved to 
    \begin{align}
        (1 + r) \log \frac{2}{1+r},
    \end{align}
    where $r = \frac{\Phi(-\frac{1}{\sigma})}{\Phi(\frac{1}{\sigma})}$. For its proof, see \appref{app:log2}.

    More generally, the bound $N_0(M, \epsilon')$ in \thmref{thm:main} uses $\chi(P_{Y|X})$ as a surrogate for the ruin probability $\psi(0)$. When a tighter value of $\psi(0)$ is available, it may be substituted for $\chi(P_{Y|X})$ throughout \eqref{eq:N0def}, since $\psi(0) \leq \chi(P_{Y|X})$ is the only property of $\chi$ used in the proof. For the BSC, the classical gambler's ruin \cite[Ch.~XIV]{feller1968} gives $\psi(0) = p/(1-p)$, which is substantially smaller than $\chi(P_{Y|X})$ for small $p$. For the BI-AWGN channel, $\psi(0)$ can be numerically evaluated via Spitzer’s formula (\thmref{thm:spitzer-lundberg-sandwich}(i) below), since $\mathbb{P}_0[S_n < 0] = \Phi(-\sqrt{n}/\sigma)$ and the series converges rapidly.
\end{remark}

\begin{remark}[Parameters for the BSC and BI-AWGN]
For the BSC with crossover probability $p\in(0,\tfrac12)$,
\begin{align}
C &= \log 2-h_{\mathrm b}(p), \\
C_1 &= (1-2p)\log\frac{1-p}{p}, \notag \\
Z(P_{Y|X}) &= 2\sqrt{p(1-p)}, \notag \\
\eta(P_\Lambda) &= \log\frac{1-p}{p}, \\
\chi(P_{Y|X}) &= 1 - \sqrt{1 - 2\sqrt{p(1-p)}}, \notag \\
\eta_0^{-}(P_\Lambda) &= \log\frac{1-p}{p}, \notag
\end{align}
where $h_{\mathrm b}(p)\triangleq -p\log p-(1-p)\log(1-p)$.

For the BI-AWGN channel $Y=X+Z$ with $X\in\{-1,+1\}$ and $Z\sim\mathcal{N}(0,\sigma^2)$, with $G\sim\mathcal{N}(0,1)$ and $\Lambda=\frac{2}{\sigma^2}+\frac{2}{\sigma}G$,
\begin{align}
C &= \log 2-\E{\log\!\left(1+e^{-\Lambda}\right)}, \\
C_1 &= \frac{2}{\sigma^2}, \notag \\
Z(P_{Y|X}) &= \exp\!\left(-\frac{1}{2\sigma^2}\right), \notag \\
\eta(P_\Lambda) &\leq \frac{2}{\sigma^2}+\frac{2}{\sigma}\,\frac{\phi(1/\sigma)}{\Phi(1/\sigma)},\footnotemark \\
\chi(P_{Y|X}) &= 1 - \sqrt{1 - e^{-1/(2\sigma^2)}}, \notag \\ 
\eta_0^{-}(P_\Lambda) &= \frac{2}{\sigma}\,\frac{\phi(1/\sigma)}{\Phi(-1/\sigma)} - \frac{2}{\sigma^2}. \label{eq:gausseta}
\end{align}
\footnotetext{Numerical evaluation shows that $\eta(P_\Lambda)$ equals the mean-excess bound $\eta_0^+(P_\Lambda)$ for most values of $\sigma^2$ of interest, including $\sigma^2 = 1$.}
\end{remark}

\subsection{Comparison with the Literature}
In \cite[Remark~7]{naghshvar2015EJS}, for general DMCs, Naghshvar \emph{et al.} derive the achievability bound
\begin{align}
    N^*(M, \epsilon) &\leq \frac{\log M + \log \log \frac{M}{\epsilon}}{C} + \frac{\log \frac{1}{\epsilon} + 1}{C_1} + \frac{96 e^{2C_2}}{C C_1}, \label{eq:nagheq}
\end{align}
where 
\begin{align}
    C_2 = \max\limits_{y \in \mc{Y}} \log \frac{\max \limits_{x \in \mc{X}} P_{Y|X}(y|x)}{\min \limits_{x \in \mc{X}} P_{Y|X}(y|x)}.
\end{align}
While this result achieves Burnashev's exponent, the constant (third) term in \eqref{eq:nagheq} is non-negligible at practical code sizes (e.g., for the BSC(0.11), the third term exceeds $11 \times 10^3$ channel uses). Moreover, the higher-order rate is worse than that of \eqref{eq:N0def} due to the $\frac{\log \log \frac{M}{\epsilon}}{C}$ term.

In \cite[Th.~3]{antonini2024}, for the BSC with crossover probability $p$, Antonini \emph{et al.} derive the achievability bound
\begin{align}
    N^*(M, \epsilon) &\leq \frac{\log(M-1) + \frac{\log(2 (1-p))}{1-p}}{C}  + \frac{C_2}{C_1} \left \lceil \frac{\log \frac{1-\epsilon}{\epsilon}}{C_2} \right \rceil \notag \\
    &\quad + e^{-C_2} \frac{1 - \frac{\epsilon}{1-\epsilon} e^{-C_2}}{1 - e^{-C_2}}  \cdot \left( \frac{\log (2 (1-p))}{(1-p) C} - \frac{C_2}{C_1}\right).
\end{align}
 
Yavas and Tan's asymptotic achievability bound in \cite[Th.~2]{yavasTan2025} yields for an arbitrary DMC with $C < \infty$ and $C_1 > 0$
\begin{align}
    \log M^*(N, \epsilon) \geq \frac{NC}{1-\epsilon} - \frac{C}{C_1} \log N - \log \log N + O(1). \label{eq:yavasTanbound}
\end{align}
Our bound in \corref{cor:asymptotic} matches \eqref{eq:yavasTanbound} in the first- and second-order terms and improves the third-order term from $-\log \log N$ to $O(1)$.
 
\figref{fig:bsc_sim} compares the non-asymptotic bounds for the BSC$(0.11)$ with $\epsilon = 10^{-3}$.  Monte Carlo simulation of the SED scheme ($50{,}000$ trials per message size, $\log_2 M$ from $20$ to $120$) is shown alongside the bounds of Yang \emph{et al.}~\cite[Th.~7]{yang2022SED}, Antonini \emph{et al.}~\cite[Th.~1]{antonini2024}, \thmref{thm:main}, Yavas--Tan~\cite[Th.~1]{yavasTan2025}, and Polyanskiy \emph{et al.}'s stop-feedback bound \cite[eq.~(102)]{polyanskiy2011feedback}. The empirical error rate in the simulations is below $10^{-3}$ at all instants. For \thmref{thm:main}, we substitute the exact ruin probability $\psi(0) = p/(1-p)$ for $\chi(P_{Y|X})$ (cf.\ Remark~\ref{rem:log2}). The bound of~\cite{antonini2024} is tighter for the BSC but does not extend to continuous-output channels. 
 
\figref{fig:awgn_sim} shows the analogous comparison for the BI-AWGN channel with $\sigma^2 = 1$ and $\epsilon = 10^{-3}$, using the lattice quantizer with $B = 31$ output levels designed in \secref{sec:awgn_numeric}. Two instances of \thmref{thm:main} are shown: one with the unquantized channel parameters, and one with the lattice channel's induced $C$ and~$C_1$; both use $\psi(0)$ from Spitzer's formula in place of~$\chi$.

\figref{fig:awgn_varyB} shows the achievable rate as a function of the number of output levels~$B$ for $\log_2 M = 40$ and $\epsilon = 10^{-3}$, using the optimized lattice channel for each $B$. For the range of message sizes and error probabilities in the figures, the optimal stop-at-time-zero parameter satisfies $\epsilon_0^* = 0$.
 
\begin{figure}[t]
    \centering
    \includegraphics[width=\columnwidth]{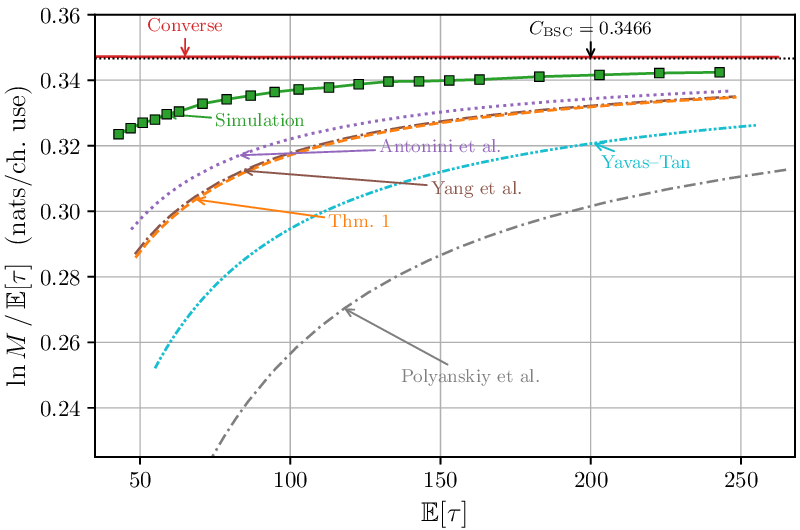}
    \caption{Rate vs.\ expected decoding time, BSC$(0.11)$, $\epsilon = 10^{-3}$.}
    \label{fig:bsc_sim}
\end{figure}

\begin{figure}[t]
    \centering
    \includegraphics[width=\columnwidth]{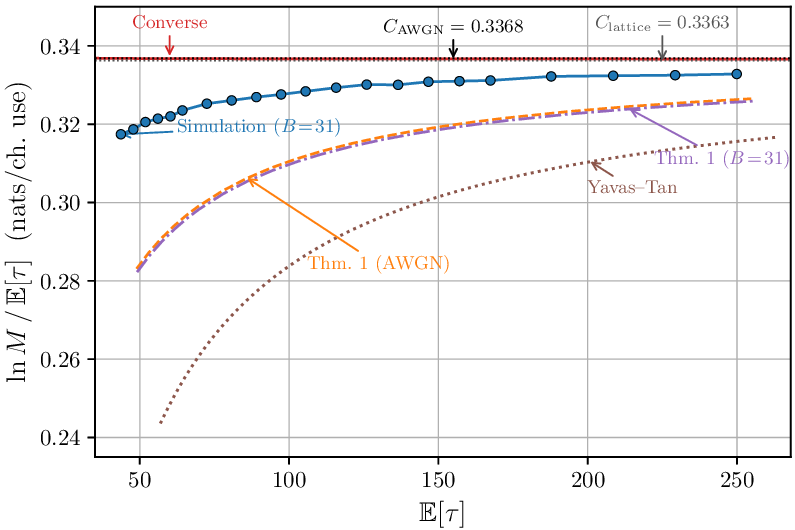}
    \caption{Rate vs.\ expected decoding time, BI-AWGN ($\sigma^2 = 1$, $B = 31$), $\epsilon = 10^{-3}$.}
    \label{fig:awgn_sim}
\end{figure}

\begin{figure}[t]
    \centering
    \includegraphics[width=\columnwidth]{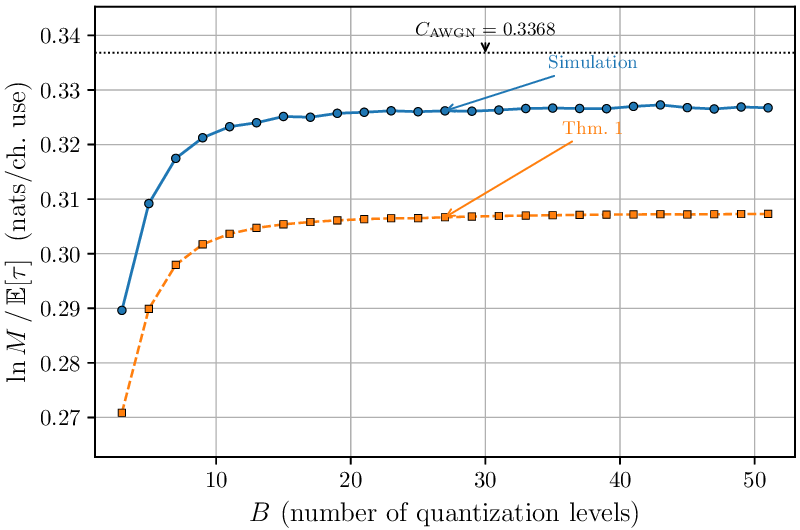}
    \caption{Rate vs.\ number of quantization levels $B$, BI-AWGN ($\sigma^2 = 1$, $\log_2 M = 40$, $\epsilon = 10^{-3}$).}
    \label{fig:awgn_varyB}
\end{figure}

\section{Analysis of the SED Scheme over Arbitrary BMS Channels} \label{sec:analysis}
We prove \thmref{thm:main} by first establishing $N_0(M, \epsilon')$ for any $\epsilon' \in (0, \frac{1}{2})$, then applying the stop-at-time-zero improvement. We begin with four auxiliary results on submartingale stopping times, first-passage times, ruin probabilities, and two-sided exit overshoots, then analyze the communication, confirmation, and recovery phases separately.

\subsection{Supporting Results on Stopping Times and Ruin Probabilities} \label{sec:supporting}
The first result is a general stopping-time bound for submartingales with positive drift and a uniform one-step mean-overshoot condition.
\begin{theorem}[Submartingale stopping-time bound]
\label{thm:hitting_time_drift_overshoot}
Let $(\Omega,\mathcal{F},\{\mathcal{F}_n\}_{n\ge 0},\mathbb{P})$ be a filtered probability space.
Let $\{S_n\}_{n\ge 0}$ be an $\{\mathcal{F}_n\}$-adapted, integrable process with $S_0 \in \mathbb{R}$ and increments
$\Delta S_n \triangleq S_n - S_{n-1}$, $n\ge 1$.
Fix $b\in\mathbb{R}$, and define the first-passage time
\begin{align}
\tau_b \triangleq \inf\big\{n\ge 1: S_n > b\big\}, \label{eq:taub}
\end{align}
with the convention $\inf\emptyset = \infty$.
Assume that the following two conditions hold.

\textit{(A1) Uniform conditional drift:}
there exists $\mu>0$ such that, for all $n\ge 1$,
\begin{align}
\mathbb{E}\!\left[\Delta S_n \mid \mathcal{F}_{n-1}\right] \ge \mu,
\qquad \text{a.s.}
\label{eq:A1_drift}
\end{align}

\textit{(A2) Uniform one-step mean-overshoot bound:}
there exists $\eta_0\in(0,\infty)$ such that, for all $n\ge 1$ and all $y\ge 0$,
\begin{align}
&\mathbb{E}\!\left[\Delta S_n - y \,\middle|\, \mathcal{F}_{n-1},\, \Delta S_n > y\right]
\le \eta_0,
\qquad \text{a.s.} \label{eq:A2_mean_excess}
\end{align}
on the event $\{\mathbb{P}[\Delta S_n>y | \mathcal{F}_{n-1}] > 0\}$.
Then $\tau_b<\infty$ almost surely and
\begin{align}
\mathbb{E}[\tau_b]
\le
\frac{(b-S_0)^+ + \eta_0}{\mu}.
\label{eq:Etau_bound}
\end{align}
In particular, if $S_0\le b$ then $\mathbb{E}[\tau_b]\le (b-S_0+\eta_0)/\mu$.
\end{theorem}
\begin{IEEEproof}
    See \appref{app:stopUniform}.
\end{IEEEproof}

The next result, from the renewal theory literature, provides a non-asymptotic bound on the expected first-passage time of an i.i.d.\ random walk with positive drift.
 \begin{theorem}[First-passage time of an i.i.d.\ random walk {\cite[Th.~1]{lorden1970}}, {\cite{mogul1974}}, {\cite[Cor.~3.1]{tchamkertenBurnashev}}]\label{thm:lordenBound}
        Let $X, X_1, X_2, \dots$ be i.i.d.\ random variables with $\E{X} = \mu > 0$ and $\E{|X|^3} < \infty$. Let $S_n = \sum_{i = 1}^n X_i$ and $\tau = \inf\{n \geq 1 \colon S_n \geq a\}$. Then, for any $a > 0$,
        \begin{align}
            \E{\tau} \leq \frac{1}{\mu} \left(a + \eta(P_X) \right), \label{eq:lordenbound}
        \end{align}
        where $\eta(P_X)$ is the overshoot parameter defined in \eqref{eq:etadef}. For $X \sim \mc{N}(\mu, \sigma^2)$, it also holds that
        \begin{align}
            \E{\tau} \leq \frac{a}{\mu} + 2 + \frac{4 \sigma}{\mu}. \label{eq:gaussspecial}
        \end{align}
    \end{theorem}
    \begin{IEEEproof}
    The bound \eqref{eq:lordenbound} follows from the universal overshoot bound $\E{S_{\tau} - a} \leq \eta(P_X)$, which combines the three bounds cited in \eqref{eq:etadef}, together with Wald's identity. The Gaussian case \eqref{eq:gaussspecial} is proved in \cite{tchamkertenBurnashev}.
    \end{IEEEproof}

The next result collects classical ruin probability bounds for random walks with positive drift.
\begin{theorem}[Spitzer's identity and Lundberg's bound]
\label{thm:spitzer-lundberg-sandwich}
Let $\{X_k\}_{k\ge 1}$ be i.i.d.\ real-valued random variables with mean
$\mu\triangleq\mathbb{E}[X_1]>0$. For $x\in\mathbb{R}$, define the random walk
\begin{align}
S_0=x,\qquad S_n=x+\sum_{k=1}^n X_k,\quad n\ge 1,
\end{align}
and let $\mathbb{P}_x,\mathbb{E}_x$ denote the probability measure and
expectation under which $S_0=x$. Define the hitting times
\begin{align}
\tau_- \triangleq \inf\{n\ge 1: S_n<0\},\qquad
\tau_A \triangleq \inf\{n\ge 1: S_n\ge A\},
\end{align}
and the ultimate ruin function
\begin{align}
    \psi(x) \triangleq \mathbb{P}_x[\tau_-<\infty] \quad \text{ for }x\ge 0. \label{eq:psifunc}
\end{align}
For $0\le x\le A$, define the two-sided ruin probability
\begin{align}
p_{\mathrm{ruin}}(x,A)\triangleq\mathbb{P}_x[\tau_-<\tau_A]. \label{eq:pfbdef}
\end{align}
\begin{enumerate}
\item[(i)]
Assume additionally that $\mathbb{P}[X_1<0]>0$.
Then
\begin{align}
\psi(0)= p_{\mr{ruin}}(0, \infty) = 1-\exp\!\left(-\sum_{n=1}^{\infty}\frac{1}{n}\,
\mathbb{P}_0[S_n<0]\right). \label{eq:spitzer-ruin}
\end{align}
\item[(ii)]
For every $A>0$ and $0\le x\le A$,
\begin{align}
\psi(x)-\psi(A)\ \le\ p_{\mathrm{ruin}}(x,A)\ \le\ \psi(x)
\ \le\ \psi(0). \label{eq:sandwich}
\end{align}
\item[(iii)]
Assume there exists $R>0$ (the adjustment coefficient) such that
$\mathbb{E}\big[e^{-R X_1}\big]=1$.
Then for all $x\ge 0$,
\begin{align}
\psi(x)\le e^{-R x}. \label{eq:lundberg}
\end{align}
Consequently, substituting \eqref{eq:lundberg} into the left-hand
side of \eqref{eq:sandwich},
\begin{align}
\psi(x)-e^{-R A}\ \le\ p_{\mathrm{ruin}}(x,A)\ \le\ \psi(x)
\ \le\ \psi(0). \label{eq:fb-bracket}
\end{align}
\end{enumerate}
The identity \eqref{eq:spitzer-ruin} is a standard specialization of
Spitzer's identity; see, e.g., \cite{spitzer1956,wendel1958}.
The inequality \eqref{eq:lundberg} is the Lundberg bound; see,
e.g., \cite[eq.~(3.2)]{asmussenalbrecher2010}.
\end{theorem}
\begin{IEEEproof}
    See \appref{app:ruin}.
\end{IEEEproof}

The final supporting result bounds the expected overshoot when a random walk exits a two-sided barrier from either side. This is needed to control the cost of falling back into the recovery phase during confirmation.
\begin{theorem}[Two-sided exit of a random walk]
\label{thm:two-sided-overshoot-fixed}
Let $\{X_k\}_{k\ge 1}$ be i.i.d.\ real-valued with $\E{X_1} = \mu > 0$.
Assume $\Prob{X_1 < 0} > 0$.
Define the random walk
\begin{align}
S_0 \triangleq 0, \quad 
S_n \triangleq \sum_{k=1}^n X_k,
\end{align}
and the two-sided stopping time
\begin{align}
\tau &\triangleq \inf\Big\{ n\ge 1 : S_n \notin [-B,A] \Big\}, \label{eq:tauBA}
\end{align}
where $A\ge 0$ and $B\ge 0$. Define the exit events
\begin{align}
\mathcal{E}_- &\triangleq \Big\{ S_\tau < -B \Big\}, \quad
\mathcal{E}_+ \triangleq \Big\{ S_\tau > A \Big\},
\end{align}
and the corresponding overshoots
\begin{align}
D_- &\triangleq -\big(S_\tau + B\big) \quad \text{on }\mathcal{E}_-, \quad
D_+ \triangleq S_\tau - A \quad \text{on }\mathcal{E}_+.
\end{align}
Define the exit probabilities $p_+ \triangleq \mathbb{P}\!\left[\mathcal{E}_+\right]$ and $p_- \triangleq \mathbb{P}\!\left[\mathcal{E}_-\right]$.
\begin{enumerate}
    \item[(i)] \textit{Positive-side overshoot.} Assume $\E{|X_1|^3} < \infty$. Then,
\begin{align}
\mathbb{E}\!\left[D_+ \mid \mathcal{E}_+\right]
&\le
\frac{\eta(P_{X_1})}{p_+}, \label{eq:pos-two-sided-fixed} \\
\mathbb{E}\!\left[D_+ \mid \mathcal{E}_+\right]
&\le \eta_0^{+}(P_{X_1}). \label{eq:pos-two-sided-me}
\end{align}
\item[(ii)] \textit{Negative-side overshoot.} Assume there exists $\lambda>0$ such that 
$\mathbb{E}\!\left[e^{-\lambda X_1}\right] = 1$.
Let $\bar{X}_1$ denote the tilted random variable defined via \eqref{eq:bartilt}, and define the tilted measure $\bar{\mathbb{P}}$ by $\frac{d\bar{\mathbb{P}}}{d\mathbb{P}}\big|_{\mathcal{F}_n} = e^{-\lambda S_n}$, with the associated expectation $\bar{\mathbb{E}}$. Let $\bar{p}_- \triangleq \bar{\mathbb{P}}[\mathcal{E}_-]$.
Assume $\E{|\bar{X}_1|^3} < \infty$ and $\bar{p}_- > 0$. Then,
\begin{align}
\mathbb{E}\!\left[D_- \mid \mathcal{E}_-\right]
&\le
\frac{\eta(P_{-\bar{X}_1})}{\bar{p}_-}, \label{eq:neg-two-sided-fixed} \\
\mathbb{E}\!\left[D_- \mid \mathcal{E}_-\right] &\leq \eta_0^{-}(P_{X_1}). \label{eq:neg-two-sided-me}
\end{align}
\end{enumerate}
\end{theorem}
\begin{IEEEproof}
    See \appref{app:twosided}.
\end{IEEEproof}

\subsection{Analysis of Error Probability and Expected Stopping Time} \label{sec:errorstop}

\paragraph{Error probability}
Fix $\epsilon' \in (0, \frac{1}{2})$. From the definition of the stopping time \eqref{eq:tau} and the decoder's estimate $\hat{W}$ in \eqref{eq:hatW}, with target error probability $\epsilon'$, we bound the error probability as
\begin{align}
    \Prob{\hat{W} \neq W} &= \E{1 - \max_{m \in [M]} \rho_m(\tau)} \leq \epsilon', \label{eq:errpb}
\end{align}
where the inequality follows from the stopping rule $\max_{m} \rho_m(\tau) \geq 1 - \epsilon'$.

\paragraph{Score process and stopping rule}
It remains to derive an upper bound on $\E{\tau}$. Following \cite{yang2022SED, antonini2024}, we define the score (log-posterior odds)
\begin{align}
    U_m(t) &\triangleq \log \frac{\rho_m(t)}{1 - \rho_m(t)}, \label{eq:Umdef} \\
    \tau &= \inf \left\{t \geq 0 \colon \max_{m \in [M]} U_m(t) \geq \log \frac{1-\epsilon'}{\epsilon'}\right\}. \label{eq:tauU}
\end{align}
By convention, $U_m(t) = \infty$ for $t \geq \tau + 1$. The following result from \cite[Lemmas~3, 6, 7]{yang2022SED} shows that conditioned on $W = m$, $\{U_m(t)\}_{t \geq 0}$ is a submartingale with drift at least $C$ when $U_m(t) < 0$ and drift $C_1$ when $U_m(t) \geq 0$. 
\begin{lemma}[Submartingale property {\cite[Lemmas~3, 6, 7]{yang2022SED}}] \label{lem:submar}
Let $\mc{F}_t$ denote the filtration generated by the output sequence $Y^t$. Let $X_{t+1}(j)$ denote the input symbol assigned to message $j \in [M]$ at time $t+1$ according to the SED rule in \eqref{eq:SED_Yang} and \eqref{eq:Xt}. Define the extrinsic probabilities
\begin{align}
    \tilde{\pi}_{x, j}(t) = \begin{cases}
        \frac{\pi_{x}(t) - \rho_j(t)}{1 - \rho_j(t)} \quad &\text{if } j \in S_{x}(t), \\
        \frac{\pi_{x}(t)}{1 - \rho_j(t)} \quad &\text{if } j \notin S_{x}(t),
    \end{cases} \quad x \in \{-, +\}. \label{eq:extrinsic}
\end{align}
Let $x_m \in \{-, +\}$ be the bin containing message $m$, i.e., $m \in S_{x_m}(t)$. Then $\tilde{\pi}_{x_m, m}(t) \leq \frac{1}{2}$.

Moreover, the submartingale increment can be written as
\begin{align}
    &w_m(t, Y_{t+1}) \triangleq U_m(t+1) - U_m(t) \notag \\
    &= \log \frac{P_{x_m}(Y_{t+1})}{\tilde{\pi}_{-, m}(t) P_{-}(Y_{t+1}) + \tilde{\pi}_{+, m}(t) P_{+}(Y_{t+1})}. \label{eq:wtY}
\end{align}
Conditioned on $W = m$,
\begin{align}
    \E{U_m(t+1) | \mc{F}_{t}, W = m} &\geq U_m(t) + C, &\text{if } U_m(t) < 0, \label{eq:driftC} \\
     \E{U_m(t+1) | \mc{F}_{t}, W = m} &= U_m(t) + C_1, &\text{if } U_m(t) \geq 0. \label{eq:driftC1}
\end{align}
\end{lemma}
\begin{IEEEproof}
In \cite{yang2022SED}, the result is proved for binary asymmetric channels where the capacity-achieving input distribution is not necessarily uniform. The same arguments apply to BMS channels.
\end{IEEEproof}

\paragraph{Reduction to a uniform per-message bound}
For each message $m \in [M]$, define the individual stopping time
\begin{align}
    \tau_m \triangleq \inf\left\{t \geq 0 \colon U_m(t) \geq \log \frac{1-\epsilon'}{\epsilon'} \right\}. \label{eq:taum}
\end{align}
Since $\tau = \min \limits_{m \in [M]} \tau_m$, we have
\begin{align}
    \E{\tau} &= \frac{1}{M} \sum_{m = 1}^{M} \E{\tau | W = m} \\
    &\leq \frac{1}{M} \sum_{m = 1}^M \E{\tau_m | W = m} \leq \max\limits_{m \in [M]}  \E{\tau_m | W = m}. \label{eq:Etaubound0}
\end{align}
It therefore suffices to derive a uniform bound on the conditional expectation $\E{\tau_m | W = m}$. 

\paragraph{Three-phase decomposition}
By \lemref{lem:submar}, conditioned on $W = m$, the increment behavior of $U_m(t)$ depends on the sign of $U_m(t)$: when $U_m(t) < 0$, the drift is at least $C$; when $U_m(t) \geq 0$, the process behaves as a random walk with mean increment $C_1$. To exploit this two-regime structure, we define the up-crossing and down-crossing times of level 0. Let $T_{\mr{down}}^{(0)} = 0$, and for $i \geq 0$,
\begin{align}
    T_{\mathrm{up}}^{(i)} &\triangleq \inf\{t > T_{\mathrm{down}}^{(i)} \colon U_m(t) \geq 0\}, \label{eq:Tup} \\
    T_{\mathrm{down}}^{(i+1)} &\triangleq \inf\{ t > T_{\mathrm{up}}^{(i)} \colon U_m(t) < 0\}. \label{eq:Tdown}
\end{align}
Here, $T_{\mathrm{up}}^{(i)}$ is the $(i+1)$-th time the score reaches the non-negative region, and $T_{\mr{down}}^{(i)}$ is the $i$-th time it falls below 0. Define the number of recoveries as
\begin{align}
    N_{\mathrm{rec}} \triangleq \sup\{i \geq 0 \colon T_{\mathrm{down}}^{(i)} < \infty\}. \label{eq:Nrec}
\end{align}
By convention, $U_m(T_{\mr{down}}^{(i)}) = \infty$ for $i > N_{\mr{rec}}$.

We partition the time indices $t$ contributing to $\tau_m$ into three sets: 
\begin{align}
    \mc{A}_{\mr{comm}} &\triangleq \{ 1 \leq t \leq \tau_m \colon U_m(\tilde{t}) < 0 \,\, \forall \,\tilde{t} \leq t-1\}, \label{eq:Acomm} \\
    \mc{A}_{\mr{rec}} &\triangleq \{1 \leq t \leq \tau_m \colon U_m(t-1) < 0, \notag \\
    &\quad \exists \,\tilde{t} < t-1 \text{ s.t.\ } U_m(\tilde{t}) \geq 0\}, \label{eq:Arec} \\
    \mc{A}_{\mr{conf}} &\triangleq \{1 \leq t \leq \tau_m \colon U_m(t-1) \geq 0\}. \label{eq:Aconf}
\end{align}
Let
\begin{align}
    \tau_{\mr{comm}} \triangleq |\mc{A}_{\mr{comm}}|, \quad
    \tau_{\mr{rec}} \triangleq |\mc{A}_{\mr{rec}}|, \quad
    \tau_{\mr{conf}} \triangleq |\mc{A}_{\mr{conf}}|.
\end{align}
These correspond to the total durations of the \emph{communication}, \emph{recovery}, and \emph{confirmation} phases from the perspective of the true message $W = m$. Then,
\begin{align}
    \tau_m = \tau_{\mathrm{comm}} + \tau_{\mathrm{rec}} + \tau_{\mathrm{conf}}. \label{eq:taudecomp}
\end{align}
Here, $\tau_{\mathrm{comm}} = T_{\mathrm{up}}^{(0)}$ is the first time $U_m(t)$ reaches the non-negative region, and the recovery and confirmation durations decompose as
\begin{align}
    \tau_{\mathrm{rec}} &= \sum_{i = 1}^{N_{\mr{rec}}} (T_{\mr{up}}^{(i)} - T_{\mr{down}}^{(i)}), \label{eq:taurec}\\
    \tau_{\mathrm{conf}} &= \sum_{i = 1}^{N_{\mr{rec}}} (T_{\mr{down}}^{(i)} - T_{\mr{up}}^{(i-1)}) + (\tau_{m} - T_{\mr{up}}^{(N_{\mr{rec}})}). \label{eq:tauconf_decomp}
\end{align}
Note that only $\tau_{\mr{comm}}$ is ``uninterrupted'' in the sense that the time indices in $\mc{A}_{\mr{comm}}$ are consecutive.

We now bound $\E{\tau_{\mr{comm}} | W = m}$, $\E{\tau_{\mr{conf}} | W = m}$, and $\E{\tau_{\mr{rec}} | W = m}$ separately in the following three subsections.

\subsection{Bound on the Communication Phase} \label{sec:commbound}
The following lemma provides a non-asymptotic upper bound on $\E{\tau_{\mr{comm}} | W = m}$. 
\begin{lemma}[Communication phase] \label{lem:tcomm} \label{lem:commphase}
    Let $P_{Y|X}$ be a BMS channel. Assume that the LLR $\Lambda$ is integrable. This implies that $P_{+}$ and $P_{-}$ are mutually absolutely continuous. Then, for a codebook of size $M$, 
    \begin{align}
        \E{\tau_{\mr{comm}} | W = m} \leq \frac{\log(M-1) + \log 2}{C}. \label{eq:tcommbound}
    \end{align}
\end{lemma}
\begin{IEEEproof}
To tighten the upper bound on $\E{\tau_{\mr{comm}} | W = m}$, as in \cite{yang2022SED, antonini2024}, we construct a surrogate submartingale $\{U'_m(t)\}_{t \geq 0}$ satisfying the properties
\begin{enumerate}
    \item for all $t\geq 0$, $U'_m(t) \leq U_m(t)$ a.s. with $U_m'(0) = U_m(0)$,
    \item $\E{U'_m(t + 1) | \mc{F}_t, W = m} = U'_m(t) + C$, \\
    \item the one-step mean-overshoot constant for $U'_m(t)$ is smaller than that of $U_m(t)$:
    \begin{align}
        &\sup \limits_{\substack{t \geq 1, u \geq 0}}\mathbb{E}\!\left[\Delta U'_m(t) - u \,\middle|\, \mathcal{F}_{t-1},\,W = m, \, \Delta U'_m(t) > u\right] \notag \\
        &< \sup \limits_{\substack{t \geq 1,  u \geq 0}}\mathbb{E}\!\left[\Delta U_m(t) - u \,\middle|\, \mathcal{F}_{t-1},\,W = m, \, \Delta U_m(t) > u\right]
    \end{align}
\end{enumerate}
where $\Delta U'_m(t) = U'_m(t) - U'_m(t-1)$ and $\Delta U_m(t) = U_m(t) - U_m(t-1)$ denote the submartingale increments. On $\{W = m\}$, define the stopping time for $U'_m(t)$ as
\begin{align}
    \tau'_{\mr{comm}} \triangleq \inf\{t \geq 0 \colon U'_m(t) \geq 0\}.
\end{align}
By Property 1, 
\begin{align}
     \E{\tau_{\mr{comm}} | W = m} \leq \E{\tau'_{\mr{comm}} | W = m}. \label{eq:tautauprime}
\end{align}
Moreover, applying \thmref{thm:hitting_time_drift_overshoot} to $\E{\tau'_\mr{comm} | W = m}$ with Properties~2 and 3 yields an overall tighter bound on $\E{\tau_{\mr{comm}} | W = m}$ than directly applying it to $\E{\tau_{\mr{comm}} | W = m}$.

\paragraph{Construction of the submartingale $\{U'_m(t)\}_{t \geq 0}$} The following construction is the generalization of the construction in \cite[eq.~(133)-(135)]{yang2022SED} to arbitrary BMS channels. Set $U'_m(0) = U_m(0) = - \log(M-1)$. For $t \geq 0$, $Y_{t+1} = y$, and $X_{t+1}(m) = x_m \in \{-, +\}$, set
\begin{align}
    U'_{m}(t + 1) = U'_m(t) + w'_m(t, y), \label{eq:Uprime}
\end{align}
where
\begin{align}
    &C(y, x_m) = \log \frac{P_{x_m}(y)}{P_Y(y)}, \\
    &P_Y(y) = \frac{1}{2} P_{+}(y) + \frac{1}{2}  P_{-}(y), \notag \\
    &w'_m(t, y) \label{eq:wprime} \\
    &=  \begin{cases}
        C(y, x_m) - \frac{P_{-x_m}(y)}{P_{x_m}(y)} \log \frac{P_Y(y)}{\sum_{x} \tilde{\pi}_{x, m}(t) P_{-x}(y)} \\
        \qquad \qquad \qquad  \text{if } \mathrm{sgn}(y) = \mathrm{sgn}(x_m), \\[4pt]
        C(y, x_m) + \log  \frac{P_Y(y)}{\sum_{x} \tilde{\pi}_{x, m}(t) P_{x}(y)} \\
        \qquad \qquad \qquad \text{if } \mathrm{sgn}(y) = -\mathrm{sgn}(x_m).
    \end{cases} \notag
\end{align}
Comparing \eqref{eq:wprime} with the true increment $w_m(t, y)$ in \eqref{eq:wtY}, the two definitions coincide whenever $\mathrm{sgn}(y) = -\mathrm{sgn}(x_m)$. In the following, we show that the constructed $U'_m(t)$ satisfies Properties~1 and 2. 

Without loss of generality, assume $x_m = +$, so that $Y_{t+1} \sim P_+$ given $W = m$. Define $\tilde{Y}_{t+1} \sim P_-$; by BMS symmetry, $\tilde{Y}_{t+1} \overset{d}{=} -Y_{t+1}$. From the definition of $w'_m$ in \eqref{eq:wprime}, using $\E{C(Y_{t+1}, +)} = C$,
\begin{align}
    &\E{w'_m(t, Y_{t+1}) | \mc{F}_t, W = m} = C - \mathbb{E}\Bigg[\frac{P_-(Y_{t+1})}{P_+(Y_{t+1})} \notag \\
    &\quad \quad \cdot 1\{Y_{t+1} \geq 0\} \cdot \log \frac{P_Y(Y_{t+1})}{\sum \limits_{x \in \{-, +\}} \tilde{\pi}_{x, m}(t) P_{-x}(Y_{t+1})} \Bigg] \notag \\
    & \quad + \E{1\{Y_{t+1} < 0\} \log \frac{P_Y(Y_{t+1})}{\sum \limits_{x \in \{-, +\}} \tilde{\pi}_{x, m}(t) P_{x}(Y_{t+1})}}. \label{eq:wprime_expand}
\end{align}
The change of measure $\E{\frac{P_-(Y_{t+1})}{P_+(Y_{t+1})} h(Y_{t+1})} = \E{h(\tilde{Y}_{t+1})}$ converts the second term in \eqref{eq:wprime_expand} to an expectation under $P_-$. Moreover, the substitution $Y_{t+1} \overset{d}{=} -\tilde{Y}_{t+1}$ together with the BMS identities $P_Y(-y) = P_Y(y)$ and $P_x(-y) = P_{-x}(y)$ converts the third term to the same expression. Therefore, the second and third terms in \eqref{eq:wprime_expand} cancel. More precisely, they differ only on the event $\{\tilde{Y}_{t+1} = 0\}$, where the logarithmic term equals zero. Hence,
\begin{align}
    \E{U'_m(t + 1) | \mc{F}_t, W = m} = U'_m(t) + C, \label{eq:C}
\end{align}
which is Property~2.

To prove Property 1, note that $w'_m(t, y) = w_m(t, y)$ when $\mathrm{sgn}(y) = - \mathrm{sgn}(x_m)$. When $\mathrm{sgn}(y) = \mathrm{sgn}(x_m)$, we have
\begin{align}
    w_m(t, y) &= C(y, x_m) + \log  \frac{P_Y(y)}{\sum \limits_{x \in \{-, +\}} \tilde{\pi}_{x, m}(t) P_{x}(y)} \\
    &\geq C(y, x_m) + \frac{P_{-x_m}(y)}{P_{x_m}(y)} \log  \frac{P_Y(y)}{\sum \limits_{x \in \{-, +\}} \tilde{\pi}_{x, m}(t) P_{x}(y)} \\
    &\geq C(y, x_m) +  \frac{P_{-x_m}(y)}{P_{x_m}(y)}  \log  \frac{\sum \limits_{x \in \{-, +\}} \tilde{\pi}_{x, m}(t) P_{-x}(y)}{P_Y(y)} \\
    &= w'_m(t, y), \label{eq:wprimem}
\end{align}
where the first inequality follows from the BMS channel assumption in \eqref{eq:PYXcond2}. The second inequality follows from the AM-GM inequality
\begin{align}
    (a p + b(1-p))\,(a (1-p) + b p)
    &\leq \left(\frac{a+b}{2}\right)^2
\end{align}
for $a, b \geq 0$ and $p \in [0, 1]$,
applied with $p \leftarrow \tilde{\pi}_{+, m}(t)$, $a \leftarrow P_{+}(y)$ and $b \leftarrow P_{-}(y)$. The above argument also implies that Property~3 holds.
The equality \eqref{eq:C} and the inequality \eqref{eq:wprimem} together show that the constructed submartingale satisfies the desired three properties.

\paragraph{A channel-independent uniform upper bound on $w'_m(t, y)$}
We next derive a uniform upper bound on $w'_m(t, y)$ that is independent of $t$, $y \in \mc{Y}$, and also the BMS channel $P_{Y|X}$. 

Recall that for $y \in \mc{Y}$ such that $\mr{sgn}(y) = - \mr{sgn}(x_m)$, 
\begin{align}
    w'_m(t, y) = \log \frac{P_{x_m}(y )}{\tilde{\pi}_{-, m}(t) P_{-}(y) + \tilde{\pi}_{+, m}(t) P_{+}(y)}.
\end{align}
Without loss of generality, assume that $x_m = -$ and $y \geq 0$. Then, by BMS monotonicity  \eqref{eq:PYXcond2}, $P_{x_m}(y ) \leq  P_{-x_m}(y)$. Therefore, $w'_m(t, y) \leq 0$ in this case. 

In the case where $\mr{sgn}(y) = \mr{sgn}(x_m)$, from the BMS assumptions in \eqref{eq:PYXcond1}--\eqref{eq:PYXcond2}, we have
\begin{align}
    w'_m(t, y) &= C(y, x_m) + \frac{P_{-x_m}(y)}{P_{x_m}(y)} \log \frac{\sum \limits_{x \in \{-, +\}} \tilde{\pi}_{x, m}(t) P_{-x}(y)}{P_Y(y)} \\
    &\leq C(y, x_m) + \frac{P_{-x_m}(y)}{P_{x_m}(y)} \log \frac{P_{x_m}(y)}{P_Y(y)} \label{eq:wineq}\\
    &= C(y, x_m) \left( 1 + \frac{P_{-x_m}(y)}{P_{x_m}(y)} \right), \label{eq:wprimeupper}
\end{align}
where in \eqref{eq:wineq}, we use the fact that $P_{x_m}(y) \geq p P_{x_m}(y ) + (1-p) P_{-x_m}(y) \geq P_{-x_m}(y)$ for any $p \in [0, 1]$. 

For $a > 0$ and $b \geq 0$, consider the function 
\begin{align}
f(a, b) &\triangleq \log \frac{a}{\frac{a}{2} + \frac{b}{2}}  \left(1 + \frac{b}{a} \right) = (1 + r) \log \frac{2}{1 + r},
\end{align}
where $r \triangleq \frac{b}{a} \geq 0$. Define $\varphi(x) = x \log \frac{2}{x}$. Since $\varphi'(x) = \log \frac{2}{x} - 1 < 0$ for all $x \geq 1$, $\varphi$ is strictly decreasing on $[1, \infty)$, so $f(a, b)$ takes its maximum value $\log 2$ at $r = 0$. Notice that the upper bound in \eqref{eq:wprimeupper} is equal to $f(P_{x_m}(y), P_{-x_m}(y))$. Therefore, for any BMS channel $P_{Y|X}$,
\begin{align}
w'_m(t, y) \leq \log 2. \label{eq:log2rho}
\end{align}

\paragraph{Completion of the proof of \lemref{lem:tcomm}} 
The integrability of $w'_m(t, Y)$ under $Y \sim P_{Y|X = x_m}$ follows from \lemref{lem:mixed_logratio_integrable} (Appendix~\ref{app:integrable}) with $P \leftarrow P_{x_m}$, $Q \leftarrow P_{-x_m}$, and $a = 1 - \tilde{\pi}_{x_m, m}(t) \geq \frac{1}{2}$, verifying the conditions of \thmref{thm:hitting_time_drift_overshoot}. Applying \thmref{thm:hitting_time_drift_overshoot} to $U'_m(t)$ with threshold $b = 0$, initial value $S_0 = -\log (M-1)$, drift $\mu = C$, and mean-overshoot constant $\eta_0 = \log 2$ from \eqref{eq:log2rho} yields \eqref{eq:tcommbound}.
\end{IEEEproof}

\subsection{Bound on the Confirmation Phase} \label{sec:confbound}

When $U_m(t) \geq 0$, message $m$ has posterior $\rho_m(t) \geq \frac{1}{2}$. The SED partition places $m$ alone in one of the bins ($S_{x_m}(t) = \{m\}$ for some $x_m$), so the extrinsic probabilities become $\tilde{\pi}_{x_m, m}(t) = 0$ and $\tilde{\pi}_{-x_m, m}(t) = 1$. The increment \eqref{eq:wtY} simplifies to
\begin{align}
    w_m(t, Y_{t+1}) = \log \frac{P_{x_m}(Y_{t+1})}{P_{-x_m}(Y_{t+1})}. \label{eq:wmconf}
\end{align}
By the BMS symmetry \eqref{eq:PYXcond1}, the increment $w_m(t, Y_{t+1})$ in \eqref{eq:wmconf} has the same distribution as $\Lambda$ regardless of whether $x_m = +$ or $x_m = -$. Therefore, during the confirmation phase, $\{U_m(t)\}$ evolves as an i.i.d.\ random walk with increments distributed as $\Lambda$ and drift $C_1$.

\begin{lemma}[Confirmation phase]\label{lem:confphase}
The total expected time spent in the confirmation region satisfies
\begin{align}
    \E{\tau_{\mathrm{conf}} \mid W = m} \leq \frac{1}{C_1}\left(\log\frac{1-\epsilon'}{\epsilon'} + \eta(P_\Lambda)\right). \label{eq:confbound}
\end{align}
\end{lemma}
\begin{IEEEproof}
The time indices in $\mathcal{A}_{\mr{conf}}$ are not necessarily consecutive because the process may temporarily drop below zero and recover. To isolate the confirmation-phase increments, we define i.i.d.\ random variables $\Lambda(1), \Lambda(2), \dots$ distributed as $\Lambda$. Let $t(i)$ be the $i$-th smallest index in $\mc{A}_{\mr{conf}}$. Then $w_m(t(i), Y_{t(i)}) = \Lambda(i)$ for $i = 1, \dots, |\mc{A}_{\mr{conf}}|$.

Define the auxiliary random walk
\begin{align}
    \tilde{U}(0) &= U_m(T_\mr{up}^{(0)}), \quad
    \tilde{U}(i) = \tilde{U}(0) + \sum_{j = 1}^{i} \Lambda(j), \quad i \geq 1,
\end{align}
which tracks only the confirmation-phase increments and ignores recovery-phase interruptions. For each $t(i) \in \mc{A}_{\mr{conf}}$, one can write
\begin{align}
    U_m(t(i)) &= \tilde{U}(i) + \sum_{k = 1}^\ell \big(U_m(T_{\mr{up}}^{(k)}) - U_m(T_{\mr{down}}^{(k)})\big), \label{eq:Umti}
\end{align}
where $\ell$ is the number of recovery phases completed up to $t(i)$. Since each term in the summation in \eqref{eq:Umti} is non-negative, $U_m(t(i)) \geq \tilde{U}(i)$.

Define $\tilde{\tau}_{\mr{conf}} \triangleq \inf\{i \geq 0 \colon \tilde{U}(i) \geq \log \frac{1-\epsilon'}{\epsilon'}\}$. Then $\tau_{\mr{conf}} \leq \tilde{\tau}_{\mr{conf}}$ almost surely. Applying \thmref{thm:lordenBound} to the i.i.d.\ random walk $\tilde{U}$ with drift $C_1$ and starting value $\tilde{U}(0) \geq 0$ yields \eqref{eq:confbound}. 
\end{IEEEproof}

\subsection{Bound on the Recovery Phase} \label{sec:recbound}
When the confirmation-phase random walk falls below $0$, a \emph{recovery event} begins. During recovery, $U_m(t) < 0$, and by \lemref{lem:submar}, the drift is at least $C$. The recovery event ends when $U_m$ returns to the non-negative region.

We first state two lemmas needed for the recovery analysis. The first improves the ruin probability bound from $Z(P_{Y|X})$ to $\chi(P_{Y|X})$; the second establishes the upper bound $\alpha(\epsilon')$ defined in \eqref{eq:Reps_halfsplit} on the exit-probability ratio.

\begin{lemma}[Ruin probability bound for BMS channels]
\label{lem:ruin-bhattacharyya}
Let $P_{Y|X}$ be a BMS channel.
Let $\Lambda(1), \Lambda(2), \ldots$ be i.i.d.\ LLRs
$\Lambda(i) \triangleq \log \frac{P_+(Y_i)}{P_-(Y_i)}$, $Y_i \sim P_+$,
and let $S_n \triangleq \sum_{i=1}^n \Lambda(i)$. Then
\begin{align}
\sum_{n=1}^{\infty} \frac{1}{n}\, \mathbb{P}[S_n < 0]
&\le \frac{1}{2}\log \frac{1}{1 - Z(P_{Y|X})}.
\label{eq:series-bhatt}
\end{align}
Consequently,
\begin{align}
\psi(0) \leq 1 - \sqrt{1 - Z(P_{Y|X})} = \chi(P_{Y|X}).
\label{eq:ruin-bhatt}
\end{align}
\end{lemma}
\begin{IEEEproof}
    See \appref{app:bhatta}.
\end{IEEEproof}

\begin{lemma}[Exit-probability ratio for BMS channels]
\label{lem:exit_ratio}
Let $P_{Y|X}$ be a BMS channel with LLR $\Lambda$ and ruin parameter $\chi(P_{Y|X})$. Fix $A > 0$, and consider the i.i.d.\ LLR random walk with two-sided barriers (in \eqref{eq:tauBA}) at $0$ and $A$. Let $p_{\mr{ruin}}(u, A)$ and $\bar{p}_{\mr{ruin}}(u, A)$ (defined in \eqref{eq:pfbdef}) denote the two-sided ruin probabilities under $\mathbb{P}$ and $\bar{\mathbb{P}}$, respectively, where $\bar{\mathbb{P}}$ denotes the measure of the random walk under LLR $\bar{\Lambda}$. Define the exit-probability ratio $f_A(u) \triangleq \frac{p_{\mr{ruin}}(u, A)}{\bar{p}_{\mr{ruin}}(u, A)}$ for $u \in [0, A)$.
\begin{enumerate}
    \item[(i)] \emph{Universal bound.} For every $u \in [0, A)$,
    \begin{align}
        f_A(u) \leq e^{-u} \leq 1. \label{eq:fA_universal}
    \end{align}
    \item[(ii)] \emph{Half-splitting bound.} 
    \begin{align}
        \sup_{u \in [0, A)} f_A(u) \leq \max\!\left\{\frac{\psi(0)}{1 - e^{-A/2}},\; e^{-A/2}\right\}. \label{eq:fA_halfsplit}
    \end{align}
    Combining with the ruin bound $\psi(0) \leq \chi(P_{Y|X})$ from \lemref{lem:ruin-bhattacharyya} yields
    \begin{align}
        \sup_{u \in [0, A)} f_A(u) \leq \max\!\left\{\frac{\chi(P_{Y|X})}{1 - e^{-A/2}},\; e^{-A/2}\right\}. \label{eq:fA_halfsplit_chi}
    \end{align}
\end{enumerate}
\end{lemma}
\begin{IEEEproof}
    See \appref{app:exit_ratio}.
\end{IEEEproof}

Evaluating \eqref{eq:fA_halfsplit_chi} at $A = \log \frac{1-\epsilon'}{\epsilon'}$ yields the bound on $\alpha(\epsilon')$ stated in \eqref{eq:Reps_halfsplit}.

\begin{lemma}[Recovery phase]\label{lem:recphase}
    The expected total time spent in recovery satisfies
    \begin{align}
        &\E{\tau_{\mr{rec}} \mid W = m} \label{eq:recbound} \\
        & \leq \frac{\min\!\big\{\alpha(\epsilon') \cdot \eta(P_\Lambda),\, \chi(P_{Y|X}) \cdot \eta_0^{-}(P_\Lambda)\big\} + \chi(P_{Y|X}) \log 2}{(1-\chi(P_{Y|X}))\,C}. \notag
    \end{align}
\end{lemma}
\begin{IEEEproof}
  Using the decomposition of $\tau_{\mr{rec}}$ in \eqref{eq:taurec}, we write
    \begin{align}
        &\E{\tau_{\mr{rec}} | W = m} \notag \\ 
        &= \sum_{i = 1}^{\infty} \E{1\{N_{\mr{rec}} \geq i\} (T_{\mr{up}}^{(i)} - T_{\mr{down}}^{(i)})} \label{eq:taurec0} \\
        &= \sum_{i = 1}^{\infty} \E{1\{N_{\mr{rec}} \geq i\} 1\{N_{\mr{rec}} \geq i-1\}(T_{\mr{up}}^{(i)} - T_{\mr{down}}^{(i)})} \label{eq:taurec1} \\
        &= \sum_{i = 1}^{\infty} \Prob{ N_{\mr{rec}} \geq i-1} \notag \\
        &\quad \cdot \mathbb{E}\Bigg[\Prob{ N_{\mr{rec}} \geq i \middle | N_{\mr{rec}} \geq i-1, U_m(T_{\mr{up}}^{(i-1)})} \notag \\
        &\quad \cdot \E{ T_{\mr{up}}^{(i)} - T_{\mr{down}}^{(i)} \middle | N_{\mr{rec}} \geq i, U_m(T_{\mr{up}}^{(i-1)})}\Bigg].\label{eq:taureclong}
    \end{align}

    We have
    \begin{align}
        &\Prob{N_{\mr{rec}} \geq i-1} \notag \\
        &= \prod_{j = 1}^{i-1} \Prob{N_{\mr{rec}} \geq j \middle | N_{\mr{rec}} \geq j-1} \\
        &= \prod_{j = 1}^{i-1} \E{\Prob{N_{\mr{rec}} \geq j \middle | N_{\mr{rec}} \geq j-1, U_m(T_{\mr{up}}^{(j-1)})}}. \label{eq:Nreclast}
    \end{align}

    Now fix $j \in \mathbb{N}$ and $u \in [0, \infty)$. Given $N_{\mr{rec}} \geq j-1$ and $U_m(T_{\mr{up}}^{(j-1)}) = u$, the probability of another fallback into the recovery phase is exactly a two-sided ruin probability for the confirmation phase random walk, i.e.,
    \begin{align}
         \tilde{U}(i, u) &\triangleq u + \sum_{\ell = 1}^{i} \Lambda(\ell)
    \end{align}
    with the associated stopping time
    \begin{align}
         \tau_{\mr{two-sided}}(u) &\triangleq \inf\left\{i \geq 0 \colon \tilde{U}(i, u) \notin \left[0, \log \frac{1-\epsilon'}{\epsilon'} \right)\right\}. 
    \end{align}
    Define 
    \begin{align}
        p_{\mr{ruin}}(u) &\triangleq \Prob{N_{\mr{rec}} \geq j \middle | N_{\mr{rec}} \geq j-1, U_m(T_{\mr{up}}^{(j-1)}) = u} \label{eq:pruindef}\\
        &=\Prob{\tilde{U}(\tau_{\mr{two-sided}}(u), u) < 0}.
    \end{align}
    Similarly, let $\bar{p}_{\mr{ruin}}(u)$ denote the two-sided ruin probability under the tilted measure $\bar{\mathbb{P}}$.
    
    By \lemref{lem:ruin-bhattacharyya}, for every $u \geq 0$,
    \begin{align}
        p_{\mr{ruin}}(u) \leq \chi(P_{Y|X}). \label{eq:pruinchi}
    \end{align}
Combining with the telescoping product for $\Prob{N_{\mr{rec}} \geq i-1}$ yields
\begin{align}
    \Prob{N_{\mr{rec}} \geq i-1}  \leq \big(\chi(P_{Y|X})\big)^{i-1}. \label{eq:chirec}
\end{align}

    To bound the conditional recovery duration, we reuse the communication-phase argument in \secref{sec:commbound}, now with starting value $U_m(T_{\mr{down}}^{(i)}) < 0$ instead of $-\log (M-1)$, and the threshold is 0. Applying \lemref{lem:tcomm} yields the upper bound
    \begin{align}
        &\E{ T_{\mr{up}}^{(i)} - T_{\mr{down}}^{(i)} \middle | N_{\mr{rec}} \geq i, U_m(T_{\mr{up}}^{(i-1)}), U_m(T_{\mr{down}}^{(i)})}  \\
        &\quad \leq \frac{-U_m(T_{\mr{down}}^{(i)}) + \log 2}{C}. \label{eq:durationbound}
    \end{align}
    We bound $-U_m(T_{\mr{down}}^{(i)}) \geq 0$, the negative-side overshoot of the confirmation-phase random walk at the moment of fallback, using two approaches from \thmref{thm:two-sided-overshoot-fixed}(ii).

    \emph{Bound~1 (via the tilted measure).} Applying \eqref{eq:neg-two-sided-fixed} with $X_1 \leftarrow \Lambda$ and using the identity $\eta(P_{-\bar{\Lambda}}) = \eta(P_\Lambda)$ from \remref{rem:tiltedLLR}, we obtain
    \begin{align}
        &\E{-U_m(T_{\mr{down}}^{(i)}) \middle | N_{\mr{rec}} \geq i, U_m(T_{\mr{up}}^{(i-1)}) = u}  \leq \frac{\eta({P_{\Lambda}})}{\bar{p}_{\mr{ruin}}(u)}. \label{eq:2overshoot_tilt}
    \end{align}

    \emph{Bound~2 (via the negative mean-excess).} Applying \eqref{eq:neg-two-sided-me} with $X_1 \leftarrow \Lambda$ yields
    \begin{align}
        &\E{-U_m(T_{\mr{down}}^{(i)}) \middle | N_{\mr{rec}} \geq i, u} \leq \eta_0^{-}(P_\Lambda). \label{eq:2overshoot_me}
    \end{align}

    We now bound the product $p_{\mr{ruin}}(u) \cdot \E{-U_m(T_{\mr{down}}^{(i)}) \mid \cdots}$ using each overshoot bound.

    \emph{From Bound~1.} Multiplying \eqref{eq:2overshoot_tilt} by $p_{\mr{ruin}}(u)$ gives
    \begin{align}
        &p_{\mr{ruin}}(u) \cdot \frac{\eta(P_\Lambda)}{\bar{p}_{\mr{ruin}}(u)} = \frac{p_{\mr{ruin}}(u)}{\bar{p}_{\mr{ruin}}(u)} \cdot \eta(P_\Lambda)  \leq \alpha(\epsilon') \cdot \eta(P_\Lambda), \label{eq:product_ratio}
    \end{align}
    where the last step uses \lemref{lem:exit_ratio}(ii).

    \emph{From Bound~2.} Multiplying \eqref{eq:2overshoot_me} by $p_{\mr{ruin}}(u) \leq \chi(P_{Y|X})$ from \eqref{eq:pruinchi} gives
    \begin{align}
        p_{\mr{ruin}}(u) \cdot \eta_0^{-}(P_\Lambda) \leq \chi(P_{Y|X}) \cdot \eta_0^{-}(P_\Lambda). \label{eq:product_me}
    \end{align}

    Combining \eqref{eq:product_ratio} and \eqref{eq:product_me} yields
    \begin{align}
        &p_{\mr{ruin}}(u) \cdot \E{-U_m(T_{\mr{down}}^{(i)}) \middle| N_{\mr{rec}} \geq i, u } \notag \\
        &\quad \leq \min\!\big\{\alpha(\epsilon') \cdot \eta(P_\Lambda),\; \chi(P_{Y|X}) \cdot \eta_0^{-}(P_\Lambda)\big\}. \label{eq:min_overshoot_product}
    \end{align}
    Applying the ruin bound \eqref{eq:pruinchi} to the $\log 2$ coefficient in \eqref{eq:durationbound} gives
    \begin{align}
        p_{\mr{ruin}}(u) \log 2 \leq \chi(P_{Y|X}) \log 2.  \label{eq:chilog2}
    \end{align}
    Substituting \eqref{eq:min_overshoot_product} and \eqref{eq:chilog2} into \eqref{eq:taureclong} and summing the geometric series with \eqref{eq:chirec} yields \eqref{eq:recbound}, completing the proof.
    \end{IEEEproof}
    
\subsection{Proof of \thmref{thm:main}} \label{sec:proofmain}

\paragraph{Establishing the bound $N_0(M, \epsilon')$}
Combining the error probability bound from \eqref{eq:errpb} with the three-phase decomposition of the expected stopping time:

From the decomposition \eqref{eq:taudecomp} and the bound \eqref{eq:Etaubound0}:
\begin{align}
    \E{\tau} &\leq \max_{m \in [M]} \Big(\E{\tau_{\mr{comm}} \mid W = m} \notag \\
    &\quad + \E{\tau_{\mr{conf}} \mid W = m}  + \E{\tau_{\mr{rec}} \mid W = m}\Big). \label{eq:totaldecomp}
\end{align}
Substituting the bounds from Lemmas~\ref{lem:commphase}, \ref{lem:confphase}, and \ref{lem:recphase}, we get
\begin{align}
    &\E{\tau} \notag \\
    &\leq \frac{\log(M-1) + \log 2}{C} + \frac{\log\frac{1-\epsilon'}{\epsilon'} + \eta(P_\Lambda)}{C_1} \notag \\
    & + \frac{\min\!\big\{\alpha(\epsilon') \cdot \eta(P_\Lambda),\, \chi(P_{Y|X}) \cdot \eta_0^{-}(P_\Lambda)\big\} + \chi(P_{Y|X}) \log 2}{(1-\chi(P_{Y|X}))\,C} \notag \\
    &= N_0(M, \epsilon'). \label{eq:finalbound}
\end{align}
Since the SED scheme is an $(N_0(M, \epsilon'), M, \epsilon')$-VLF code with $\Prob{\hat{W} \neq W} \leq \epsilon'$, it follows that $N^*(M, \epsilon') \leq N_0(M, \epsilon')$ for every $\epsilon' \in (0, \frac{1}{2})$.

\paragraph{Stop-at-time-zero improvement}
Fix $\epsilon_0 \in [0, \epsilon)$ and set $\epsilon' = \frac{\epsilon - \epsilon_0}{1 - \epsilon_0}$, so that $\epsilon' \in (0, \epsilon]$. Given the $(N_0(M, \epsilon'), M, \epsilon')$-VLF code constructed above, we define a new code using the common randomness $V$ from Definition~\ref{def:VLF} with $\mc{V} = \{0, 1\}$ and $P_V(1) = \epsilon_0$. If $V = 1$, the decoder outputs a fixed message at time $\tau = 0$ without any transmission. If $V = 0$, the SED code with target error probability $\epsilon'$ is executed. The new code has average decoding time
\begin{align}
    \E{\tau_{\mr{new}}} = (1-\epsilon_0) \E{\tau} \leq (1-\epsilon_0) N_0(M, \epsilon'), \label{eq:Etaunew}
\end{align}
and average error probability
\begin{align}
    \Prob{\hat{W} \neq W} &\leq \epsilon_0 \cdot \frac{M-1}{M} + (1-\epsilon_0) \epsilon' \\
    &\leq \epsilon_0 + (1-\epsilon_0) \epsilon' = \epsilon, \label{eq:errnew}
\end{align}
where the last equality follows from the choice of $\epsilon'$. Since \eqref{eq:Etaunew}--\eqref{eq:errnew} hold for every $\epsilon_0 \in [0, \epsilon)$, minimizing over $\epsilon_0$ yields \eqref{eq:Nach}.\hfill $\square$

\section{Low-Complexity Encoder via Lattice-LLR Grouping and SED Repair} \label{sec:algorithm}

The SED partitioning rule in \eqref{eq:SED_Yang} and \eqref{eq:Xt} requires partitioning $[M]$ into two bins $S_{-}(t-1)$ and $S_+(t-1)$ at every time step. A direct implementation that sorts all $M$ posteriors and searches for a valid split has complexity $O(M \log M)$ per step \cite{naghshvar2015EJS, yang2022SED}, which scales exponentially in the number of bits transmitted, $K = \log_2 M$, and is prohibitive for message sizes of practical interest.

Antonini \emph{et al.}~\cite{antonini2024} address this for the BSC by introducing the TOP algorithm, which produces a partition satisfying the SEAD condition \eqref{eq:SED_antonini1}--\eqref{eq:SED_antonini2} at complexity $O(\log^2 M)$ by grouping messages by Hamming type and partitioning at the group level. However, as noted in \secref{sec:YangSED}, the SEAD condition is strictly weaker than the SED condition \eqref{eq:SED_Yang}; in particular, SEAD does not directly yield the uniform conditional drift bound used in \lemref{lem:submar}.

In this section, we describe a low-complexity encoder that satisfies the \emph{exact} SED condition \eqref{eq:SED_Yang} at every time step, while retaining the grouped-implementation efficiency of Antonini \emph{et al.}'s approach. The algorithm proceeds in two stages: (1)~a TOP initialization that produces an SEAD-valid partition at low cost, followed by (2)~a batched repair procedure that transfers messages from minimum-posterior groups until the two-sided SED constraint is met. To enable grouping for general BMS channels (beyond the BSC), we assume throughout this section that the LLR $\Lambda$ takes values on a lattice; this assumption is enforced in practice via output quantization, whose design and performance analysis are deferred to \secref{sec:quantization}.

\subsection{Lattice Assumption and Grouped Representation} \label{sec:lattice}
\begin{definition}[Lattice LLR]\label{def:lattice}
We say that the BMS channel $P_{Y|X}$ has a \emph{lattice LLR} with step $\delta > 0$, interior half-range $L \in \mathbb{N}$, and tail label $L_{\mr{tail}} \in \mathbb{N}$ satisfying $L_{\mr{tail}} \geq L$, if the LLR $\Lambda = \log \frac{P_+(Y)}{P_-(Y)}$ satisfies
\begin{align}
    \Lambda \in \{0, \pm\delta, \dots, \pm(L{-}1)\delta, \pm L_{\mr{tail}}\delta\} \label{eq:lattice}
\end{align}
almost surely under $Y \sim P_+$. The number of distinct LLR values is $B = 2L + 1$.

We denote the set of lattice indices by 
\begin{align}
   \mc{K}_B \triangleq \{0, \pm 1, \dots, \pm(L{-}1), \pm L_{\mr{tail}}\}. \label{eq:kb}
\end{align}

For discrete-output channels such as the BSC and BEC, the lattice is contiguous ($L_{\mr{tail}} = L$) and determined by the channel parameters. For continuous-output channels, the lattice structure and the gap $L_{\mr{tail}} - L \geq 0$ are imposed through output quantization; see \secref{sec:quantization}. 
\end{definition}

Under the lattice assumption, denote the observed lattice index at time $t$ by $n_t \in \mc{K}_B$, so that $\Lambda(Y_t) = n_t \delta$. By Bayes' rule, the log-posterior of message $m$ satisfies
\begin{align}
    \log \rho_m(t) &= \log \rho_m(t{-}1) + \log P_{X_t(m)}(Y_t) - \log Z_t, \label{eq:logposteriorupdate}
\end{align}
where $Z_t = \pi_{-}(t{-}1) P_{-}(Y_t) + \pi_{+}(t{-}1) P_+(Y_t)$ is the normalizing constant. Since
\begin{align}
    \log P_{X_t(m)}(Y_t) = \log P_-(Y_t) + 1\{m \in S_+(t-1)\} \cdot \Lambda(Y_t),
\end{align}
unrolling \eqref{eq:logposteriorupdate} from the uniform initialization $\rho_m(0) = 1/M$ gives
\begin{align}
    \log \rho_m(t) = S(t) + \delta \cdot k_m(t), \label{eq:logpost_lattice}
\end{align}
where
\begin{align}
    S(t) &\triangleq - \log M + \sum_{s=1}^t \big(\log P_-(Y_s) - \log Z_s\big) \label{eq:commonshift}
\end{align}
is a common shift independent of $m$, and
\begin{align}
    k_m(t) \triangleq \sum_{s=1}^{t} 1\{m \in S_+(s-1)\} \cdot n_s \in \mathbb{Z} \label{eq:label}
\end{align}
is the \emph{cumulative LLR label} of message $m$ at time $t$, with $k_m(0) = 0$ for all $m \in [M]$.

The representation \eqref{eq:logpost_lattice} has two immediate consequences.

\begin{lemma}[Grouped posteriors] \label{lem:grouped}
Under the lattice assumption in \defnref{def:lattice}, messages with the same cumulative label $k_m(t) = k_j(t)$ have the same posterior $\rho_m(t) = \rho_j(t)$. Consequently, the $M$ messages can be partitioned into groups indexed by distinct label values $k \in \mathbb{Z}$, and both the partition and the posterior update can be performed at the group level.
\end{lemma}
\begin{IEEEproof}
Immediate from \eqref{eq:logpost_lattice}: the common shift $S(t)$ cancels in any log-posterior difference.
\end{IEEEproof}

For each active label $k$ at time $t$, let $N_k(t)$ denote the number of messages with label $k$, and define the per-message posterior and group mass as
\begin{align}
    r_k(t) &\triangleq \frac{e^{\delta k}}{\sum_{g} N_g(t)\, e^{\delta g}}, \label{eq:rk} \\
    q_k(t) &\triangleq N_k(t) \cdot r_k(t). \label{eq:qk}
\end{align}
The partition masses decompose as
\begin{align}
    \pi_{-}(t{-}1) &= \sum_k n_k^{-}(t{-}1)\, r_k(t{-}1), \\
    \pi_+(t{-}1) &= \sum_k n_k^{+}(t{-}1)\, r_k(t{-}1),
\end{align}
where $n_k^{-}(t{-}1)$ and $n_k^{+}(t{-}1)$ are the numbers of group-$k$ messages assigned to $S_-(t{-}1)$ and $S_+(t{-}1)$, respectively, with $n_k^{-}(t{-}1) + n_k^{+}(t{-}1) = N_k(t{-}1)$. Let $G_t$ denote the number of distinct active labels at time $t$.
\begin{lemma}[Group count] \label{lem:groupcount}
$G_t \leq 1 + L_{\mr{tail}}\,t$.
\end{lemma}
\begin{IEEEproof}
At $t = 0$, all messages share the label $k_m(0) = 0$, so $G_0 = 1$. At each time step, the label update \eqref{eq:label} shifts labels of messages in $S_+(t{-}1)$ by $n_t$ with $|n_t| \leq L_{\mr{tail}}$, while labels in $S_-(t{-}1)$ remain unchanged. This increases the range of active integer labels by at most $|n_t| \leq L_{\mr{tail}}$, giving $G_t \leq 1 + L_{\mr{tail}}\,t$.
\end{IEEEproof}
At a typical decoding time $\tau = O(K)$, the group count is $G_{\tau} = O(L_{\mr{tail}} K)$, which is polynomial in $K$ since $L_{\mr{tail}} = O(B)$.

\subsection{TOP Initialization} \label{sec:TOP}
The TOP rule, introduced by Antonini \emph{et al.}~\cite{antonini2024} for the BSC, is a single-threshold partitioning method that operates on the sorted list of group posteriors. We state the rule in the notation of the present paper, adapted for general lattice-LLR BMS channels.

Sort the active labels in decreasing order: $k^{(1)} > k^{(2)} > \dots > k^{(G_{t-1})}$, which by \eqref{eq:rk} corresponds to decreasing per-message posteriors. The TOP rule builds up the bin $S_-(t{-}1)$ by adding groups in this order until the cumulative mass first reaches or exceeds $\frac{1}{2}$.

\begin{definition}[Crossing group and overshoot] \label{def:crossing}
The \emph{crossing group} $k^{\star}$ is the label of the first group for which the cumulative mass reaches $\frac{1}{2}$. Let $P_{\mr{pre}}$ denote the cumulative mass of all groups added before $k^{\star}$, and let $a^{\star} = \lceil (\frac{1}{2} - P_{\mr{pre}}) / r_{k^{\star}} \rceil$ be the minimum number of messages from group $k^{\star}$ needed to reach $\frac{1}{2}$. The \emph{overshoot} is $d = P_{\mr{pre}} + a^{\star} r_{k^{\star}} - \frac{1}{2}$, which satisfies $0 \leq d \leq r_{k^{\star}}$.
\end{definition}

\begin{definition}[TOP partition] \label{def:TOP}
    If the maximum per-message posterior satisfies $r_{k^{(1)}} \geq \frac{1}{2}$, the TOP rule places a single message from the top group into $S_-(t{-}1)$ (singleton safeguard). Otherwise, TOP assigns $a$ messages from the crossing group $k^{\star}$ to $S_-(t{-}1)$, where
    \begin{align}
        a = \begin{cases}
            a^{\star}, &\text{if } d \leq r_{k^{\star}} / 2, \\
            a^{\star} - 1, &\text{if } d > r_{k^{\star}} / 2,
        \end{cases} \label{eq:TOPchoice}
    \end{align}
    with all groups preceding $k^{\star}$ fully assigned to $S_-(t{-}1)$ and all remaining messages assigned to $S_+(t{-}1)$.
\end{definition}

The TOP partition is \emph{contiguous} with respect to the posterior-sorted order: $S_-(t{-}1)$ consists of the messages with the largest posteriors. As shown in \cite[Th.~2]{antonini2024}, TOP satisfies the SEAD condition \eqref{eq:SED_antonini1}. However, it does \emph{not} guarantee the SED condition \eqref{eq:SED_Yang} in general \cite[Sec.~IV-A]{antonini2024}. The failure occurs in the back-off case ($a = a^{\star} - 1$) of \eqref{eq:TOPchoice}: the resulting imbalance
\begin{align}
    \Delta(t{-}1) \triangleq \pi_-(t{-}1) - \pi_+(t{-}1)
\end{align}
can be negative with $|\Delta(t{-}1)|$ exceeding $\min_{j \in S_+(t-1)} \rho_j(t{-}1)$, which violates the SED lower bound in \eqref{eq:SED_Yang}. Enforcing SED may require \emph{non-contiguous} partitions with respect to the posterior-sorted order.

\subsection{SED Repair via Minimum-Mass Moves} \label{sec:repair}
We restore the SED condition by transferring messages from the minimum-posterior group on the violating side. The overall strategy is:
\begin{enumerate}
    \item \textit{TOP initialization:} Run the TOP rule (\defnref{def:TOP}) to obtain an initial partition satisfying SEAD.
    \item \textit{SED repair:} If the SED condition \eqref{eq:SED_Yang} is violated, repair the partition by moving messages in batches until SED holds.
\end{enumerate}
Since an SEAD-valid partition with non-negative imbalance $\Delta(t{-}1) \geq 0$ automatically satisfies SED (see \eqref{eq:SED_Yang} and \eqref{eq:SED_antonini1}), repair is triggered only when $\Delta(t{-}1) < 0$ and the SED lower bound $\Delta(t{-}1) \geq -\min_{j \in S_+(t-1)} \rho_j(t{-}1)$ fails. For completeness, the algorithm also handles the symmetric upper-bound violation.

Moving a single message with posterior $\rho$ from $S_+(t{-}1)$ to $S_-(t{-}1)$ increases $\Delta(t{-}1)$ by $2\rho$. In the grouped implementation, we operate on entire minimum-posterior groups in batches. Suppose the current violation is $\Delta(t{-}1) < -\rho_{\min}^{+}$, where $\rho_{\min}^{+} \triangleq \min_{j \in S_+(t-1)} \rho_j(t{-}1)$. Let $k_{\min}$ be the label achieving $\rho_{\min}^{+}$, let $r = r_{k_{\min}}$, and let $n$ be the number of group-$k_{\min}$ messages currently in $S_+(t{-}1)$. We transfer
\begin{align}
    a = \min\left\{n, \left\lceil \frac{-r - \Delta(t{-}1)}{2r} \right\rceil \right\} \label{eq:repaircount}
\end{align}
of these messages to $S_-(t{-}1)$. If $a = n$ (the batch is \emph{saturating}), the entire minimum group moves to $S_-(t{-}1)$, a new minimum group emerges in $S_+(t{-}1)$, and the procedure repeats. If $a < n$, the SED lower bound is restored and the loop terminates. The symmetric repair for the upper-bound violation $\Delta(t{-}1) > \min_{j \in S_-(t-1)} \rho_j(t{-}1)$ moves messages from $S_-(t{-}1)$ to $S_+(t{-}1)$ analogously.

\begin{lemma}[Repair termination and SED enforcement] \label{lem:repair}
    The batched repair procedure terminates after at most $G_{t-1}$ iterations and produces a partition satisfying the SED condition \eqref{eq:SED_Yang}.
\end{lemma}
\begin{IEEEproof}
    Each saturating iteration removes the current minimum group from the violating side, reducing the number of active groups on that side by one. Since there are at most $G_{t-1}$ groups in total, the loop performs at most $G_{t-1}$ iterations before a non-saturating batch occurs. Upon termination of a non-saturating batch moving $a$ messages of posterior $r$ from $S_+(t{-}1)$ to $S_-(t{-}1)$, the new imbalance satisfies $\Delta_{\mr{new}} = \Delta(t{-}1) + 2ar \geq -r$. Since $r$ is the minimum posterior in $S_+(t{-}1)$ (or larger, if the minimum group was fully exhausted), the SED lower bound holds. The upper-bound case is symmetric.
\end{IEEEproof}

\subsection{Identity Tracking via Interval Lists} \label{sec:intervals}
The group label $k$ determines the posterior value but not the message identity. To recover the decoded message index $\hat{W} \in [M]$ at the stopping time $\tau$, the algorithm must track which messages belong to which group. Storing all $M$ message indices individually would negate the complexity savings from grouping.

We adopt a compact representation. For each active label $k$ at time $t$, we maintain an ordered list of disjoint index intervals
\begin{align}
    \mc{J}_k(t) = \big[(s_1, \ell_1), \dots, (s_p, \ell_p)\big], \label{eq:intervallist}
\end{align}
representing the set $\bigcup_{u=1}^p \{s_u, s_u + 1, \dots, s_u + \ell_u - 1\}$. The total message count in group $k$ is $N_k(t) = \sum_{u = 1}^p \ell_u$.

\textit{Prefix split.} When the partition assigns $n_k^{-}(t{-}1)$ messages from group $k$ to $S_-(t{-}1)$, the algorithm takes the first $n_k^-(t{-}1)$ indices from $\mc{J}_k(t{-}1)$ in order, splitting at most one interval.

\textit{Label update and coalescing.} After observing the lattice index $n_t$, the labels of messages in $S_+(t{-}1)$ shift by $n_t$ while those in $S_-(t{-}1)$ remain unchanged. For each active label $g$ at time $t$, the updated interval list is
\begin{align}
    \mc{J}_g(t) = \mc{J}_g^{(-)}(t{-}1) \,\|\, \mc{J}_{g - n_t}^{(+)}(t{-}1), \label{eq:labelupdate}
\end{align}
where $\|$ denotes list concatenation. After concatenation, adjacent intervals $(s, \ell)$ and $(s + \ell, r)$ are \emph{coalesced} into $(s, \ell + r)$.

Let $F_t$ denote the total number of stored interval fragments across all groups at time $t$.
\begin{lemma}[Fragment count] \label{lem:fragments}
    $F_t \leq 1 + 2t$.
\end{lemma}
\begin{IEEEproof}
    At $t = 0$, there is a single fragment: $\mc{J}_0(0) = [(1, M)]$, so $F_0 = 1$. At each round, new fragments arise only from partial splits. The TOP initialization partially splits at most one group (the crossing group $k^{\star}$), creating at most one new fragment. The SED repair performs at most one non-saturating batch---all saturating batches move entire groups without splitting---creating at most one additional fragment. Hence $F_t \leq 1 + 2t$.
\end{IEEEproof}

\subsection{Complete Algorithm} \label{sec:completealg}
We present the full encoder--decoder procedure. \algref{alg:main} gives the main loop. \algref{alg:groupedstate} computes the grouped representation and performs the TOP initialization. \algref{alg:sedrepair} runs the batched SED repair. \algref{alg:materialize} materializes the partition via prefix splits, performs the label update after channel observation, and coalesces adjacent interval fragments.

We write $\Delta \triangleq \pi_-(t{-}1) - \pi_+(t{-}1)$ for the partition imbalance.

\begin{algorithm}[t]
\caption{Lattice-TOP with SED Repair (Main Loop)}
\label{alg:main}
\begin{algorithmic}[1]
\REQUIRE $K$, $M = 2^K$, lattice parameters $\delta, L, L_{\mr{tail}}$, channel $P_{Y|X}$, threshold $1 - \epsilon$.
\STATE Initialize: $t \leftarrow 0$; single active label $k = 0$ with $N_0 \leftarrow M$, $\mc{J}_0 \leftarrow [(1, M)]$.
\STATE $\{k, r_k, q_k, n_k^-\} \leftarrow$ \textsc{GroupedState \& TopInit}$(\{k, N_k, \mc{J}_k\})$.
\COMMENT{\algref{alg:groupedstate}}
\WHILE{$\max_k r_k < 1 - \epsilon$}
    \STATE $t \leftarrow t + 1$.
    \STATE $\{k, r_k, q_k, n_k^-\} \leftarrow$ \textsc{GroupedState \& TopInit}$(\{k, N_k, \mc{J}_k\})$.
    \COMMENT{\algref{alg:groupedstate}}
    \STATE $n_k^+ \leftarrow N_k - n_k^-$ for all active $k$.
    \STATE $\pi_- \leftarrow \sum_k n_k^- r_k$; \, $\Delta \leftarrow 2\pi_- - 1$.
    \STATE $\{n_k^-, n_k^+, \Delta\} \leftarrow$
    \textsc{SedRepair}$(\{k, r_k, n_k^-, n_k^+, \Delta\})$.
    \COMMENT{\algref{alg:sedrepair}}
    \STATE Transmit $X_t = +$ if $W \in \bigcup_k \mc{J}_k^{(+)}$, \, $X_t = -$ otherwise.
    \STATE Observe $Y_t$; compute lattice index $n_t$; feed back $n_t$.
    \STATE $\{(\mc{J}_k, N_k)\} \leftarrow$
    \textsc{Materialize, Update \& Coalesce}$(\{\mc{J}_k, n_k^-\}, n_t)$.
    \COMMENT{\algref{alg:materialize}}
\ENDWHILE
\STATE $k_{\max} \leftarrow \argmax_k r_k$; output $\hat{W}$ as the first index in $\mc{J}_{k_{\max}}$.
\end{algorithmic}
\end{algorithm}

\begin{algorithm}[t]
\caption{\textsc{GroupedState \& TopInit}}
\label{alg:groupedstate}
\label{alg:topinit}
\begin{algorithmic}[1]
\REQUIRE Active classes $\{(k, N_k, \mc{J}_k)\}$.
\ENSURE Sorted labels (decreasing $k$); values $r_k, q_k$; counts $\{n_k^-\}$.
\STATE Sort active labels: $k^{(1)} > \cdots > k^{(G)}$.
\STATE $\Sigma \leftarrow \sum_k N_k\, e^{\delta k}$.
\FOR{each active label $k$}
    \STATE $r_k \leftarrow e^{\delta k} / \Sigma$; \, $q_k \leftarrow N_k\, r_k$.
\ENDFOR
\STATE Set $n_k^- \leftarrow 0$ for all active $k$.
\IF{$r_{k^{(1)}} \geq \frac{1}{2}$}
    \STATE $n_{k^{(1)}}^- \leftarrow 1$; \textbf{return}.
    \COMMENT{Singleton safeguard}
\ENDIF
\STATE $\mr{cum} \leftarrow 0$.
\FOR{$j = 1$ to $G$}
    \IF{$\mr{cum} + q_{k^{(j)}} < \frac{1}{2}$}
        \STATE $n_{k^{(j)}}^- \leftarrow N_{k^{(j)}}$;
        \, $\mr{cum} \leftarrow \mr{cum} + q_{k^{(j)}}$.
    \ELSE
        \STATE $k^{\star} \leftarrow k^{(j)}$; \textbf{break}.
    \ENDIF
\ENDFOR
\STATE $r^{\star} \leftarrow r_{k^{\star}}$;
\, $a^{\star} \leftarrow \lceil (\frac{1}{2} - \mr{cum}) / r^{\star} \rceil$.
\STATE $d \leftarrow \mr{cum} + a^{\star} r^{\star} - \frac{1}{2}$.
\IF{$d \leq r^{\star} / 2$}
    \STATE $n_{k^{\star}}^- \leftarrow a^{\star}$.
\ELSE
    \STATE $n_{k^{\star}}^- \leftarrow a^{\star} - 1$.
\ENDIF
\end{algorithmic}
\end{algorithm}

\begin{algorithm}[t]
\caption{\textsc{SedRepair}: Batched SED Repair}
\label{alg:sedrepair}
\begin{algorithmic}[1]
\REQUIRE Labels; $r_k$; $n_k^-, n_k^+$; $\Delta$
\ENSURE Repaired counts satisfying \eqref{eq:SED_Yang}
\WHILE{true}
    \STATE $\rho_{\min}^- \leftarrow \min\{r_k : n_k^- > 0\}$;
    \, $\rho_{\min}^+ \leftarrow \min\{r_k : n_k^+ > 0\}$
    \IF{$-\rho_{\min}^+ \leq \Delta \leq \rho_{\min}^-$}
        \STATE \textbf{return} \hfill $\triangleright$ SED satisfied
    \ENDIF
    \IF{$\Delta < -\rho_{\min}^+$}
        \STATE $k_{\min} \leftarrow$ label achieving $\rho_{\min}^+$;
        \, $r \leftarrow r_{k_{\min}}$; \, $n \leftarrow n_{k_{\min}}^+$
        \STATE $a \leftarrow \min\big\{n, \lceil (-r - \Delta) / (2r) \rceil\big\}$
        \STATE $n_{k_{\min}}^+ \!\mathrel{-}= a$; \, $n_{k_{\min}}^- \!\mathrel{+}= a$; \, $\Delta \leftarrow \Delta + 2ar$
    \ELSE
        \STATE $k_{\min} \leftarrow$ label achieving $\rho_{\min}^-$;
        \COMMENT{$\Delta > \rho_{\min}^-$}
        \, $r \leftarrow r_{k_{\min}}$; \, $n \leftarrow n_{k_{\min}}^-$
        \STATE $a \leftarrow \min\big\{n, \lceil (\Delta - r) / (2r) \rceil\big\}$
        \STATE $n_{k_{\min}}^- \!\mathrel{-}= a$; \, $n_{k_{\min}}^+ \!\mathrel{+}= a$; \, $\Delta \leftarrow \Delta - 2ar$
    \ENDIF
\ENDWHILE
\end{algorithmic}
\end{algorithm}

\begin{algorithm}[t]
\caption{\textsc{Materialize, Update \& Coalesce}}
\label{alg:materialize}
\label{alg:update}
\label{alg:coalesce}
\begin{algorithmic}[1]
\REQUIRE For each active $k$: list $\mc{J}_k$, count $n_k^-$; lattice index $n_t$.
\ENSURE Updated $\{\mc{J}_g\}$ and $\{N_g\}$ after observation.
\STATE \COMMENT{\textit{Phase~1: Prefix split}}
\FOR{each active label $k$}
    \STATE Split $\mc{J}_k$ into $\mc{J}_k^{(-)}$ (first $n_k^-$ indices) and $\mc{J}_k^{(+)}$ (rest), splitting at most one interval at the boundary.
\ENDFOR
\STATE \COMMENT{\textit{Phase~2: Label update}}
\FOR{each label $g$ that will be nonempty}
    \STATE $\mc{J}_g \leftarrow \mc{J}_g^{(-)} \,\|\, \mc{J}_{g - n_t}^{(+)}$.
    \COMMENT{Labels in $S_+$ shift by $n_t$}
\ENDFOR
\STATE \COMMENT{\textit{Phase~3: Coalesce adjacent fragments}}
\FOR{each nonempty $\mc{J}_g$}
    \STATE Merge adjacent pairs $(s, \ell), (s{+}\ell, r) \to (s, \ell{+}r)$ in $\mc{J}_g$.
    \STATE $N_g \leftarrow \sum_{(s, \ell) \in \mc{J}_g} \ell$.
\ENDFOR
\STATE Remove empty classes.
\end{algorithmic}
\end{algorithm}

\subsection{Complexity Analysis} \label{sec:complexity}
We summarize the per-round and cumulative complexity of the algorithm.

\textit{TOP initialization.} The TOP portion of \algref{alg:groupedstate} iterates over the sorted group list once: $O(G_{t-1})$ operations.

\textit{SED repair.} By \lemref{lem:repair}, the repair loop runs at most $G_{t-1}$ iterations. Each iteration involves $O(1)$ arithmetic and $O(G_{t-1})$ work to locate the new minimum posterior (or $O(\log G_{t-1})$ with a min-heap). The worst-case cost per round is $O(G_{t-1}^2)$, improving to $O(G_{t-1} \log G_{t-1})$ with a heap.

\textit{Partition materialization, label update, and coalescing.} \algref{alg:materialize} processes each interval in each active group once: $O(F_t)$, where $F_t \leq 1 + 2t$ by \lemref{lem:fragments}.

\textit{Overall per-round complexity.} The dominant cost is the SED repair. Since $G_{t-1} \leq 1 + L_{\mr{tail}}\,(t-1)$ by \lemref{lem:groupcount}, the per-round cost at round $t$ is $O(L_{\mr{tail}}^2\,t^2)$, or $O(L_{\mr{tail}}\,t \log(L_{\mr{tail}} t))$ with a heap.

For any deterministic horizon $T$, the cumulative cost up to time $T$ is
\begin{align}
    \sum_{t=1}^{T} O(L_{\mr{tail}}^2\,t^2) = O(L_{\mr{tail}}^2\,T^3). \label{eq:totalcomplexity}
\end{align}
Hence, on any event of the form $\{\tau \leq T\}$, the cumulative encoder cost up to decoding is bounded by $O(L_{\mr{tail}}^2\,T^3)$, or by $O(L_{\mr{tail}}\,T^2 \log(L_{\mr{tail}} T))$ with a heap. Since the submartingale structure of \lemref{lem:submar} together with Chebyshev's inequality yields $\Prob{\tau > cK} = O(1/K)$ for a channel-dependent constant $c$, the cumulative encoder cost up to decoding is $O(L_{\mr{tail}}^2\,K^3)$, or $O(L_{\mr{tail}}\,K^2 \log(L_{\mr{tail}} K))$ with a heap, with high probability.

For the BI-AWGN channel with a $B$-level quantizer, $L_{\mr{tail}} = O(B)$, so the high-probability cumulative bounds become $O(B^2 K^3)$ and $O(B K^2 \log(BK))$, respectively.

\textit{Memory.} The algorithm stores the interval lists $\{\mc{J}_k\}$ and the group counts $\{N_k\}$. The total storage is $O(G_t + F_t) = O(L_{\mr{tail}}\,t)$, which equals $O(L_{\mr{tail}}\,K)$ at the decoding time. No per-message storage is required; the algorithm never enumerates the $M = 2^K$ messages individually.

\section{Output Quantization for Continuous-Output BMS Channels}
\label{sec:quantization}

The low-complexity encoder in \secref{sec:algorithm} requires a lattice-valued LLR (\defnref{def:lattice}).  Discrete-output BMS channels satisfy this natively; continuous-output channels do not.  We enforce the lattice property through output quantization, designing the quantizer so that the induced LLR on each output symbol equals the lattice value $k\delta$ exactly.  This \emph{exact induced-lattice} constraint preserves the grouped-posterior structure of \secref{sec:lattice}.  Both encoder and decoder apply the same quantizer, so the scheme operates on the induced discrete channel without mismatch.  Prior work on output quantization for binary-input channels~\cite{KurkoskiYagi2014, Rave2009, SinghDabeerMadhow2009, BinshtokShamai1999} does not impose such a constraint.

\subsection{Quantizer Definition and Symmetry}
\label{sec:quant_def}

Recall the LLR statistic $\Lambda(y) \triangleq \log \frac{P_+(y)}{P_-(y)}$, and let $\mu_+$ and $\mu_-$ denote the law of $\Lambda(Y)$ under $Y \sim P_+$ and $Y \sim P_-$, respectively.  By definition of the LLR,
\begin{align}
    \frac{d\mu_-}{d\mu_+}(\lambda) = e^{-\lambda}, \qquad \frac{d\mu_+}{d\mu_-}(\lambda) = e^{\lambda}. \label{eq:mu_change}
\end{align}

Fix an odd integer $B = 2L + 1$, a spacing $\delta > 0$, and an integer $L_{\mr{tail}} \geq L$.  The quantizer outputs the symbol set $\mc{K}_B$ from \eqref{eq:kb}.

A \emph{symmetric $B$-level quantizer} is a measurable partition $\{I_k\}_{k \in \mc{K}_B}$ of $\mathbb{R}$ satisfying the antisymmetry condition
\begin{align}
    I_{-k} = -I_k, \qquad k \in \{0, 1, \dots, L{-}1, L_{\mr{tail}}\}, \label{eq:antisymmetry}
\end{align}
where $-I_k \triangleq \{-\lambda : \lambda \in I_k\}$.  The quantized output is
\begin{align}
    \widetilde{Y} = k \quad \Longleftrightarrow \quad \Lambda(Y) \in I_k, \qquad k \in \mc{K}_B, \label{eq:tildeY_def}
\end{align}
where $\mc{K}_B$ is given in \eqref{eq:kb}.
The antisymmetry \eqref{eq:antisymmetry} ensures that the induced channel $X \mapsto \widetilde{Y}$ is itself BMS.

In practice, the partition is specified by $L$ positive cutoffs
\begin{align}
    0 < a_1 < a_2 < \dots < a_L. \label{eq:cutoffs}
\end{align}
The bins are
\begin{align}
    I_0 &= [-a_1, a_1), \notag \\
    I_k &= [a_k, a_{k+1}), \qquad k = 1, \dots, L{-}1, \notag \\
    I_{L_{\mr{tail}}} &= [a_L, \infty), \qquad I_{-k} = -I_k. \label{eq:bins}
\end{align}

\begin{definition}[Exact induced-lattice constraint]
\label{def:exact_lattice}
The quantizer $\{I_k\}_{k \in \mc{K}_B}$ satisfies the \emph{exact induced-lattice constraint} with spacing $\delta$ if, for every $k \in \mc{K}_B$ with $\mu_+(I_k) > 0$,
\begin{align}
    \log \frac{\mu_+(I_k)}{\mu_-(I_k)} = k\delta. \label{eq:exact_constraint}
\end{align}
\end{definition}

\begin{lemma}[Conditional-expectation characterization]
\label{lem:condexp_char}
The constraint \eqref{eq:exact_constraint} holds if and only if
\begin{align}
    \mathbb{E}\!\left[e^{-\Lambda} \,\middle|\, \Lambda \in I_k\right] = e^{-k\delta} \label{eq:condexp_equiv}
\end{align}
for every $k \in \mc{K}_B$ with $\mu_+(I_k) > 0$.
\end{lemma}
\begin{IEEEproof}
By \eqref{eq:mu_change}, $\mu_-(I_k) = \int_{I_k} e^{-\lambda}\,d\mu_+(\lambda) = \mu_+(I_k)\,\mathbb{E}[e^{-\Lambda} \mid \Lambda \in I_k]$.  Dividing by $\mu_+(I_k)$ and taking logarithms gives the equivalence.
\end{IEEEproof}

\subsection{Induced Capacity and Cell-wise Loss}
\label{sec:induced_cap}

The capacity analysis rests on the scalar function
\begin{align}
    g(\lambda) \triangleq \log \frac{2 e^{\lambda}}{1 + e^{\lambda}}. \label{eq:g_def}
\end{align}
The symmetric capacity of the original BMS channel can be written as
\begin{align}
    C = \mathbb{E}[g(\Lambda)], \qquad \Lambda \sim \mu_+. \label{eq:C_as_g}
\end{align}

Under the exact induced-lattice constraint \eqref{eq:exact_constraint}, the LLR of the quantized symbol $\widetilde{Y} = k$ equals $k\delta$, so the induced capacity is
\begin{align}
    C_{B,\delta} \triangleq I(X; \widetilde{Y}) = \sum_{k \in \mc{K}_B} \mu_+(I_k) \, g(k\delta). \label{eq:C_Bd}
\end{align}
Subtracting \eqref{eq:C_Bd} from \eqref{eq:C_as_g} yields the exact cell-wise decomposition
\begin{align}
    C - C_{B,\delta} = \sum_{k \in \mc{K}_B} \mu_+(I_k) \Big( \mathbb{E}\!\big[g(\Lambda) \mid \Lambda \in I_k\big] - g(k\delta) \Big). \label{eq:loss_cellwise}
\end{align}

Since $g(\lambda) = \log 2 - \beta(e^{-\lambda})$ with $\beta(u) \triangleq \log(1 + u)$, the cell-wise loss can be written as
\begin{align}
    \mathbb{E}\!\big[g(\Lambda) \mid \Lambda \in I_k\big] - g(k\delta) = \beta(e^{-k\delta}) - \mathbb{E}\!\big[\beta(e^{-\Lambda}) \mid \Lambda \in I_k\big]. \label{eq:cell_loss_psi}
\end{align}
By \eqref{eq:condexp_equiv} and the concavity of $\beta$, each summand is non-negative (Jensen's inequality), so $C_{B,\delta} \leq C$.

To connect the quantizer to the grouped encoder of \secref{sec:algorithm}, we identify the quantized output label with the lattice increment and write $n_t \triangleq \widetilde{Y}_t \in \mc{K}_B$, so that the LLR increment at time $t$ equals $n_t \delta$ exactly.

\begin{lemma}[Derivatives of $g$]
\label{lem:g_derivatives}
The function $g$ in \eqref{eq:g_def} satisfies
\begin{align}
    g'(\lambda) &= \frac{1}{1 + e^{\lambda}}, \qquad g''(\lambda) = -\frac{e^{\lambda}}{(1 + e^{\lambda})^2}. \label{eq:gprimegpp}
\end{align}
Define $h(\lambda) \triangleq (1 + e^{\lambda})^{-2}$.  Then
\begin{align}
    g'(\lambda) + g''(\lambda) = h(\lambda), \label{eq:gpp_identity}
\end{align}
and
\begin{align}
    \sup_{\lambda \in \mathbb{R}} |g'''(\lambda)| = \frac{\sqrt{3}}{18}, \qquad \sup_{\lambda \in \mathbb{R}} |h'(\lambda)| = \frac{8}{27}. \label{eq:sup_bounds}
\end{align}
\end{lemma}
\begin{IEEEproof}
Direct differentiation yields \eqref{eq:gprimegpp} and \eqref{eq:gpp_identity}.  The extrema in \eqref{eq:sup_bounds} follow by writing $u = e^{\lambda}$ and optimizing over $u > 0$.
\end{IEEEproof}

\subsection{Capacity-Loss Expansion}
\label{sec:caploss_expansion}

The decomposition \eqref{eq:loss_cellwise} shows that the capacity loss is controlled by how well each representative value $k\delta$ captures the conditional distribution of $\Lambda$ within $I_k$.  To make this precise, we impose regularity conditions on the interior cells and bound the tail contribution separately.

Let $\mc{K}_B^{\circ} \triangleq \{0, \pm 1, \dots, \pm(L{-}1)\}$ denote the interior labels.  For each $k \in \mc{K}_B^{\circ}$, define the centered residual 
\begin{align}
    \xi_k \triangleq \Lambda - k\delta
\end{align}
on $I_k$.  Assume there exist constants $\Gamma < \infty$ and $M_3 < \infty$, and a sequence $\epsilon_B \to 0$, such that for all sufficiently large~$B$:
\begin{enumerate}[label=\textnormal{(R\arabic*)}]
    \item \label{cond:R1} \emph{Interior localization:} $|\xi_k| \leq \Gamma \delta$ on $I_k$ for every $k \in \mc{K}_B^{\circ}$.
    \item \label{cond:R2} \emph{Uniform second-moment approximation:}
    \begin{align}
        \left| \mathbb{E}\!\left[\xi_k^2 \mid \Lambda \in I_k\right] - \frac{\delta^2}{12} \right| \leq \epsilon_B \cdot \frac{\delta^2}{12}, \quad k \in \mc{K}_B^{\circ}. \label{eq:R2}
    \end{align}
    \item \label{cond:R3} \emph{Uniform third-moment bound:}
    \begin{align}
        \mathbb{E}\!\left[|\xi_k|^3 \mid \Lambda \in I_k\right] \leq M_3 \delta^3, \quad k \in \mc{K}_B^{\circ}. \label{eq:R3}
    \end{align}
\end{enumerate}
Conditions \ref{cond:R1}--\ref{cond:R3} require that the interior bins behave approximately like uniform intervals of width $\delta$; the tail cells are handled separately via $\Psi(a_L)$.

Define the channel constant
\begin{align}
    \kappa_{P_{Y|X}} \triangleq \mathbb{E}[h(\Lambda)] = \mathbb{E}\!\left[\frac{1}{(1 + e^{\Lambda})^2}\right], \label{eq:kappa_def}
\end{align}
and the two-sided tail probability  
\begin{align}
    \Psi(t) \triangleq \mu_+(|\Lambda| \geq t)
\end{align}
for $t \geq 0$.

\begin{theorem}[Capacity-loss expansion]
\label{thm:caploss_expansion}
Under conditions \ref{cond:R1}--\ref{cond:R3}, if additionally $\Psi(a_L) = o(\delta^2)$, then
\begin{align}
    C - C_{B,\delta} = \frac{\kappa_{P_{Y|X}}}{24} \, \delta^2 + o(\delta^2). \label{eq:caploss_expansion}
\end{align}
\end{theorem}
The leading term $\frac{\kappa_{P_{Y|X}}}{24}\,\delta^2$ is governed by the channel-dependent constant $\kappa_{P_{Y|X}}$; the tail labels $\pm L_{\mr{tail}}\delta$ enter only through $\Psi(a_L)$.  A good asymptotic design balances the mesh size $\delta$ against the tail decay of the LLR distribution.

\begin{IEEEproof}
See \appref{app:caploss_proof}.
\end{IEEEproof}

\subsection{BI-AWGN Specialization}
\label{sec:quant_awgn}

We now specialize the framework to the BI-AWGN channel $Y = X + Z$ with $Z \sim \mc{N}(0, \sigma^2)$ and $X \in \{+1, -1\}$.

\subsubsection{LLR distribution and channel constant}
The LLR is $\Lambda(Y) = 2Y/\sigma^2$.  Under $X = +1$,
\begin{align}
    \Lambda \sim \mc{N}\!\left(\frac{2}{\sigma^2},\, \left(\frac{2}{\sigma}\right)^{\!2}\right). \label{eq:awgn_llr_dist}
\end{align}
Writing $\mu \triangleq 2/\sigma^2$ and $s \triangleq 2/\sigma$, we have $\Lambda = \mu + s\,G$ with $G \sim \mc{N}(0,1)$.  The channel constant \eqref{eq:kappa_def} becomes
\begin{align}
    \kappa_{\mr{AWGN}} = \mathbb{E}\!\left[\frac{1}{\big(1 + e^{\mu + s\,G}\big)^2}\right], \qquad G \sim \mc{N}(0,1), \label{eq:AAWGN}
\end{align}
which can be computed by numerical quadrature.  The tail function satisfies, for $t \geq \mu$,
\begin{align}
    \Psi(t) \leq 2\!\left(1 - \Phi\!\left(\frac{t - \mu}{s}\right)\right). \label{eq:awgn_tail_bound_phi}
\end{align}

It is useful to record the induced LLR of the positive tail cell
\begin{align}
    \ell_{\mr{tail}}(t) \triangleq \log \frac{1 - \Phi\!\left(\frac{t - \mu}{s}\right)}{1 - \Phi\!\left(\frac{t + \mu}{s}\right)}, \label{eq:tail_llr_def}
\end{align}
together with the expansions
\begin{align}
    \ell_{\mr{tail}}(t) = t + \frac{2\mu}{t} + O(t^{-2}), \qquad \ell_{\mr{tail}}'(t) = 1 - \frac{2\mu}{t^2} + O(t^{-3}). \label{eq:tail_llr_expansion}
\end{align}

\subsubsection{Exact-lattice construction and asymptotic scaling}

We construct an exact-lattice BI-AWGN quantizer with $B = 2L + 1$ output symbols.  The interior cells carry the contiguous LLR values $0, \pm\delta_B, \dots, \pm(L{-}1)\delta_B$, while the tail cells carry $\pm L_{\mr{tail}}\delta_B$ for an integer $L_{\mr{tail}} \geq L$ to be determined.

\paragraph{Choice of $T_B$ and $L_{\mr{tail}}$}
Set $T_B^{\star} \triangleq \mu + \frac{4}{\sigma}\sqrt{\log B}$ and $J_B(t) \triangleq (L - \tfrac{1}{2}) \frac{\ell_{\mr{tail}}(t)}{t}$.  By \eqref{eq:tail_llr_expansion}, $J_B(t) = L - \tfrac{1}{2} + O(L/t^2)$ and $J_B'(t) = O(L/t^3)$, so on $[T_B^{\star}, T_B^{\star}+1]$ the image of $J_B$ has length exceeding $1$ for large $B$.  By continuity, there exist
\begin{align}
    T_B \in [T_B^{\star},\, T_B^{\star} + 1] \qquad \text{and} \qquad L_{\mr{tail}} \in \mathbb{Z},\quad L_{\mr{tail}} \geq L, \label{eq:TB_Ltail_exist}
\end{align}
such that $J_B(T_B) = L_{\mr{tail}}$.  Define $\delta_B \triangleq T_B/(L - \frac{1}{2})$.  By construction, $\ell_{\mr{tail}}(T_B) = L_{\mr{tail}}\,\delta_B$, so both tail cells satisfy the exact induced-lattice constraint.

\paragraph{Interior thresholds}
Set $a_L \triangleq T_B$.  For $k = L{-}1, L{-}2, \dots, 1$, let $a_k$ be the unique solution of
\begin{align}
    \ell(a_k, a_{k+1}) = k\,\delta_B, \label{eq:recursive_ak}
\end{align}
where $\ell(u, v) \triangleq \log \frac{\mu_+([u,v))}{\mu_-([u,v))}$ is the interval LLR.  The negative thresholds are set by symmetry, and the bins are given by \eqref{eq:bins} with $\delta = \delta_B$.  Uniqueness and existence of the solution in \eqref{eq:recursive_ak} follow from the strict monotonicity of $u \mapsto \ell(u, v)$.  By construction, every cell satisfies the exact induced-lattice constraint.

The following lemma shows that the resulting thresholds are approximately equally spaced.
\begin{lemma}[Threshold geometry]
\label{prop:threshold_geometry}
If $(1 + T_B)\,\delta_B \to 0$ as $B \to \infty$, then uniformly in $k = 1, \dots, L$,
\begin{align}
    a_k = \left(k - \tfrac{1}{2}\right)\delta_B + O\!\big((1 + T_B)\,\delta_B^2\big). \label{eq:ak_expansion}
\end{align}
Consequently, the bin midpoints and bin widths satisfy, uniformly for $k = 1, \dots, L{-}1$,
\begin{align}
    m_k \triangleq \frac{a_k + a_{k+1}}{2} &= k\,\delta_B + O\!\big((1 + T_B)\,\delta_B^2\big), \label{eq:mk_expansion} \\
    w_k \triangleq a_{k+1} - a_k &= \delta_B + O\!\big((1 + T_B)\,\delta_B^2\big). \label{eq:wk_expansion}
\end{align}
\end{lemma}
\begin{IEEEproof}
    See \appref{app:proofprop}.
\end{IEEEproof}

\subsubsection{Verification of the regularity conditions}
By \lemref{prop:threshold_geometry}, the interior bins have width $\delta_B(1 + O((1+T_B)\delta_B))$ and midpoints $k\delta_B + O((1+T_B)\delta_B^2)$, so \ref{cond:R1} holds with $\Gamma > \frac{1}{2}$ and \ref{cond:R3} follows from $|\xi_k| = O(\delta_B)$. The Gaussian smoothness estimate in the proof of \lemref{prop:threshold_geometry} shows that the conditional density on each bin is a $1 + O((1+T_B)\delta_B)$ perturbation of the uniform density, giving \ref{cond:R2} with $\epsilon_B = O((1+T_B)\delta_B) \to 0$.

\subsubsection{Asymptotic scaling law}

\begin{corollary}
\label{cor:awgn_theta}
For the exact-lattice BI-AWGN quantizer constructed above,
\begin{align}
    T_B \in \left[\mu + \frac{4}{\sigma}\sqrt{\log B},\, \mu + \frac{4}{\sigma}\sqrt{\log B} + 1\right], \label{eq:TB_location_final}
\end{align}
and $\delta_B = \frac{T_B}{L - \frac{1}{2}} = \frac{8}{\sigma}\,\frac{\sqrt{\log B}}{B}\,(1 + o(1))$.  Moreover, $(1 + T_B)\,\delta_B = O(\log B / B) \to 0$ and $\Psi(T_B) = O\!\big(1/(B^2\sqrt{\log B})\big) = o(\delta_B^2)$.  Consequently, for fixed $\sigma^2 < \infty$,
\begin{align}
    C - C_{B,\delta_B} = \frac{8\,\kappa_{\mr{AWGN}}}{3\sigma^2} \cdot \frac{\log B}{B^2} \cdot (1 + o(1)) = \Theta\!\left(\frac{\log B}{B^2}\right). \label{eq:awgn_scaling}
\end{align}
\end{corollary}

\begin{IEEEproof}
The interval \eqref{eq:TB_location_final} restates \eqref{eq:TB_Ltail_exist}.  Since $L - \frac{1}{2} = (B - 1)/2$,
\begin{align}
    \delta_B = \frac{T_B}{(B-1)/2} = \frac{8}{\sigma}\,\frac{\sqrt{\log B}}{B}\,(1 + o(1)). \label{eq:deltaB_asymptotic}
\end{align}
Combining \eqref{eq:TB_location_final} and \eqref{eq:deltaB_asymptotic} gives $(1 + T_B)\,\delta_B = O(\log B / B) \to 0$, so \lemref{prop:threshold_geometry} applies and conditions \ref{cond:R1}--\ref{cond:R3} hold with $\epsilon_B \to 0$.

The tail probability is bounded via \eqref{eq:awgn_tail_bound_phi} at $t = T_B$.  Since $(T_B - \mu)/s \geq 2\sqrt{\log B}$, the Gaussian tail bound $1 - \Phi(u) \leq \frac{1}{u\sqrt{2\pi}}\,e^{-u^2/2}$ gives
\begin{align}
    \Psi(T_B) = O\!\left(\frac{1}{B^2\sqrt{\log B}}\right) = o(\delta_B^2). \label{eq:tail_negligible}
\end{align}
Since $\Psi(T_B) = o(\delta_B^2)$, the hypotheses of \thmref{thm:caploss_expansion} are satisfied, giving
\begin{align}
    C - C_{B,\delta_B} = \frac{\kappa_{\mr{AWGN}}}{24}\,\delta_B^2\,(1 + o(1)).
\end{align}
Substituting $\delta_B^2 = \frac{64}{\sigma^2}\,\frac{\log B}{B^2}\,(1 + o(1))$ completes the proof.
\end{IEEEproof}

\figref{fig:caploss_verify} compares the capacity loss of the exact-lattice BI-AWGN quantizer against the analytical predictions of \thmref{thm:caploss_expansion} and \corref{cor:awgn_theta}.  For each odd $B$ from $3$ to $101$, the optimal spacing $\delta^*$ and tail label $L_{\mr{tail}}^*$ are computed numerically via \algref{alg:exactlattice_awgn_inner}, and the exact induced capacity loss $C_{\mr{AWGN}} - C_{B,\delta^*}$ is evaluated (blue curve).  The red curve plots the leading-order term $\frac{\kappa_{\mr{AWGN}}}{24}\,\delta^{*2}$ from~\eqref{eq:caploss_expansion}; by $B = 51$ this accounts for over $90\%$ of the loss, confirming that the higher-order remainder is negligible for moderate $B$.  The green curve shows the fully asymptotic approximation from~\eqref{eq:awgn_scaling}, which overestimates the loss at finite $B$ but captures the correct $O(\log B / B^2)$ decay rate.
 
\begin{figure}[t]
    \centering
    \includegraphics[width=\columnwidth]{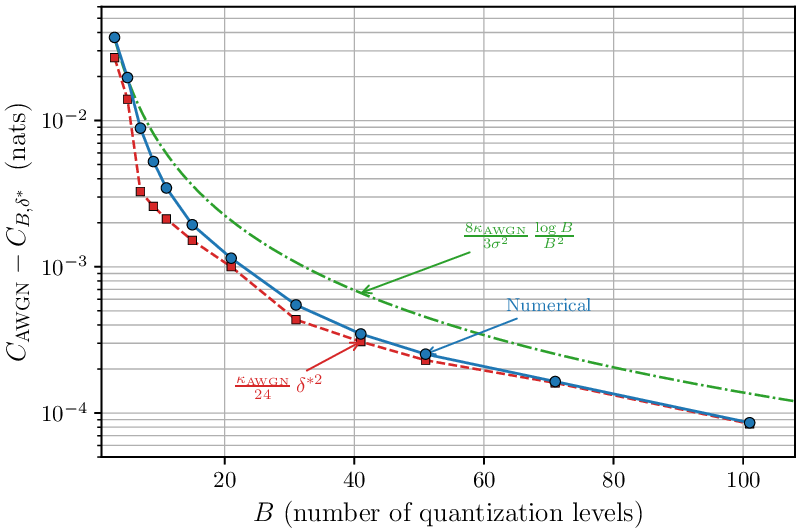}
    \caption{Capacity loss $C_{\mr{AWGN}} - C_{B,\delta^*}$ (in nats) of the exact-lattice BI-AWGN quantizer ($\sigma^2 = 1$) as a function of the number of output levels~$B$.}
    \label{fig:caploss_verify}
\end{figure}

\subsubsection{Numerical design algorithm}
\label{sec:awgn_numeric}

Since $\Lambda(Y) = 2Y/\sigma^2$ is strictly monotone in $Y$, the quantizer can equivalently be specified through thresholds in $Y$-space.  By the antisymmetry, only the positive thresholds $0 < b_{1/2} < b_{3/2} < \dots < b_{L-1/2}$ need to be determined; the negative thresholds are $b_{-(k+1/2)} = -b_{k+1/2}$.  The bins in $Y$-space are
\begin{align}
    \widetilde{\mc{B}}_0 &= [-b_{1/2},\, b_{1/2}), \notag \\
    \widetilde{\mc{B}}_k &= [b_{k-1/2},\, b_{k+1/2}), \quad 1 \leq k \leq L{-}1, \notag \\
    \widetilde{\mc{B}}_L &= [b_{L-1/2},\, \infty), \qquad \widetilde{\mc{B}}_{-k} = -\widetilde{\mc{B}}_k.
\end{align}

Define the Gaussian bin probabilities
\begin{align}
    P^+(u,v) &\triangleq \Phi\!\left(\tfrac{v-1}{\sigma}\right) - \Phi\!\left(\tfrac{u-1}{\sigma}\right), \\
    P^-(u,v) &\triangleq \Phi\!\left(\tfrac{v+1}{\sigma}\right) - \Phi\!\left(\tfrac{u+1}{\sigma}\right),
\end{align}
and the bin LLR $\ell(u,v) \triangleq \log \frac{P^+(u,v)}{P^-(u,v)}$, with the tail variant $\ell_{\mr{tail}}(u) \triangleq \log \frac{P^+(u,\infty)}{P^-(u,\infty)}$.  The exact-lattice conditions become
\begin{align}
    \ell(b_{k-1/2},\, b_{k+1/2}) &= k\delta, \quad k = 1, \dots, L{-}1, \label{eq:exact_interior_awgn} \\
    \ell_{\mr{tail}}(b_{L-1/2}) &= L_{\mr{tail}}\,\delta, \label{eq:exact_tail_awgn}
\end{align}
which is a system of $L$ equations in the $L$ unknowns $b_{1/2}, \dots, b_{L-1/2}$, for fixed $\delta$ and $L_{\mr{tail}}$.

\emph{Bidiagonal structure and design algorithm (\algref{alg:exactlattice_awgn_inner}).}  Since $f_+(y)/f_-(y) = e^{2y/\sigma^2}$ is strictly increasing, each $\ell(u,v)$ is strictly increasing in $u$ and decreasing in $v$, so the system has a bidiagonal structure that permits backward sequential root-finding: $b_{L-1/2}$ is found from the tail equation, then $b_{k-1/2}$ from $\ell(u,\, b_{k+1/2}) = k\delta$ for $k = L{-}1, \dots, 1$, each by bracketed bisection. The outer loop enumerates a small range of integer $L_{\mr{tail}}$ candidates; for each, a golden-section search maximizes $C_{B,\delta}$ over $\delta$.

\begin{algorithm}[t]
\caption{\textsc{ExactLatticeAWGN}: Joint optimization and sequential bisection}
\label{alg:exactlattice_awgn_inner}
\label{alg:exactlattice_awgn_outer}
\begin{algorithmic}[1]
\REQUIRE $B = 2L+1$, $\sigma^2$, search interval $(\delta_{\min}, \delta_{\max})$, search radius $r$, tolerances $\varepsilon, \varepsilon_\delta$.
\ENSURE Optimized $\delta^*$, $L_{\mr{tail}}^*$, thresholds $\mb{b}^*$, induced capacity $C_{B,\delta^*}$.
\STATE $T_B^{\star} \leftarrow \frac{2}{\sigma^2} + \frac{4}{\sigma}\sqrt{\log B}$; \, $L_{\mr{tail}}^{\star} \leftarrow \mr{round}\!\left((L - \tfrac{1}{2})\,\frac{\ell_{\mr{tail}}(T_B^{\star})}{T_B^{\star}}\right)$; \, $C_{\mr{best}} \leftarrow 0$.
\FOR{$L_{\mr{tail}} = \max(L,\, L_{\mr{tail}}^{\star} - r)$ \textbf{to} $L_{\mr{tail}}^{\star} + r$}
    \FOR{each $\delta$ queried by golden-section search on $(\delta_{\min}, \delta_{\max})$}
        \STATE \COMMENT{\textit{Inner solve: sequential bisection at fixed $(\delta, L_{\mr{tail}})$}}
        \STATE Solve $\ell_{\mr{tail}}(u) = L_{\mr{tail}}\,\delta$ for $b_{L-1/2}$ by bracketed bisection.
        \FOR{$k = L{-}1, L{-}2, \dots, 1$}
            \STATE Solve $\ell(u,\, b_{k+1/2}) = k\,\delta$ for $b_{k-1/2}$ by bracketed bisection.
        \ENDFOR
        \STATE Evaluate $C_{B,\delta}$; update $C_{\mr{best}}$, $\delta^*$, $L_{\mr{tail}}^*$, $\mb{b}^*$ if improved.
    \ENDFOR
\ENDFOR
\end{algorithmic}
\end{algorithm}

\section{Extensions and Discussion} \label{sec:extension}

\paragraph{Non-uniform message distributions and joint source-channel coding}
The analysis leading to \thmref{thm:main} depends on $M$ only through the initial log-posterior odds. Inspecting the proof, the term $\frac{\log(M-1)}{C}$ in \eqref{eq:N0def} may be replaced by $\frac{H(\boldsymbol{\rho}(0))}{C}$, where $H(\boldsymbol{\rho}(0))$ is the Shannon entropy of the initial posterior vector $\boldsymbol{\rho}(0)$. The bound then applies directly to joint source-channel coding, as in \cite{antonini2024} for the BSC.

\paragraph{The total number of feedback bits for the BI-AWGN channel} By \corref{cor:awgn_theta}, the quantization penalty on $\log M$ scales as $\Theta(N \log B / B^2)$. Combining with \corref{cor:asymptotic}, the expansion \eqref{eq:asymp} is achievable using $B = \Theta(\sqrt{N}/\log N)$ levels, i.e., $\Theta(\log K)$ feedback bits per channel use and $\Theta(K \log K)$ in total. Designing a scheme where the quantization level varies over time and the total feedback budget is strictly less than $\Theta(K \log K)$ remains open. 

\paragraph{Role of the BMS symmetry}
The supporting results (\secref{sec:supporting}), the SED rule, and the low-complexity encoder (\secref{sec:algorithm}) do not require BMS symmetry. BMS symmetry simplifies \thmref{thm:main} by collapsing two overshoot parameters into one ($\bar{\Lambda} \overset{d}{=} -\Lambda$) and yielding the improved ruin bound $\psi(0) \leq \chi(P_{Y|X})$; for asymmetric channels with bounded LLR, these can be replaced by additional constants~\cite{yang2022SED}. The most critical use is the surrogate submartingale in \secref{sec:commbound}: for continuous-output channels, $C_2 = \infty$, and the BMS cancellation yielding $\log 2$ is what makes the communication-phase bound finite. Extending this to asymmetric channels with unbounded LLR remains open.

\paragraph{Summary}
We have developed a VLF coding framework for BMS channels that combines a non-asymptotic analysis of SED-based posterior matching with a practical encoder construction. The analysis removes the bounded-LLR-increment assumption present in all prior SED-based works, achieving the second-order term $-\frac{C}{C_1}\log N$ with an $O(1)$ third-order remainder. On the algorithmic side, a TOP initialization followed by batched SED repair maintains the exact SED partition at polynomial cost in the number of transmitted bits; for the quantized BI-AWGN construction, the cumulative encoder cost is $O(BK^2 \log(BK))$ with high probability. The output-quantization framework maps continuous-output channels to lattice-valued ones with capacity loss $O(\log B / B^2)$, making the scheme applicable to BI-AWGN and other BMS channels. Open directions include extending the surrogate submartingale construction to asymmetric channels with unbounded LLR and reducing the total feedback budget below $\Theta(K \log K)$ bits.

    \appendices
\section{Proof of \corref{cor:asymptotic}} \label{app:cor}
We optimize the bound in \eqref{eq:Nach} over $\epsilon_0$. Substituting $\delta \triangleq \epsilon - \epsilon_0 > 0$ and using $\frac{1-\epsilon'}{\epsilon'} = \frac{1-\epsilon}{\delta}$, the bound becomes
\begin{align}
    &N^*(M, \epsilon) \leq (1{-}\epsilon{+}\delta) \notag \\
    &\quad \cdot \!\left[\frac{\log M {+} \log 2}{C} + \frac{\log \frac{1-\epsilon}{\delta} + \eta(P_\Lambda)}{C_1} + \zeta(\delta)\right], \label{eq:Nbound_delta}
\end{align}
where
\begin{align}
    \zeta(\delta) 
    &\triangleq \frac{\min\!\big\{\alpha(\epsilon'(\delta)) \cdot \eta(P_\Lambda),\, \chi(P_{Y|X}) \cdot \eta_0^{-}(P_\Lambda)\big\}}{(1-\chi(P_{Y|X})) \, C} \notag \\
    &\quad + \frac{\chi(P_{Y|X}) \log 2}{(1-\chi(P_{Y|X})) \, C}\label{eq:zetadef}
\end{align}
with $\epsilon'(\delta) = \delta / (1 - \epsilon + \delta)$. Since $\alpha(\epsilon') \to \chi(P_{Y|X})$ as $\epsilon' \to 0$, we have $\zeta(\delta) \to \zeta_0$ as $\delta \to 0$, where $\zeta_0$ is obtained by replacing $\alpha(\epsilon')$ with $\chi(P_{Y|X})$ in \eqref{eq:zetadef}; the error is $O(\sqrt{\delta}) = o(1)$ since $\delta = O(N^{-1})$.  Replacing $\zeta(\delta)$ by $\zeta_0$ in \eqref{eq:Nbound_delta}, differentiating with respect to $\delta$, and setting the derivative to zero yields
\begin{align}
    &\frac{\log M + \log 2}{C} + \frac{\log \frac{1-\epsilon}{\delta} + \eta(P_\Lambda)}{C_1} + \zeta_0 = \frac{1-\epsilon+\delta}{C_1 \delta}. \label{eq:optcond}
\end{align}
The left-hand side of \eqref{eq:optcond} equals $\frac{N}{1-\epsilon+\delta}$, giving $N = \frac{(1-\epsilon + \delta)^2}{C_1 \delta}$, a quadratic in $\delta$ whose smaller root satisfies
\begin{align}
    \delta^* = \frac{(1-\epsilon)^2}{NC_1} + O(N^{-2}), \label{eq:deltastar}
\end{align}
which translates to \eqref{eq:eps0star}.

To obtain $\log M$ as a function of $N$, we solve \eqref{eq:optcond} for $\log M$ and get
\begin{align}
    &\frac{\log M + \log 2}{C} = \frac{1-\epsilon+\delta}{C_1\delta}  - \frac{\log\frac{1-\epsilon}{\delta} + \eta(P_\Lambda)}{C_1} - \zeta_0.\label{eq:logMexpr}
\end{align}
Substituting \eqref{eq:deltastar} into each term on the right-hand side:
\begin{align}
    \frac{1-\epsilon+\delta^*}{C_1 \delta^*} &= \frac{N}{1-\epsilon} - \frac{1}{C_1} + O(N^{-1}), \label{eq:term1} \\
    \log \frac{1-\epsilon}{\delta^*} &= \log N + \log \frac{C_1}{1-\epsilon} + O(N^{-1}). \label{eq:term2}
\end{align}
Substituting \eqref{eq:term1}--\eqref{eq:term2} into \eqref{eq:logMexpr}, multiplying through by $C$, and collecting terms yields \eqref{eq:Keps} and hence \eqref{eq:asymp}.

\section{Proof of Remark~\ref{rem:log2}} \label{app:log2}
In \eqref{eq:log2rho}, we have a channel-independent bound $w_m'(t, y) \leq \log 2$ for all $y \in \mathbb{R}$, which is used in \thmref{thm:main} and \corref{cor:asymptotic}. We here improve the universal bound for the BSC and BI-AWGN channels. 

By symmetry, we assume $x_m = +$ is transmitted. For the BSC, we have
\begin{align}
    w_m'(t, -1) \leq w_m'(t, 1) = (1 + r) \log \frac{2}{1 + r},
\end{align}
where $r = \frac{P_{-}(1)}{P_{+}(1)} = \frac{p}{1-p}$, proving the claim for the BSC.

For the BI-AWGN channel, we use \thmref{thm:hitting_time_drift_overshoot} on the increment $w_m'(t, Y)$ and get
\begin{align}
    \eta_0 = \E{w_m'(t, Y) \middle | Y \geq 0},
\end{align}
where $Y \sim P_+$.
Define $r(y) \triangleq \frac{P_-(y)}{P_+(y)} = e^{-\frac{2y}{\sigma^2}}$. Then,
\begin{align}
    w_m'(t, Y) = (1 + r(Y)) \log \frac{2}{1 + r(Y)}.
\end{align}
Since the function $\varphi(u) = (1 + u) \log \frac{2}{1 + u}$ is concave, applying Jensen's inequality, we get
\begin{align}
    \E{w_m'(t, Y) | Y \geq 0} \leq \varphi(\E{r(Y) | Y \geq 0}). \label{eq:Eg}
\end{align}
We have
\begin{align}
    P_+[Y \geq 0] &= \Phi\left(\frac{1}{\sigma}\right), \\
    \E{r(Y) 1\{Y \geq 0\}} &= P_-[Y \geq 0] = \Phi\left(-\frac{1}{\sigma}\right), 
\end{align}
giving $\E{r(Y) | Y \geq 0} = \frac{\Phi\left(-\frac{1}{\sigma}\right)}{\Phi\left(\frac{1}{\sigma}\right)}$. Substituting this into \eqref{eq:Eg} completes the proof.

    \section{Proof of \thmref{thm:hitting_time_drift_overshoot}} \label{app:stopUniform}
We split the proof into two steps.

\emph{Step 1: A uniform bound on the expected overshoot (truncated in time).}
For $N\in\mathbb{N}$, define the truncated overshoot
$R_b^{(N)} \triangleq (S_{\tau_b}-b)\,1\{\tau_b\le N\}$.
On $\{\tau_b=k\}$, $S_{k-1}\le b$ and $S_k>b$, so $S_k - b = \Delta S_k - (b-S_{k-1})$ with $\Delta S_k > b-S_{k-1}\ge 0$. Therefore,
\begin{align}
R_b^{(N)}
=
\sum_{k=1}^N \big(\Delta S_k-(b-S_{k-1})\big)^+\,1\{\tau_b=k\}.
\end{align}
Taking expectation and conditioning on $\mathcal{F}_{k-1}$ yields
\begin{align}
\mathbb{E}\!\left[R_b^{(N)}\right] =
\sum_{k=1}^N
\mathbb{E}\!\left[
1\{\tau_b=k\}\,
\mathbb{E}\!\left[
\big(\Delta S_k-(b-S_{k-1})\big)^+ \,\middle|\, \mathcal{F}_{k-1}
\right]
\right].
\label{eq:overshoot_conditionalize}
\end{align}
Fix $k$. Let $Y_{k-1}\triangleq b-S_{k-1}$, which is $\mathcal{F}_{k-1}$-measurable.
Using \eqref{eq:A2_mean_excess} and the identity
\begin{align}
&\mathbb{E}\!\left[(\Delta S_k-Y_{k-1})^+ \mid \mathcal{F}_{k-1}\right] \notag \\
&=
\mathbb{P}\!\left[\Delta S_k>Y_{k-1}\mid \mathcal{F}_{k-1}\right] \notag \\
&\quad \cdot \mathbb{E}\!\left[\Delta S_k-Y_{k-1}\,\middle|\,\mathcal{F}_{k-1},\,\Delta S_k>Y_{k-1}\right],
\end{align}
we obtain
\begin{align}
&\mathbb{E}\!\left[(\Delta S_k-Y_{k-1})^+ \mid \mathcal{F}_{k-1}\right] \\
&\le
\eta_0\,
\mathbb{P}\!\left[\Delta S_k>Y_{k-1}\mid \mathcal{F}_{k-1}\right] \\
&\le
\eta_0,
\qquad \text{a.s.}
\label{eq:pospart_bound}
\end{align}
Substituting \eqref{eq:pospart_bound} into \eqref{eq:overshoot_conditionalize} gives
\begin{align}
\mathbb{E}\!\left[R_b^{(N)}\right]\le
\sum_{k=1}^N \eta_0\,\mathbb{P}[\tau_b=k]
=
\eta_0\,\mathbb{P}[\tau_b\le N]
\le
\eta_0.
\label{eq:overshoot_trunc_bound}
\end{align}

\emph{Step 2: Optional sampling at $\tau_b\wedge N$ and passage to the limit.}
Define
\begin{align}
T_N \triangleq \tau_b \wedge N.
\end{align}
By \eqref{eq:A1_drift}, the process
\begin{align}
M_n \triangleq S_n - \mu n
\end{align}
is an integrable submartingale.
Since $T_N\le N$ almost surely, Doob's optional sampling theorem for bounded stopping times
(see, e.g., Doob~\cite{Doob1953} or Williams~\cite[Sec.~10.9]{Williams1991})
yields
\begin{align}
\mathbb{E}[M_{T_N}] \ge \mathbb{E}[M_0],
\end{align}
equivalently,
\begin{align}
\mathbb{E}[S_{T_N}]
\ge
S_0 + \mu\,\mathbb{E}[T_N].
\label{eq:OST_lower}
\end{align}

Next, we upper bound $\mathbb{E}[S_{T_N}]$.
On $\{\tau_b>N\}$, we have $S_N\le b$, hence $S_{T_N}=S_N\le b$.
On $\{\tau_b\le N\}$, we have $S_{T_N}=S_{\tau_b}=b+R_b^{(N)}$.
Therefore, pathwise,
\begin{align}
S_{T_N}
&\le
b + R_b^{(N)}.
\end{align}
Taking expectations and using \eqref{eq:overshoot_trunc_bound} gives
\begin{align}
\mathbb{E}[S_{T_N}]
\le
b + \mathbb{E}[R_b^{(N)}]
\le
b + \eta_0.
\label{eq:upper_STN}
\end{align}
Combine \eqref{eq:OST_lower} and \eqref{eq:upper_STN}:
\begin{align}
S_0 + \mu\,\mathbb{E}[T_N]
\le
b+\eta_0,
\end{align}
hence
\begin{align}
\mathbb{E}[T_N]
\le
\frac{b-S_0+\eta_0}{\mu},
\qquad \forall N\in\mathbb{N}.
\label{eq:TN_bound}
\end{align}

If $S_0>b$, then $\tau_b=1$ and \eqref{eq:Etau_bound} is trivial.
Assume $S_0\le b$, so that $\tau_b\ge 1$ and $T_N\uparrow \tau_b$ pointwise.
By the monotone convergence theorem,
\begin{align}
\mathbb{E}[\tau_b]
=
\lim_{N\to\infty}\mathbb{E}[T_N]
\le
\frac{b-S_0+\eta_0}{\mu}.
\end{align}
This proves \eqref{eq:Etau_bound} (with $(b-S_0)^+=b-S_0$ in this case).
Finally, since $\mathbb{E}[\tau_b]<\infty$, we must have $\mathbb{P}[\tau_b<\infty]=1$.

\section{Proof of \thmref{thm:spitzer-lundberg-sandwich}} \label{app:ruin}
\begin{IEEEproof}
(i) This is a classical consequence of Spitzer's identity;
see \cite{spitzer1956,wendel1958}.

(ii) The rightmost inequality $\psi(x)\le\psi(0)$ is the monotonicity
of~$\psi$: the coupling $S_n^{(x)}=x+\sum_{k=1}^n X_k$ shows that
$\psi(x)\le\psi(y)$ whenever $x\ge y$.

The upper bound $p_{\mathrm{ruin}}(x,A)\le\psi(x)$ is immediate:
$\{\tau_-<\tau_A\}\subseteq\{\tau_-<\infty\}$, so
$p_{\mathrm{ruin}}(x,A)\le\psi(x)$.

For the lower bound, partition $\{\tau_-<\infty\}$ into
$\{\tau_-<\tau_A\}$ and $\{\tau_A<\tau_-\}$ to obtain
\begin{align}
\psi(x)-p_{\mathrm{ruin}}(x,A)
&=\mathbb{P}_x[\tau_-<\infty,\ \tau_A<\tau_-].
\end{align}
By the strong Markov property at $\tau_A$,
\begin{align}
\psi(x)-p_{\mathrm{ruin}}(x,A)
&=\mathbb{E}_x\big[1\{\tau_A<\tau_-\}\,\psi(S_{\tau_A})\big]
\le \psi(A),
\end{align}
where the inequality uses $S_{\tau_A}\ge A$ and monotonicity of $\psi$.
Rearranging yields the lower bound in~\eqref{eq:sandwich}.

(iii) Define $M_n\triangleq e^{-R S_n}$. Since $\mathbb{E}[e^{-R X_1}]=1$,
the process $(M_n)_{n\ge 0}$ is a nonnegative martingale under~$\mathbb{P}_x$
with $M_0=e^{-Rx}$. For each $N\ge 1$, the stopping time $\tau_-\wedge N$
is bounded, so Doob's optional stopping theorem
\cite[Th.~10.9]{Williams1991} gives
\begin{align}
\mathbb{E}_x\big[M_{\tau_-\wedge N}\big]=e^{-Rx}.
\end{align}
On $\{\tau_-\le N\}$, we have $S_{\tau_-}<0$, so
\begin{align}
    M_{\tau_-}=e^{-R S_{\tau_-}}\ge 1.
\end{align}
Dropping the nonnegative contribution
on $\{\tau_->N\}$ and using $M_{\tau_-}\ge 1$ on $\{\tau_-\le N\}$,
\begin{align}
e^{-Rx}
=\mathbb{E}_x\big[M_{\tau_-\wedge N}\big]
\ge\mathbb{E}_x\big[1\{\tau_-\le N\}\,M_{\tau_-}\big]
\ge\mathbb{P}_x[\tau_-\le N].
\end{align}
Letting $N\to\infty$, by monotone convergence of the events
$\{\tau_-\le N\}\nearrow\{\tau_-<\infty\}$, we obtain
$\psi(x)\le e^{-Rx}$, which is~\eqref{eq:lundberg}. The
bound~\eqref{eq:fb-bracket} follows by substituting~\eqref{eq:lundberg}
into the left-hand side of~\eqref{eq:sandwich}.
\end{IEEEproof}

\section{Proof of \thmref{thm:two-sided-overshoot-fixed}} \label{app:twosided}

We use the notation from the theorem statement throughout. Define the one-sided passage times
\begin{align}
\tau_A^+ &\triangleq \inf\Big\{ n\ge 1 : S_n > A \Big\}, \\
\tau_B^- &\triangleq \inf\Big\{ n\ge 1 : S_n < -B \Big\}.
\end{align}
On $\mathcal{E}_+$, we have $\tau = \tau_A^+$ and $D_+ = S_\tau - A$; on $\mathcal{E}_-$, we have $\tau = \tau_B^-$ and $D_- = -(S_\tau + B)$.

\subsection*{Part~(i): Positive-side overshoot}

\paragraph{Ratio bound \eqref{eq:pos-two-sided-fixed}}
Define the one-sided upward overshoot $R_A \triangleq S_{\tau_A^+} - A$.
On $\mathcal{E}_+$, $D_+ = R_A$; on $\mathcal{E}_-$, $D_+ \, 1\{\mathcal{E}_+\} = 0 \leq R_A$. Hence,
\begin{align}
    D_+ \, 1\{\mathcal{E}_+\} \leq R_A. \label{eq:Dp_dominate}
\end{align}
Taking expectations and dividing by $p_+ = \Prob{\mathcal{E}_+}$ yields
\begin{align}
    \E{D_+ \mid \mathcal{E}_+} = \frac{\E{D_+ \, 1\{\mathcal{E}_+\}}}{p_+} \leq \frac{\E{R_A}}{p_+} \leq \frac{\eta(P_{X_1})}{p_+},
\end{align}
where the last step applies \thmref{thm:lordenBound} to the one-sided overshoot $R_A$.

\paragraph{Mean-excess bound \eqref{eq:pos-two-sided-me}}
On the event $\{\tau = k\} \cap \mathcal{E}_+$, we have $S_{k-1} \leq A$ (since the walk has not yet exited) and $S_k > A$. Define the $\mathcal{F}_{k-1}$-measurable threshold
\begin{align}
    y_{k-1} \triangleq A - S_{k-1} \geq 0.
\end{align}
Then $D_+ = X_k - y_{k-1}$ with $X_k > y_{k-1}$. Since the increments $\{X_k\}$ are i.i.d.\ and independent of $\mathcal{F}_{k-1}$, conditioning on $\mathcal{F}_{k-1}$ and the event $\{X_k > y_{k-1}\}$ gives
\begin{align}
    \E{D_+ \mid \tau = k, \mathcal{E}_+, \mathcal{F}_{k-1}} &= \E{X_1 - y_{k-1} \mid X_1 > y_{k-1}} \\
    &\leq \eta_0^{+}(P_{X_1}), \label{eq:Dp_me_cond}
\end{align}
where the inequality holds for every realization $y_{k-1} \geq 0$ by the definition \eqref{eq:eta0pos}. Taking iterated expectations yields
\begin{align}
    \E{D_+ \mid \mathcal{E}_+} \leq \eta_0^{+}(P_{X_1}).
\end{align}

\subsection*{Part~(ii): Negative-side overshoot}

\paragraph{Mean-excess bound \eqref{eq:neg-two-sided-me}}
By the same argument as the positive-side mean-excess bound, with $z_{k-1} \triangleq B + S_{k-1} \geq 0$ playing the role of $y_{k-1}$ and $D_- = -(X_k + z_{k-1})$, we obtain $\E{D_- \mid \mathcal{E}_-} \leq \eta_0^{-}(P_{X_1})$.

\paragraph{Ratio bound \eqref{eq:neg-two-sided-fixed}}
The ratio bound proceeds through a change-of-measure argument under the tilted probability $\bar{\mathbb{P}}$ defined in the theorem statement.

\emph{Step~1 (likelihood-ratio identity at $\tau$).}
Define $L_n \triangleq e^{-\lambda S_n}$.
Since $\E{e^{-\lambda X_1}} = 1$, the process $\{L_n\}_{n\ge 0}$ is a nonnegative mean-one
$\mathbb{P}$-martingale. For any nonnegative $\mathcal{F}_\tau$-measurable
random variable~$H$, the likelihood-ratio identity at the stopping
time~$\tau$ gives \cite[Th.~III.1.3]{asmussenalbrecher2010}
\begin{align}
\mathbb{E}\!\left[H\right]
&=
\bar{\mathbb{E}}\!\left[H\,e^{\lambda S_\tau}\right] .
\label{eq:lr-stopping-fixed}
\end{align}
On $\mathcal{E}_-$, $S_\tau = -B - D_-$, so
$e^{\lambda S_\tau} = e^{-\lambda B}\,e^{-\lambda D_-}$.
Applying~\eqref{eq:lr-stopping-fixed} with $H=1\{\mathcal{E}_-\}$ and with $H=D_-\,1\{\mathcal{E}_-\}$, then dividing, yields
\begin{align}
\mathbb{E}\!\left[D_- \mid \mathcal{E}_-\right]
&=
\frac{
\bar{\mathbb{E}}\!\left[
D_- e^{-\lambda D_-}\mid \mathcal{E}_-
\right]
}{
\bar{\mathbb{E}}\!\left[
e^{-\lambda D_-}\mid \mathcal{E}_-
\right]
} .
\label{eq:ratio-minus-fixed}
\end{align}

\emph{Step~2 (Chebyshev's other inequality).}
By Chebyshev's other inequality \cite{chebyshev} applied to the
conditional law of $D_-$ given $\mathcal{E}_-$ under $\bar{\mathbb{P}}$,
\begin{align}
\bar{\mathbb{E}}\!\left[
D_- e^{-\lambda D_-}\mid \mathcal{E}_-
\right]
&\le
\bar{\mathbb{E}}\!\left[D_- \mid \mathcal{E}_-\right]\,
\bar{\mathbb{E}}\!\left[
e^{-\lambda D_-}\mid \mathcal{E}_-
\right] .
\end{align}
Combining with \eqref{eq:ratio-minus-fixed} yields
\begin{align}
\mathbb{E}\!\left[D_- \mid \mathcal{E}_-\right]
&\le
\bar{\mathbb{E}}\!\left[D_- \mid \mathcal{E}_-\right]
=
\frac{\bar{\mathbb{E}}\!\left[D_- \,1\{\mathcal{E}_-\}\right]}
{\bar{p}_-} .
\label{eq:cond-to-uncond-minus-fixed}
\end{align}

\emph{Step~3 (domination by a one-sided lower overshoot under $\bar{\mathbb{P}}$).}
By the same domination argument as Part~(i), with $R_B^- \triangleq -(S_{\tau_B^-}+B)$,
\begin{align}
\bar{\mathbb{E}}\!\left[D_- \,1\{\mathcal{E}_-\}\right]
&\le
\bar{\mathbb{E}}\!\left[R_B^-\right] .
\label{eq:dominate-minus-fixed}
\end{align}

\emph{Step~4 (apply \thmref{thm:lordenBound} to the reflected walk under $\bar{\mathbb{P}}$).}
Under $\bar{\mathbb{P}}$, the increments $\{X_k\}$ are i.i.d.\ with the distribution of $\bar{X}_1$. Let $Y_k \triangleq -X_k$, so the $Y$-increments are i.i.d.\ with distribution $P_{-\bar{X}_1}$ under $\bar{\mathbb{P}}$. The mean of $Y_1$ under $\bar{\mathbb{P}}$ satisfies
\begin{align}
\bar{\mathbb{E}}\!\left[Y_1\right]
&=
-\bar{\mathbb{E}}\!\left[X_1\right]
=
-\mathbb{E}\!\left[X_1 e^{-\lambda X_1}\right]
> 0. \label{eq:ineq00}
\end{align}
To verify the strict inequality in~\eqref{eq:ineq00}, note that $m(\theta)\triangleq\mathbb{E}[e^{\theta X_1}]$ is convex with $m(0) = m(-\lambda) = 1$, so $m'(-\lambda) \le 0$; equality would force $X_1=0$ a.s., contradicting $\mathbb{E}[X_1]>0$.

Moreover, $R_B^-$ is the overshoot of the $Y$-walk $\sum_{k=1}^n Y_k$ past level~$B$.
Applying \thmref{thm:lordenBound} to this walk under $\bar{\mathbb{P}}$ gives
\begin{align}
\bar{\mathbb{E}}\!\left[R_B^-\right]
&\le \eta(P_{-\bar{X}_1}).
\end{align}
Combining with \eqref{eq:cond-to-uncond-minus-fixed}
and \eqref{eq:dominate-minus-fixed}
yields~\eqref{eq:neg-two-sided-fixed}.

\section{Integrability of mixed log-ratio functionals} \label{app:integrable}
\begin{lemma}
\label{lem:mixed_logratio_integrable}
Let $P$ and $Q$ be probability laws on a common measurable space.
Assume $P\ll Q$ and $Q\ll P$.
Let $X\sim P$ and $Y\sim Q$.
Assume
\begin{align}
\log\frac{P(X)}{Q(X)}
\ \stackrel{d}{=}\
\log\frac{Q(Y)}{P(Y)},
\label{eq:sym_llr}
\end{align}
and
\begin{align}
\E{\left|\log\frac{P(X)}{Q(X)}\right|}<\infty.
\label{eq:llr_int}
\end{align}
For $a\in[1/2,1]$, define
\begin{align}
Z_a(X)
&\triangleq
\log\frac{P(X)}{\frac{1}{2}P(X)+\frac{1}{2}Q(X)} \notag \\
&-\frac{Q(X)}{P(X)}\,
\log\frac{\frac{1}{2}P(X)+\frac{1}{2}Q(X)}
{aP(X)+(1-a)Q(X)},
\label{eq:Za_def_align}
\end{align}
and
\begin{align}
W_a(X)
\triangleq
\log\frac{P(X)}{(1-a)P(X)+aQ(X)}.
\label{eq:Wa_def}
\end{align}
Then, for every $a\in[1/2,1]$,
\begin{align}
\E{|Z_a(X)|}<\infty,
\qquad
\E{|W_a(X)|}<\infty.
\end{align}
\end{lemma}

\begin{IEEEproof}
Since $P\ll Q$ and $Q\ll P$, the likelihood ratio
$T \triangleq Q(X)/P(X) \in (0,\infty)$ a.s.
Note that
\begin{align}
    \mathbb{E}[T] = \int \frac{Q(x)}{P(x)} P(x)\,dx = 1, \label{eq:ET1}
\end{align}
and $|\log T| = |\log(P(X)/Q(X))|$, so $\mathbb{E}[|\log T|] < \infty$ by \eqref{eq:llr_int}.

\emph{Integrability of $W_a(X)$.}
From \eqref{eq:Wa_def}, $W_a(X) = -\log\!\big((1{-}a) + a T\big)$.
For $a = 1$, this reduces to $\log(P(X)/Q(X))$, which is integrable by \eqref{eq:llr_int}.
For $a \in [1/2, 1)$, bounding by cases $T \leq 1$ and $T > 1$ gives $|W_a(X)| \leq -\log(1{-}a) + \log(2a) + |\log T|$, so $\mathbb{E}[|W_a(X)|] < \infty$.

\emph{Integrability of $Z_a(X)$.}
Substituting $T = Q(X)/P(X)$ into \eqref{eq:Za_def_align} yields
\begin{align}
    Z_a(X) = \log\frac{2}{1+T} - T\,g_a(T), \label{eq:Za_T}
\end{align}
where $g_a(t) \triangleq \log\frac{(1+t)/2}{a + (1{-}a)t}$.
We bound the two terms separately.

\emph{First term.}
Since $\frac{2}{1+t} \in (0,2)$ for $t > 0$, we have
$\big|\log\frac{2}{1+t}\big| \leq \log 2 + |\log t|$,
so $\mathbb{E}\big[\big|\log\frac{2}{1+T}\big|\big] \leq \log 2 + \mathbb{E}[|\log T|] < \infty$.

\emph{Second term, $a \in [1/2, 1)$.}
A direct calculation gives $g_a'(t) = \frac{2a-1}{(1+t)(a+(1{-}a)t)} \geq 0$, so $g_a$ is nondecreasing on $(0, \infty)$ with
$\lim_{t \downarrow 0} g_a(t) = -\log(2a)$ and
$\lim_{t \to \infty} g_a(t) = -\log(2(1{-}a))$.
Hence $|g_a(t)| \leq C_a$ for all $t > 0$, where $C_a \triangleq \max\{|\log(2a)|,\, |\log(2(1{-}a))|\}$.
Therefore, $\mathbb{E}[T\,|g_a(T)|] \leq C_a\,\mathbb{E}[T] = C_a < \infty$ by \eqref{eq:ET1}.

\emph{Second term, $a = 1$.}
Here $g_1(t) = \log\frac{1+t}{2}$. Since $|\log\frac{1+t}{2}| \leq \log 2 + (\log t)^+$, we have $\mathbb{E}[T\,|g_1(T)|] \leq \log 2 + \mathbb{E}[T(\log T)^+]$.
By the change-of-measure identity $\mathbb{E}[T\, h(X)] = \mathbb{E}_Q[h(Y)]$ and the symmetry assumption \eqref{eq:sym_llr}, $\mathbb{E}[T(\log T)^+] \leq \mathbb{E}_Q[|\log\frac{Q(Y)}{P(Y)}|] = \mathbb{E}[|\log\frac{P(X)}{Q(X)}|] < \infty$ by \eqref{eq:llr_int}.

Combining the bounds on both terms of \eqref{eq:Za_T} gives $\mathbb{E}[|Z_a(X)|] < \infty$ for all $a \in [1/2, 1]$.
\end{IEEEproof}

\section{Proof of \lemref{lem:ruin-bhattacharyya}} \label{app:bhatta}
We write $Z \triangleq Z(P_{Y|X})$ for brevity throughout.

\paragraph{Step~1 (half-Bhattacharyya inequality)}
On $\{\Lambda(y) < 0\}$, $P_+(y) < P_-(y)$, so $P_+(y) \leq \sqrt{P_+(y) P_-(y)}$; integrating gives $\Prob{\Lambda < 0} \leq \int_{\Lambda(y) < 0} \sqrt{P_+ P_-}\, dy$. By the BMS symmetry \eqref{eq:PYXcond1}, $\Prob{\Lambda < 0} = \int_{\Lambda(y) > 0} P_-(y)\, dy \leq \int_{\Lambda(y) > 0} \sqrt{P_+ P_-}\, dy$. Adding the two bounds and using \eqref{eq:Zbhatt} yields
\begin{align}
    \Prob{\Lambda < 0} \leq \frac{Z}{2}. \label{eq:half_bhatt}
\end{align}

\paragraph{Step~2 (product-channel extension)}
The $n$-fold product channel is BMS with Bhattacharyya parameter $Z^n$, so \eqref{eq:half_bhatt} gives $\Prob{S_n < 0} \leq Z^n/2$.  Summing, $\sum_{n=1}^{\infty} \frac{1}{n}\, \Prob{S_n < 0} \leq \frac{1}{2}\log \frac{1}{1-Z}$, which is \eqref{eq:series-bhatt}. Substituting into Spitzer's formula \eqref{eq:spitzer-ruin} gives
\begin{align}
    \psi(0) &= 1 - \exp\!\left(-\sum_{n=1}^{\infty}\frac{1}{n}\,\Prob{S_n<0}\right) \notag \\
    &\leq 1 - \exp\!\left(-\frac{1}{2}\log\frac{1}{1-Z}\right) \notag \\
    &= 1 - (1-Z)^{1/2} = \chi(P_{Y|X}),
\end{align}
which is \eqref{eq:ruin-bhatt}.
\section{Proof of \lemref{lem:exit_ratio}} \label{app:exit_ratio}

Throughout this proof, we consider the i.i.d.\ LLR random walk $S_n = \sum_{k=1}^n \Lambda_k$ with $\Lambda_k \sim P_\Lambda$ under $\mathbb{P}$, and the two-sided stopping time $\tau = \inf\{n \geq 1 \colon u + S_n \notin [0, A]\}$. We write $\mc{E}_- = \{u + S_\tau < 0\}$ for the lower-exit event and $D_- = -(u + S_\tau) > 0$ for the overshoot on $\mc{E}_-$.

\paragraph{Reflection identity}
By \remref{rem:tiltedLLR}, $\bar{\Lambda} \overset{d}{=} -\Lambda$ under $\mathbb{P}$, so the tilted walk $u + \bar{S}_n$ under $\bar{\mathbb{P}}$ has the same law as $u - S_n$ under $\mathbb{P}$. Since $u - S_n = A - ((A - u) + S_n)$, it follows that
\begin{align}
    \bar{p}_{\mr{ruin}}(u, A) = 1 - p_{\mr{ruin}}(A - u, A), \label{eq:pruin_reflection}
\end{align}
and hence
\begin{align}
    f_A(u) = \frac{p_{\mr{ruin}}(u, A)}{1 - p_{\mr{ruin}}(A - u, A)}. \label{eq:fA_representation}
\end{align}

\paragraph{Part~(i): universal bound}
By \remref{rem:tiltedLLR}, $\mathbb{E}[e^{-\Lambda}] = 1$, so $\{e^{-S_n}\}_{n \geq 0}$ is a $\mathbb{P}$-martingale. Applying the likelihood-ratio identity at the stopping time $\tau$ yields
\begin{align}
    \bar{p}_{\mr{ruin}}(u, A) = \mathbb{E}\!\left[e^{-S_\tau} \, 1\{\mc{E}_-\}\right]. \label{eq:lr_identity_exit}
\end{align}
On $\mc{E}_-$, we have $S_\tau = -u - D_-$, hence $e^{-S_\tau} = e^{u + D_-}$. Substituting into \eqref{eq:lr_identity_exit} and conditioning on $\mc{E}_-$ gives
\begin{align}
    \bar{p}_{\mr{ruin}}(u, A) = e^{u} \cdot p_{\mr{ruin}}(u, A) \cdot \mathbb{E}\!\left[e^{D_-} \mid \mc{E}_-\right]. \label{eq:pruin_bar_exact}
\end{align}
Since $D_- > 0$ on $\mc{E}_-$, we have $\mathbb{E}[e^{D_-} \mid \mc{E}_-] \geq 1$. Dividing both sides of \eqref{eq:pruin_bar_exact} by $\bar{p}_{\mr{ruin}}(u, A)$ yields
\begin{align}
    f_A(u) = \frac{e^{-u}}{\mathbb{E}\!\left[e^{D_-} \mid \mc{E}_-\right]} \leq e^{-u},
\end{align}
which is \eqref{eq:fA_universal}.

\paragraph{Part~(ii): half-splitting bound}
We split $[0, A]$ into two halves and apply a different bound on each. Recall that $\psi(0) = p_{\mr{ruin}}(0, \infty)$ is the one-sided ruin probability at the origin, given by Spitzer's formula \eqref{eq:spitzer-ruin}.

\emph{Case $u \in [0, A/2]$.}
For the numerator in \eqref{eq:fA_representation}, \thmref{thm:spitzer-lundberg-sandwich}(ii) gives
\begin{align}
    p_{\mr{ruin}}(u, A) \leq \psi(u) \leq \psi(0). \label{eq:numerator_bound}
\end{align}
For the denominator, applying \eqref{eq:sandwich} and the Lundberg bound \eqref{eq:lundberg} with adjustment coefficient $R = 1$ (by \remref{rem:tiltedLLR}, $\lambda = 1$ for the BMS LLR walk), we obtain
\begin{align}
    p_{\mr{ruin}}(A - u, A) \leq \psi(A - u) \leq e^{-(A-u)} \leq e^{-A/2}, \label{eq:denominator_bound}
\end{align}
where the last step uses $A - u \geq A/2$. Substituting \eqref{eq:numerator_bound} and \eqref{eq:denominator_bound} into \eqref{eq:fA_representation} gives
\begin{align}
    f_A(u) \leq \frac{\psi(0)}{1 - e^{-A/2}}. \label{eq:fA_left_half}
\end{align}

\emph{Case $u \in [A/2, A]$.}
Applying part~(i) directly,
\begin{align}
    f_A(u) \leq e^{-u} \leq e^{-A/2}. \label{eq:fA_right_half}
\end{align}

Taking the larger of \eqref{eq:fA_left_half} and \eqref{eq:fA_right_half} yields \eqref{eq:fA_halfsplit}. Substituting $\psi(0) \leq \chi(P_{Y|X})$ from \lemref{lem:ruin-bhattacharyya} into \eqref{eq:fA_left_half} gives \eqref{eq:fA_halfsplit_chi}.

\section{Proof of \thmref{thm:caploss_expansion}} \label{app:caploss_proof}
Write $C - C_{B,\delta} = S_{\mr{int}} + S_{\mr{tail}}$, where
\begin{align}
    S_{\mr{int}} &\triangleq \sum_{k \in \mc{K}_B^{\circ}} \mu_+(I_k) \notag \\
    &\quad \cdot \Big( \mathbb{E}\!\big[g(\Lambda) \mid \Lambda \in I_k\big] - g(k\delta) \Big), \label{eq:Sint_def} \\
    S_{\mr{tail}} &\triangleq \sum_{k \in \{-L_{\mr{tail}},\, L_{\mr{tail}}\}} \mu_+(I_k) \notag \\
    &\quad \cdot \Big( \mathbb{E}\!\big[g(\Lambda) \mid \Lambda \in I_k\big] - g(k\delta) \Big). \label{eq:Stail_def}
\end{align}

\emph{Step~1 (tail bound).}
By \eqref{eq:cell_loss_psi} and the concavity of $\beta(u) = \log(1+u)$, each tail cell loss is non-negative and bounded above by $\beta(e^{-k\delta})$. Using the BMS antisymmetry $\mu_+(I_{-L_{\mr{tail}}}) = e^{-L_{\mr{tail}}\delta}\,\mu_+(I_{L_{\mr{tail}}})$, the identity $\beta(e^t) = t + \beta(e^{-t})$ for $t \geq 0$, and the bound $\Psi(a_L) = \mu_+(I_{L_{\mr{tail}}})(1 + e^{-L_{\mr{tail}}\delta})$, one obtains $S_{\mr{tail}} \leq \Psi(a_L) \cdot F(L_{\mr{tail}}\delta)$, where $F(t) \triangleq \beta(e^{-t}) + \frac{t\,e^{-t}}{1+e^{-t}}$ is decreasing with $F(0) = \log 2$. Hence
\begin{align}
    0 \leq S_{\mr{tail}} \leq (\log 2)\,\Psi(a_L) = o(\delta^2). \label{eq:tail_bound}
\end{align}

\emph{Step~2 (interior cell loss via Taylor expansion).}
Fix $k \in \mc{K}_B^{\circ}$.  By \lemref{lem:condexp_char}, $\mathbb{E}[e^{-\xi_k} \mid \Lambda \in I_k] = 1$.  Expanding $e^{-u} = 1 - u + u^2/2 + O(u^3)$ and using $|\xi_k| = O(\delta)$ from \ref{cond:R1} gives
\begin{align}
    \mathbb{E}[\xi_k \mid \Lambda \in I_k] = \tfrac{1}{2}\,\mathbb{E}[\xi_k^2 \mid \Lambda \in I_k] + O(\delta^3). \label{eq:step2}
\end{align}
Similarly, expanding $g(k\delta + \xi_k)$ to second order via Taylor's theorem, substituting \eqref{eq:step2}, and using the identity $g'(k\delta) + g''(k\delta) = h(k\delta)$ from \eqref{eq:gpp_identity} yields
\begin{align}
    \mathbb{E}\!\big[g(\Lambda) \mid \Lambda \in I_k\big] - g(k\delta) = \tfrac{1}{2}\,h(k\delta)\,\mathbb{E}[\xi_k^2 \mid \Lambda \in I_k] + O(\delta^3), \label{eq:step3_main}
\end{align}
where the $O(\delta^3)$ error uses $g' \leq 1$, $|g'''| \leq \frac{\sqrt{3}}{18}$ from \lemref{lem:g_derivatives}, and \ref{cond:R3}.  Applying \ref{cond:R2} to replace $\mathbb{E}[\xi_k^2 \mid \Lambda \in I_k]$ by $\delta^2/12$, multiplying by $\mu_+(I_k)$, and summing over $k \in \mc{K}_B^{\circ}$:
\begin{align}
    S_{\mr{int}} = \frac{\delta^2}{24} \sum_{k \in \mc{K}_B^{\circ}} \mu_+(I_k)\,h(k\delta) + o(\delta^2). \label{eq:step4}
\end{align}

\emph{Step~3 (Riemann sum approximation).}
By the Lipschitz bound $\sup|h'| = \frac{8}{27}$, the localization $|Q_B(\Lambda) - \Lambda| \leq \Gamma\delta$ on $\{|\Lambda| < a_L\}$ (where $Q_B$ rounds $\Lambda$ to the nearest lattice point), and $0 \leq h \leq 1$:
\begin{align}
    \left| \sum_{k \in \mc{K}_B^{\circ}} \mu_+(I_k)\,h(k\delta) - \kappa_{P_{Y|X}} \right| \leq \frac{8\Gamma}{27}\,\delta + \Psi(a_L) = o(1). \label{eq:step5}
\end{align}
Combining \eqref{eq:tail_bound}, \eqref{eq:step4}, and \eqref{eq:step5} gives $C - C_{B,\delta} = \frac{\kappa_{P_{Y|X}}}{24}\,\delta^2 + o(\delta^2)$.

\section{Proof of \lemref{prop:threshold_geometry}} \label{app:proofprop}
We first establish a local expansion for the interval LLR.  Let $I = [u, v)$ with midpoint $m = \frac{u+v}{2}$, width $w = v - u$, $|m| \leq T_B + 1$, and $0 < w \leq 2\delta_B$.  Centering via $\Xi \triangleq \Lambda - m$ on $I$ and writing
\begin{align}
    \ell(u, v) = m - \log \mathbb{E}[e^{-\Xi}] \label{eq:local_interval_centering}
\end{align}
yields
\begin{align}
    \ell(u, v) = m + O\!\big((1 + T_B)\,w^2\big), \label{eq:local_interval_expansion}
\end{align}
uniformly over such intervals.  To see this, note that the conditional density of $\Lambda$ given $\Lambda \in I$ is proportional to the Gaussian density $f_+(\lambda)$ of $\Lambda$ under~$P_+$.  Since
\begin{align}
    \sup_{|\lambda| \leq T_B + 1} \left| \frac{d}{d\lambda} \log f_+(\lambda) \right| = O(1 + T_B), \label{eq:log_density_bound}
\end{align}
the conditional law of $\Xi$ is a multiplicative $1 + O((1 + T_B)\,w)$ perturbation of the uniform law on $[-w/2,\, w/2]$.  Consequently, $\mathbb{E}[\Xi] = O((1 + T_B)\,w^2)$ and $\mathbb{E}[\Xi^2] = O(w^2)$.  Since $|\Xi| \leq w/2$, expanding $e^{-\Xi} = 1 - \Xi + \frac{\Xi^2}{2} + O(w^3)$ and taking expectation gives $\mathbb{E}[e^{-\Xi}] = 1 + O((1 + T_B)\,w^2)$, from which \eqref{eq:local_interval_expansion} follows via \eqref{eq:local_interval_centering}.

We prove \eqref{eq:ak_expansion} by backward induction.  Write $a_k = (k - \frac{1}{2})\,\delta_B + e_k$.  The base case $e_L = 0$ holds since $a_L = T_B = (L - \frac{1}{2})\,\delta_B$.  Suppose $|e_{k+1}| \leq C_0\,(1 + T_B)\,\delta_B^2$ for a constant $C_0$ to be chosen.  Define $\tilde{a}_k \triangleq (k - \frac{1}{2})\,\delta_B - e_{k+1}$, so that $[\tilde{a}_k,\, a_{k+1})$ has midpoint $k\,\delta_B$ and width $\delta_B + O((1 + T_B)\,\delta_B^2)$.  By \eqref{eq:local_interval_expansion}, $\ell(\tilde{a}_k,\, a_{k+1}) = k\,\delta_B + O((1 + T_B)\,\delta_B^2)$.  Considering the bracketing values $a_k^{\pm} \triangleq \tilde{a}_k \pm C_0\,(1 + T_B)\,\delta_B^2$ and applying \eqref{eq:local_interval_expansion} again shows that for $C_0$ sufficiently large, $\ell(a_k^-,\, a_{k+1}) < k\,\delta_B < \ell(a_k^+,\, a_{k+1})$.  Since $u \mapsto \ell(u,\, a_{k+1})$ is continuous and strictly increasing, the intermediate value theorem gives $|e_k| \leq C_0\,(1 + T_B)\,\delta_B^2$, closing the induction and proving \eqref{eq:ak_expansion}.  The expansions \eqref{eq:mk_expansion} and \eqref{eq:wk_expansion} follow immediately.

\bibliographystyle{IEEEtran}
\bibliography{mac} 

\vspace{-2em}

\begin{IEEEbiographynophoto}{Recep Can Yavas}
(Member, IEEE) received the B.S. degree (Hons.) in electrical engineering from Bilkent University, Ankara, Turkey, in 2016. He received the M.S. and Ph.D. degrees in electrical engineering from California Institute of Technology (Caltech) in 2017 and 2023, respectively. He was a research fellow at CNRS@CREATE, Singapore, between 2022 and 2024. He was with the Department of Computer Science at National University of Singapore as a research fellow. Since January 2026, he has been an Assistant Professor with the Department of Electrical and Electronics Engineering at Bilkent University, Ankara, Turkey. 
His research interests include information theory, communications, and multi-armed bandits.
\end{IEEEbiographynophoto}

\end{document}